%% file: main.tex
\RequirePackage{tikz}
\documentclass{article}
\usepackage{a4wide} %

\usepackage[utf8]{inputenc}

\usepackage{hyperref}

\usepackage{amsmath}
\usepackage{amssymb}
\usepackage{amsthm}

\usepackage{caption}
\usepackage{subcaption}
\usepackage[numbers]{natbib}

\usepackage{amsthm}
\usepackage{amsmath}
\usepackage{amssymb}
\usepackage{mathtools}
\usepackage{dsfont}
\usepackage{url}  
\usepackage{blindtext}

\DeclareMathOperator*{\argmax}{arg\,max}

\usepackage{mathrsfs}
\usepackage{graphicx}
\usepackage{xspace}
\usepackage{bbm}
\usepackage{tikz-cd}
\usepackage{tikz}
\usetikzlibrary {calc,positioning,shapes.misc}

\usepackage{centernot}

\usepackage{algorithm}
\usepackage{algorithmic}
\usepackage[algo2e, noend, noline, ruled, linesnumbered]{algorithm2e}

\usepackage{xcolor}

\usepackage{nicefrac}

\DeclareMathOperator*{\extregret}{ExtReg}

\DeclareMathOperator*{\extreg}{ExtReg}

\DeclareMathOperator*{\swapreg}{SwapReg}

\theoremstyle{plain}%
\newtheorem{theorem}{Theorem}%
\newtheorem{lemma}[theorem]{Lemma}
\newtheorem*{unnum_theorem}{Theorem}
\newtheorem{proposition}[theorem]{Proposition}%

\theoremstyle{plain}%
\newtheorem{example}{Example}%
\newtheorem{remark}{Remark}%

\theoremstyle{definition}%
\newtheorem{definition}{Definition}%

\providecommand{\keywords}[1]
{
  \small	
  \textbf{\textit{Keywords---}} #1
}

\newcommand{\footremember}[2]{%
    \footnote{#2}
    \newcounter{#1}
    \setcounter{#1}{\value{footnote}}%
}
\newcommand{\footrecall}[1]{%
    \footnotemark[\value{#1}]%
} 

\title{%
Learning Correlated Equilibria in Mean-Field Games}
\author{Paul Muller\footremember{dm}{DeepMind Paris and London. Corresponding author: pmuller@deepmind.com} \and Romuald Elie\footrecall{dm} \and Mark Rowland\footrecall{dm} \and Mathieu Lauriere\footremember{gbrain}{Google Research, Brain Team}\footremember{nyush}{NYU Shanghai} \and Julien Perolat\footrecall{dm} \and Sarah Perrin\footnote{Inria Scool} \and Matthieu Geist\footrecall{gbrain} \and Georgios Piliouras\footrecall{dm} \and Olivier Pietquin\footrecall{gbrain} \and Karl Tuyls\footrecall{dm}}

\date{\today}

\begin{document}

\maketitle

\abstract{
The designs of many large-scale systems today, from traffic routing environments to smart grids, rely on game-theoretic equilibrium concepts. However, as the size of an $N$-player game typically grows exponentially with $N$,  
 standard game theoretic analysis becomes effectively infeasible
beyond a low number of players. Recent approaches have gone around this limitation by instead considering Mean-Field games, an approximation of anonymous $N$-player games, where the number of players is infinite and the population's state distribution, instead of every individual player's state, is the object of interest. The practical computability of Mean-Field Nash equilibria, the most studied Mean-Field equilibrium to date, 
however, typically depends on beneficial non-generic structural properties %
 such as monotonicity \cite{lasry2007mean} or contraction properties \cite{guo2019learningmfg}, which are required for known algorithms to converge.
In this work, we provide an alternative route for studying Mean-Field games, by developing the concepts of Mean-Field correlated and coarse-correlated equilibria. %
We show that they can be efficiently learnt in \emph{all games}, without requiring any additional assumption on the structure of the game, using three classical algorithms. Furthermore, we establish correspondences between our notions and those already present in the literature, derive optimality bounds for the Mean-Field - $N$-player transition, and empirically demonstrate the convergence of these algorithms on simple games.}


\setcounter{tocdepth}{2}
\tableofcontents

\newpage
\section{Introduction}\label{sec:introduction}

The complexity of describing and computing equilibria in games with a finite number of players grows exponentially as the size of the population increases\footnote{To see this, picture a M-action $N$-player game. The payoff tensor of a such game is of size $N \, M^N$, a quantity exponential in $N$. Assuming that this game is such that its Nash equilibrium is fully mixed, thus computing the Nash equilibrium will require going through every payoff tensor cell at least once, hence leading to an at least exponential relationship between equilibrium computation time and number of players.}.
Such computations are however extremely useful in many different fields: traffic routing \cite{Elhenawy2015gttraffic, Wang2020traffic, khachoai2018traffic}, energy management \cite{Naz2019energy, Akash2021energy, APLAK2013energy, xu2021energy, Fetanat2021energy, maddouri2018energy, zhang2019energy}, mechanism-design \cite{Narahari2014mechanism, bergemann2012mechanism}, among many others. 
Computation is hampered by, among others, the need to consider every individual player's states and actions. 
The joint player state space complexity thus grows combinatorially, the difficulty being akin to producing an exact simulation of an $N$ particle system - easy for low $N$, impossible for high $N$. 
In such context, taking insight from statistical physics, we focus directly on the distribution of the population of particles instead of simulating every one of them, and thus consider \emph{Mean-Field games}. 

Mean-Field games (MFGs) have been introduced to simplify the analysis of Nash equilibria in games with a very large number of identical players interacting in a symmetric fashion (\textit{i.e.}, through the distribution of all the players). 
The key idea is to solely focus on the interactions between a representative infinitesimal player and a (so-called Mean-Field) term capturing the effect of the population of players. 
Understanding the behavior of one typical player is enough, as the behavior of the whole population can be deduced from it, since all players are assumed to be identical. %
This approach circumvents the difficulties induced by representing an extremely large population of agents. Since their introduction by Lasry and Lions~\cite{lasry2007mean}, and Caines, Huang and Malham{\'e}~\cite{MR2346927-HuangCainesMalhame-2006-closedLoop}, MFGs have been extensively studied both from a theoretical and a numerical viewpoint~\cite{Cardaliaguet-2013-notes,MR3134900,CarmonaDelarue_book_I,CarmonaDelarue_book_II,achdoulauriere2020meanfieldgamesnumsurvey}. 
Applications in various fields such as  energy management \cite{meriaux2012energy, matoussi:hal-01740707, djehiche2017engineering}, financial markets \cite{carmona2021finance, carmona2020finance, Gueant2011}, macroeconomics \cite{Gueant2011, gomes2015economics, Angiuli2021econ}, vehicle routing \cite{tanaka2018traffic, debbah4gnetworks, djehiche2017engineering}, mechanism design \cite{Klinger2021meanfieldconsensus, Krishnamurthy2014meanfieldauctions, doncel:hal-01277098} or epidemics dynamics \cite{lee2021covid, PETRAKOVA2022covid, bremaud2022covid, aurell2022covid} have already been considered. 
Most of the literature focuses on stochastic differential games and characterize their solution via the consideration of partial differential or stochastic differential equations~\cite{Cardaliaguet-2013-notes,MR3134900,CarmonaDelarue_book_I,CarmonaDelarue_book_II}. 
A forward equation captures the full population dynamics, while a backward one represents the evolution of the value function for a representative agent. 
With few exceptions, such as in~\cite{lackerleflem2021closed} which considers a class of closed-loop controls with a common signal or in~\cite{DeglInnocenti-phdthesis,campi2020correlated} which considers correlated equilibria as we explain below, only pure or mixed Nash equilibria have been considered so far. This is in stark contrast with the panoply of alternative notions of equilibria considered for games with a finite number of players \cite{BlumInternalExternalRegret, morrill2020hindsight, morrill2021efficient, Aumann1987CorrelatedEA, Neumann1928,Neumann1944, omidshafiei2019alpha, farina2022faster, farina2019coarse, piliouras2021phiregret}.
In the context of MFGs, mixed Nash equilibria with relaxed controls have been studied in~\cite{lacker2015meancontrolledmartingale,carmonadelaruelacker2016meancommon}. For example, mixed controls arise naturally in the context of MFG with optimal stopping where players should avoid simultaneous actions, as studied by Bertuci~\cite{bertucci2018optimal} or Bouveret et al.~\cite{bouveretdumitrescutankov2020mean}. Moreover, mixed policies are commonly considered in the setting of reinforcement learning for MFG, see for example~\cite{guo2019learningmfg,perrin2020fictitious,anahtarci2020q}. More generally the question of learning equilibria in MFGs has gained momentum in the past few years \cite{Hadikhanloo-phdthesis,cardaliaguet2017learning,perolat2021scaling, achdoulauriere2020meanfieldgamesnumsurvey, perrin2020fictitious, xinguo2020generalmfg}.

Studying and understanding learning behaviors in games has been a problem of fundamental importance within traditional game theory. Shortly after Von Neumann's seminal work on the existence and effective uniqueness of equilibria in zero-sum games via his minimax theorem~\cite{Neumann1928,Neumann1944}, Brown and Robinson~\cite{Brown1951,Robinson1951} developed the first learning procedures %
that converge successfully to equilibrium in zero-sum games in a time-average sense. Unfortunately, this initial glimmer of hope of general positive results connecting Nash equilibria and learning took a step backwards when Shapley \cite{shapley1964some} established that, even in the case of simple non-zero-sum games learning dynamics, one does not have to converge to Nash equilibria (even in a time-average sense).  %
This result was a strong precursor of the evolution of the field with many, increasingly strong, negative results establishing the lack of any meaningful correlation between Nash equilibria and learning dynamics~\cite{Gaunersdorfer,Jordan475,Sato02042002,daskalakis10,paperics11}. 

In the face of these persistent failures, a natural follow-up direction has been to pursue connections between the time-average of learning dynamics and other weaker game theoretic solutions concepts. The most well known approach of this type has focused on the tightly coupled notions of correlated equilibria (CE)~\cite{aumann1974subjectivity} and coarse correlated equilibria (CCE)~\cite{moulin1978strategically}. These solutions concepts are inspired by the possibility for a mediator to provide correlated advice to each player in regards to which action to pick from a joint distribution that is common knowledge to all players. 
Extending such concepts to MFGs somehow reduces the gap between Mean-Field Control, where a central coordinator imposes their will on decentralized controllers with no agency, and Mean-Field Games, where decentralized agents traditionally manifest their own will with no coordination mechanism. This bridge also entails the possibility of circumventing Price of Anarchy and Stability issues~\cite{AGTbook}, i.e., achieving performance guarantees better those possibly by Nash equilibria, which is an known issue in Mean Field Games \cite{cardaliaguet2019efficiency}, by introducing a way for agents to coordinate their actions. Besides, unlike Nash equilibria, these solution concepts enjoy an inextricable connection to a wide class of learning procedures known as no-regret or regret-minimizing dynamics~\cite{hazan2019introduction,hart2013simple,roughgarden2016twenty}. Specifically, all regret-minimizing dynamics converge in a time-average sense to coarse correlated equilibria and, vice versa, for any coarse correlated equilibrium in any game, there exists a tuple of regret-minimizing dynamics that converge to it~\cite{monnot2017limits}. Such notions of equilibria and related learning mechanisms have surprisingly been so far neglected in the context of Mean-Field games. Only DeglInnocenti~\cite{DeglInnocenti-phdthesis} as well as  Campi and Fischer~\cite{campi2020correlated} considered the notion of Mean Field correlated equilibria in both static and dynamic settings. They prove in particular, under suitable conditions and in the fully discrete (State, Action and Time) setting, that $N$-player CEs converge to Mean-Field CEs as N tends to infinity. 

In contrast, our paper presents another vision of Mean-Field correlated equilibria (and introduces coarse correlated ones), which we argue is closer to the one considered in the traditional game theory literature \cite{blumregret, Aumann1987CorrelatedEA, gordon2008regret}, as well as more intuitive and easier to manipulate. Yet, we are able to provide equivalence results between our definition and the one in~\cite{campi2020correlated} and focus our attention on relevant properties of these equilibria. In particular, we draw connections with no-regret learning in a mean-field setting and show that using a Mean Field Correlated Equilibrium policy in an $N$-player game generates a $O(1/\sqrt{N})$ approximate Correlated Equilibrium under suitable conditions. 
We study Correlated and Coarse Correlated Equilibria for a large class of Mean Field Games, both in the static and the evolutive settings. Importantly, this more flexible notion of equilibrium allows to capture the efficiency of learning mechanisms in Mean Field Games with several Nash equilibria. Building on the connection with no regret learning, we establish the convergence of classical learning algorithms for Mean Field Games to Coarse Correlated Equilibria in settings where no condition ensuring uniqueness of Nash (monotonicity, contraction property) is available. The three algorithms that we consider are Online Mirror Descent \cite{perolat2021scaling}, a variant of Fictitious Play \cite{cardaliaguet2017learning,perrin2020fictitious}, and Policy Space Response Oracle (PSRO) \cite{lanctot2017psro} (already introduced in \cite{muller2021learning} and reported here for the sake of completeness). %
We summarize the main contributions of the paper as follows: %
\begin{itemize}
    \item We provide the first formulation of coarse correlated equilibrium for Mean Field Games together with a more convenient one for correlated equilibria in this setting. Equivalence between our new formulation and the existing literature \cite{campi2020correlated} is provided.
    \item We explore properties of our new equilibrium notions and in particular demonstrate that using a Mean-Field (coarse) correlated equilibrium in $N$-player games provides an $O(1/\sqrt{N})$ approximate Nash equilibrium. 
    \item We introduce Mean-Field regret minimization and establish the connection between no-regret and (coarse) correlated equilibria properties. %
    \item For \emph{all games} in our framework (without requiring \emph{e.g.} monotonicity conditions), we show the convergence of several learning algorithms %
    towards approximate Mean-Field coarse-correlated equilibria. Numerical experiments illustrate this property for Online Mirror Descent \cite{perolat2021scaling}, a well-suited modification of Fictitious Play \cite{perrin2020fictitious} as well as Mean-Field PSRO \cite{muller2021learning}.
\end{itemize}

The paper is organized as follows. Section~\ref{sec:n_player_anonymous_games} revisits the notion of correlation device for Symmetric anonymous $N$-player games and paves the way to the intuitive notion of Mean Field (Coarse) Correlated Equilibrium presented in Section~\ref{sec:notions_mean_field_eq}. Section~\ref{sec:properties_mean_field_eq} links this more intuitive notion to the existing literature \cite{campi2020correlated} and derives some of its relevant  theoretical properties, such as existence conditions and special cases characterization. Section~\ref{sec:n_player_mf_eq} deals with the relationship between $N$-player games and Mean-Field games. It first establishes how to use a mean-field correlated equilibrium in an $N$-player game and then proves that sequences of $N$-player (coarse) correlated equilibria converge towards mean-field correlated equilibria with N, a property adapted from \citet{campi2020correlated}. Moreover, it provides optimality bounds for using Mean-Field (coarse) correlated equilibria in $N$-player games. The connections between Mean Field (Coarse) Correlated equilibria and corresponding regret minimization is discussed in Section~\ref{sec:regret_minimization_empirical_play}. Finally, Section \ref{sec:learning} establishes the convergence of learning algorithms to coarse correlated equilibria for general Mean Field Games, which is illustrated in Section~\ref{sec:exp_results}.

\section*{Notations}

We introduce here the main notations of the paper.

\noindent \textbf{Setting.} Given a finite set $\mathcal{Y}$, we denote by  $\Delta(\mathcal{Y})$ the set of distributions over $\mathcal{Y}$. To emphasize the difference between the finite and non-finite cases, if $\mathcal{Y}$ is not finite, we write $\mathcal{P}(\mathcal{Y})$ the set of distributions over $\mathcal{Y}$. A game - be it Mean-Field or $N$-player symmetric-anonymous - is a set $(\mathcal{X}, \mathcal{A}, r, P, \mu_0)$ where $\mathcal{X}$ is the finite set of states, $\mathcal{A}$ is the finite set of actions, $r: \mathcal{X} \times \mathcal{A} \times \Delta(\mathcal{X}) \rightarrow \mathbb{R}$ is a reward function, $p: \mathcal{X} \times \mathcal{A} \times \Delta(\mathcal{X}) \rightarrow \Delta(\mathcal{X})$ is a state transition function and $\mu_0 \in \Delta(\mathcal{X})$ is an initial state occupancy measure. The dependence of $r$ and $P$ on an element of $\Delta(\mathcal{X})$ captures the interaction between the players. It measures the influence of the full distribution of players over states on the reward and dynamics of each identical player. This assumption considers that all players are \emph{anonymous}, \emph{i.e.} only their state distribution affects others while their identity is irrelevant; and that the game is \emph{symmetric}, since all players share the same reward and dynamic functions. In an $N$-player game, we denote by $\mathcal{N}$ the set of players. Note that we only consider symmetric-anonymous $N$-player games, hence we deviate from the traditional vision of considering all players' individual states, to consider that the other players affect one another only through their empirical distribution. Finally, since we consider finite games, we name $\mathcal{T}$ the discrete set of times.\\

\noindent \textbf{Policy.} A policy is a mapping $\bar\pi: \mathcal{X} \rightarrow \Delta(\mathcal{A})$, where $\bar\pi(x, a)$ represents the probability of playing action $a$ while at state $x$. The set of such policies is denoted $\bar\Pi$. We also consider $\Pi$ the set of deterministic policies, \textit{i.e.} of the form: $\forall x \in \mathcal{X},\, \exists a \in \mathcal{A}, \, \pi(x) = \delta_a$, or $\pi(x, a') = \mathbf{1}_{\{a' = a\}}$.  The set $\Pi$ of deterministic policies is \emph{finite} and its convex hull is the set $\bar\Pi$ of all (stochastic) policies. In an $N$-player game, we write $\Pi_i$ the set of policies of player $i$. If the game is symmetric, we write $\Pi = \Pi_1 = ... = \Pi_N$ the common set of policies available to all players. \\ %

\noindent \textbf{Payoff.} We name $J_i$ the expected payoff function for player $i$ : $J(\pi_i, \pi_{-i})$ is the expected return for player $i$ when they play policy $\pi_i$ while the population of all other players play the joint policy $\pi_{-i}$. If $\mu_{i, t}$ is the distribution of player $i$ over states at time $t \in \mathcal{T}$, and $r^\pi_{i, t}$ is their expected reward vector (One component per state and per time, averaged over actions given their probability of occurrence following $\pi$) - both given the actions of other players, then
\[
    J(\pi_i, \pi_{-i}) \, = \; \sum_{t \in \mathcal{T}} \; \langle r^{\pi}_{i, t}, \, \mu_{i, t} \rangle \,.
\]
where the scalar product between two vectors $x, y \in \mathbb{R}^N$ is defined as $\langle x, y \rangle = \sum\limits_{i=1}^N x_i y_i $.

\noindent \textbf{Policy swap.} Finally, we define the set of policy swap functions
\begin{equation}\label{eq_def_UCE_UCCE}
\mathcal{U}_{CE} := \left\{ u : \Pi \rightarrow \Pi \right \}\,,\; \mbox{ and }\; \mathcal{U}_{CCE} := \left\{ u : \Pi \rightarrow \Pi \mid u \text{ constant } \right \}
\end{equation}
the set of unilateral deviation functions, i.e. the restriction of $\mathcal{U}_{CE}$ to constant functions. Intuitively, policy swaps which are defined over deterministic policies, are  functions that shift the probability mass assigned on one policy to another - thereby swapping policies around in a distribution of play. Note that swaps do not need to be bijective, and can for example always return the same policy.

\section{A Small Detour Through \texorpdfstring{$N$}{\textit{N}}-Player Games}\label{sec:n_player_anonymous_games}

By construction, Mean-Field Games identify  to the limit of symmetric-anonymous $N$-player  games, when N tends to infinity. Correlated and coarse correlated equilibria have been widely studied in games with finite number of players \cite{blumregret, BlumInternalExternalRegret, Aumann1987CorrelatedEA, marrismuller2021ce, farina2022faster, farina2019coarse, gordon2008regret}. Hence, we first ground our intuition and formalism by focusing on the particular case of symmetric-anonymous games with a finite number of players. We derive new expressions for correlated and coarse correlated equilibria for these games, paving the way to their straightforward extension in the Mean-Field setting in Section \ref{sec:notions_mean_field_eq}.

 Considering (coarse) correlated equilibria removes this issue by letting the correlation device choose which joint policy to recommend. Note that a correlation device which only recommends one Nash equilibrium is a correlated equilibrium.
 thus correlated equilibria may thus be used to solve the equilibrium selection problem.

\subsection{Notions and Intuitions of Equilibria in \texorpdfstring{$N$}{\textit{N}}-Player Games}

For sake of completeness, we first briefly recall the classical notions of Nash~\cite{nash1950theorem}, correlated and coarse correlated~\cite{aumann1974subjectivity} equilibria in $N$-player games. 

\begin{definition}[$N$-Player Nash Equilibrium]\label{Nplayer_Nash}
Given $\epsilon>0$, we define an $\epsilon$-Nash Equilibrium $(\pi_1,\ldots,\pi_n)\in \bar\Pi_1 \times \dots \times \bar\Pi_N$ as an $n$-tuple of strategies such that
\[
    \max_{\pi'_i\in\Pi_i} J_i(\pi'_i, \pi_{-i}) - J_i(\pi_i, \pi_{-i}) \leq  \epsilon\,, \qquad \forall i\in\mathcal{N}.
\]
\end{definition}

We will call a Nash \emph{pure} if ever it is deterministic. Otherwise, we will call it \emph{mixed}.

Contrarily to Nash Equilibria, where players choose separately which policy to follow, correlated and coarse correlated equilibria must be implemented with an additional entity atop the game whose only purpose is to coordinate agents' behaviors. It does so by selecting a joint strategy for the full population of players, and then recommends each player their policy within the joint strategy. Each player is aware of their own recommended policy together with the joint distribution over the population, but does not know the recommendation given to every other player: from a joint policy $\pi$, a player $i$ only sees $\pi_i$.

The goal of the additional entity, termed a \emph{correlation device}, is to render their recommendations stable in the presence of payoff maximizing player. That is, given a policy recommendation, and given knowledge of the probability distribution over the joint policies recommended by the correlation device, does the player have an incentive to deviate and play something else? If the answer is negative, the correlation device is a correlated equilibrium:

\begin{definition}[$N$-player $\epsilon$-Correlated Equilibrium]\label{def_CE_Gen}
Given $\epsilon>0$, we define an $\epsilon$-Correlated Equilibrium $\rho \in \Delta(\Pi_1 \times \dots \times \Pi_N)$ as a distribution over joint strategies such that %
\[
    \mathbb{E}_{\pi \sim \rho} \left[ J_i(u(\pi_i), \pi_{-i}) - J_i(\pi_i, \pi_{-i}) \right ] \leq \epsilon \;, \qquad  \forall u\in \mathcal{U}_{CE},\, i\in\mathcal{N}.
\]
\end{definition}

Another question, different from and less restrictive than the correlated equilibria's, concerns the player's ability to \emph{a priori} find a fixed deviating policy, independent of their received advice, so that they can improve their payoff without even taking their own recommendation into account. If this question's answer is negative, the correlation device is a coarse correlated equilibrium:

\begin{definition}[$N$-player $\epsilon$-Coarse Correlated Equilibrium]\label{def_CCE_Gen}
Given $\epsilon>0$, we define an $\epsilon$-Coarse Correlated Equilibrium $\rho \in \Delta(\Pi_1 \times \dots \times \Pi_N)$ as a distribution over joint strategies such that %
\[
    \mathbb{E}_{\pi \sim \rho} \left[ J_i(u(\pi_i), \pi_{-i}) - J_i(\pi_i, \pi_{-i}) \right ] \leq \epsilon \;, \qquad  \forall u\in \mathcal{U}_{CCE},\, i\in\mathcal{N}.
\]
\end{definition}

We see here that correlated and coarse correlated equilibria are very similar, only differing by the collection of admissible deviation types $\mathcal{U}_{CE}$ and $\mathcal{U}_{CCE}$  defined in \eqref{eq_def_UCE_UCCE}. In particular, any correlated equilibrium is obviously a coarse correlated equilibrium.\\

We note that the presence of the correlation device helps solving one issue which plagues Nash equilibria in $N$-player games: the equilibrium selection problem. Indeed, as mentioned above, Nash equilibria are characterized by all players acting in a payoff maximizing manner but without  coordination. When several Nash equilibria exist in a game, players must all somehow choose the same Nash equilibrium to receive any individual optimality guarantee.

This formulation is however too general to provide straightforward definitions of these equilibria for Mean-Field games: there is no direct, general way to define a joint strategy over an infinity of unique players, hence neither is there one for distributions over this space. %
However, Mean-Field games are 
 a particular class of infinite player games, i.e. infinite-player symmetric-anonymous games. In the following section, we provide an equivalent writing of (coarse) correlated equilibria in this setting which naturally scales to the Mean-Field limit.

\subsection{The Special Case of Symmetric-Anonymous \texorpdfstring{$N$}{\textit{N}}-Player Games}

We start this section with a remark: we will use interchangeably the terms \emph{symmetric-anonymous} and \emph{symmetric}, because all symmetric games are anonymous: indeed, take a symmetric game, a given player $i$, and the set of permutations $\sigma_i$ composed of all permutations which do not permutate $i$. Since the game is symmetric, $i$'s payoff remains identical whatever the permutation, hence $i$'s payoff is only affected by the number of players playing a given strategy, but not by their identity. Since this is true for all players, the game is anonymous. This result is also derived in~\cite{ham2013anonymity}, along with many other properties of symmetric and anonymous games.

In symmetric-anonymous games, on top of all individual policy sets $\Pi_i$ being identical and equal to $\Pi$, the payoff functions must not be impacted by player identities. Namely, all payoff functions $J_i$ are such that, for any permutation  $\tau : [1, N] \rightarrow [1, N]$, we have 
\[
    J_i(\pi_1, \dots, \pi_N) = J_{\tau^{-1}(i)}(\pi_{\tau(1)}, \dots, \pi_{\tau(N)}) \,, \qquad \pi = (\pi_1, \dots, \pi_N) \in \Pi^N.
\]
In other words, the reward for a given player $i$ only depends on player $i$'s own policy together with the distribution of policies over the population of all the other players, without any impact from each player identity. %
This rewrites analogously as follows: the payoff that player $i$ receives when playing $\pi_i$ only depends on the proportion of other players playing every policy in $\Pi$.%

We therefore introduce the following concept:

\begin{definition}[$N$-player Population Distribution]
    The \emph{Population Distribution} of N players playing policies in $\Pi$ is defined as $\nu_N = \frac{1}{N} \sum_\pi n_{\pi} \delta_{\pi}$, where $n_{\pi}$ is the number of players playing $\pi$, and $\delta_\pi$ is a dirac centered on $\pi$. The set of $N$-player population distributions is written $\Delta_N(\Pi)$.
\end{definition}

We will analogously denote $\nu_{-i} \in \Delta_{N-1}(\Pi)$ the distribution over policies in the population of all players except player $i$. By construction, in symmetric-anonymous $N$-player games, we can express $J_i$ as a function that is independent of the specific identity of the current player $i$, of $i$'s policy and other players' policy distribution following 
\begin{equation}\label{eq:j_jmathcal}
    J_i(\pi_i, \pi_{-i}) = \mathcal{J}(\pi_i, \nu_{-i}).
\end{equation}

When N players sample their policies from $\Delta_N(\Pi)$, \emph{i.e.} they sample from $\nu_N \in \Delta(\Pi) \cap \Delta_N(\Pi)$ as a distribution, the policy distribution obtained as an outcome of this sample may not match $\nu$ anymore. To guarantee that this remains the case, and that no asymmetry exists between players when sampling from members of $\Delta_N(\Pi)$, we define a new notion of sampling:

\begin{definition}[Symmetric sampling from $\Delta_N(\Pi)$]
    When $N$ players sample from $\nu_N \in \Delta_N(\Pi)$, they are \emph{symmetrically} assigned a policy from $\nu_N$ such that their population distribution is equal to $\nu_N$. The symmetrical assignment is such that the sampling distribution is invariant to player permutation.
\end{definition}

We remark that sampling from $\Delta_N(\Pi)$ is akin to an assignment. 
This new sampling definition will guarantee that our new correlated equilibrium concept is symmetric. 

Finally, we need to define the concept of population recommenders, which recommend different population distributions to the players:

\begin{definition}[Population Recommenders]
    A population recommender $\rho$ is a distribution over population distributions, \emph{i.e.} $\rho \in \Delta(\Delta_N(\Pi))$. A population distribution sampled by a population recommender is also called a \emph{population recommendation}. 
\end{definition}  

With these definitions introduced, we are in a position to rewrite both (C)CE definitions \ref{def_CE_Gen} and \ref{def_CCE_Gen} for a representative player $i$.

\begin{definition}[$N$-player Symmetric-Anonymous $\epsilon$-(Coarse)-Correlated Equilibrium]\label{def:NplayerCCE} 
We define a symmetric-anonymous $\epsilon$-(coarse)-correlated equilibrium $\rho$ as a distribution in $\Delta(\Delta_N(\Pi))$ such that $\forall i, \forall u \in \mathcal{U}_{(C)CE}$,

\[
     \mathbb{E}_{\nu \sim \rho, \, \pi_i \sim \nu} \left [ \mathcal{J}(u(\pi_i), \nu_{-i}) - \mathcal{J}(\pi_i, \nu_{-i}) \right ] \leq \epsilon\;.
\]
\end{definition}

By construction, we observe that the correlating device $\rho$ defined above only samples population distributions $\nu \in \Delta(\Pi)$. Individual players then receive player-symmetric policy recommendations such that their marginal policy distribution is equal to $\nu$ in a permutation-invariant way. Hereby, all such correlated equilibria are symmetric and anonymous, hence their names: symmetric-anonymous equilibria. Note here that $\nu_{-i}$ is computed independently of the players' policy assignments, it is the result of removing from $\nu$ the policy assigned to player $i$. 

We also see that being recommended a given policy does not necessarily imply knowing which $\nu_{-i}$ was sampled by $\rho$: knowledge of $\rho$ only allows one to make estimates about others' expected behavior.

We see below that symmetric-anonymous equilibria are in fact equivalent to standard equilibria (as in Def. \ref{def_CE_Gen} or Def. \ref{def_CCE_Gen}) that are symmetric, \emph{i.e.} that are in $\Delta_{sym}(\Pi^N) = \{ \nu \in \Delta(\Pi^N) \mid \forall \tau \text{ permutation }, \, \nu \circ \tau = \nu \}$ the set of distributions over $\Pi^N$ that are invariant to player permutations.

\begin{theorem}[Equilibrium Equivalence]\label{theorem:n_player_equivalence}
    In symmetric-anonymous $N$-player games, there is one to one correspondence between symmetric-anonymous $\epsilon$-(C)CE and $\epsilon$-(C)CE with symmetric correlating device, i.e. such that $\rho\in\Delta_{sym}(\Pi^N)$.\end{theorem}
The proof of Theorem \ref{theorem:n_player_equivalence} is located in Appendix \ref{Appendix_Sec_A}.\\

We have introduced the concept of Population Policy Distribution, and we observed that Correlation Devices can be distributions over Population Policy Distributions. Intuitively, the first concept can easily scale to the Mean-Field limit by taking $N$ to infinity in $\Delta_N(\Pi)$, thus becoming $\Delta(\Pi)$: intuitively, the "granularity" of $\Delta_N(\Pi)$ is $\frac{1}{N}$; as $N$ tends to infinity, this "granularity" tends to 0 and $\Delta_N(\Pi)$ is able to represent an increasing amount of members of $\Delta(\Pi)$ - when $N$ is infinite, both sets coincide. The second concept can be transferred from $\Delta(\Delta_N(\Pi))$ to $\mathcal{P}(\Delta(\Pi))$. In the next section, after initially defining the Mean-Field setting of interest and recalling what Mean-Field Nash equilibria are, we define Mean-Field correlated and coarse correlated equilibria in the same spirit. Section \ref{sec:properties_mean_field_eq} provides an analogue of Theorem \ref{theorem:n_player_equivalence} by proving that our new notion of correlated equilibrium is equivalent to the pre-existing ones established by Campi and Fischer~\cite{campi2020correlated}.

\section{Notions of Mean Field Equilibrium}\label{sec:notions_mean_field_eq}

We now describe a general setting, which is able to encompass the consideration of both static and dynamic Mean-Field games. The state space is denoted by $\mathcal{X}$. We denote by $\mathcal{T}$ the finite set of times within the game, so that $\mathcal{T}$ simply reduces to a singleton for static games. 
The set of distribution flows $\Delta(\mathcal{X})^\mathcal{T}$ on the state space $\mathcal{X}$ over times in $\mathcal{T}$ is denoted by $\mathcal{M}$.\\ %
Whenever every player in the population follows the policy $\pi\in\Pi$, the game generates a Mean-Field flow over $\mathcal{X}$ denoted by $\mu^\pi \in \mathcal{M}$. Formally, $\mu^\pi$ is defined by
\[
    \mu^\pi_{t+1}(x) = \sum_{x_t \in \mathcal{X}} \sum_{a \in \mathcal{A}} p(x \mid x_t, a, \mu_t^\pi)\, \pi(x_t, a) \mu_t^\pi(x_t) \qquad \forall t \in \mathcal{T}, x \in \mathcal{X}\,,
\]
with $\mu^\pi_0 = \mu_0$ a predefined initial state distribution of the population.

Given a Mean-Field flow $\mu\in\mathcal{M}$ of the population, the expected reward of a representative player playing policy $\pi\in\Pi$ is given by 
\begin{align*}
    J(\pi, \mu) &= \sum_{t \in \mathcal{T}} \sum_{x, a} r(x, a, \mu_t) \mu_t^\pi(x) \pi(x, a) \;=\; \sum_{t \in \mathcal{T}} \langle r^\pi(\cdot, \mu^\pi_t), \mu^\pi_t \rangle\;,
\end{align*}
    
where $\mu^\pi$ the expected state distribution of policy $\pi$ when the population follows the Mean-Field flow $\mu$, and $r^\pi(x, \mu) = \sum_a \pi(x, a) r(x, a, \mu)$.

Given a fixed mean-field flow $\mu\in\mathcal{M}$, an individual player can maximise their expected return by solving the following Markov Decision Process (MDP) policy optimisation problem

\begin{equation}\label{eq:objective}
    \sup_{\pi\in\Pi} J(\pi, \mu)\;.
\end{equation}

Whenever the population of players plays a distribution of strategies $\nu\in\Delta(\Pi)$, the induced Mean-Field flow over the state space $\mathcal{X}$ is denoted by 
\[
\mu(\nu) \in \mathcal{M}.
\]

In the case when the  dynamics depend on $\mu$, it is difficult to express $\mu(\nu)$ in closed form, since policies' state distributions will interfere with one another's state distributions, leading to some potentially very strong non-linearities. However, in the $\mu$-independent-dynamics case, $\mu(\nu)$ can be expressed in closed form:

\begin{lemma}[Closed-form $\mu(\nu)$]
In the $\mu$-independent-dynamics case, 
    \[
    \mu(\nu) = \sum\limits_{\pi \in \Pi} \nu(\pi) \mu^\pi.
    \]
\end{lemma}

By extension, for $\nu \in \Delta(\Pi)$, we write 
\[
\pi(\nu)
\]
the stochastic policy defined by sampling, at every initial state of the game, a policy $\pi \in \Pi$ with probability $\nu(\pi)$, and playing it until the end of the game. This definition ensures that $\mu^{\pi(\nu)} = \mu(\nu)$ by definition; however, we note that the set $\{ \pi \mid \mu^\pi = \mu(\nu) \}$ may have more than one element; in degenerate cases where $p$ does not depend on actions, for example, this set is equal to $\Pi$, as all policies have the same state distributions. This definition of $\pi(\nu)$ yields a \emph{unique} policy.

Conversely, given a policy $\bar\pi \in \bar\Pi$, we write 
\[ 
\nu_{\bar\pi} \in \Delta(\Pi)
\]
for the distribution such that 
\[
    \pi(\nu_{\bar\pi}) = \bar\pi.
\]

\begin{figure}
    \centering
    \begin{tikzpicture}[
    state/.style={circle, draw=black!60, fill=white!5, thin, minimum size=15mm},
    invisible/.style={rectangle, draw=white!60, fill=white!0, very thin, minimum size=0mm, draw opacity=1]},]
    \node[state]      (Shop)            at (0, 0)                  {Shop};

    \node[state]      (FashionA)             at (-2, -3)                  {Fashion A};
    \node[invisible, label={[align=center, color=purple]$r = -\mu(A)$}]        (rec_text2)               at (-2, -4.7) {};
    \node[state]      (FashionB)             at (2,  -3)                  {Fashion B};
    \node[invisible, label={[align=center, color=purple]$r = -\mu(B)$}]        (rec_text2)               at (2, -4.7) {};
    \draw[->] ($(Shop.south) $) -- ($ (FashionA.north) $) node [midway, above, sloped] (TextNode) {$A$};
    \draw[->] ($(Shop.south) $) -- ($ (FashionB.north) $) node [midway, above, sloped] (TextNode) {$B$};
    
    \end{tikzpicture}
    \caption{Two-Actions Hipster game.}
    \label{fig:2_actions_hipster_game}
\end{figure}
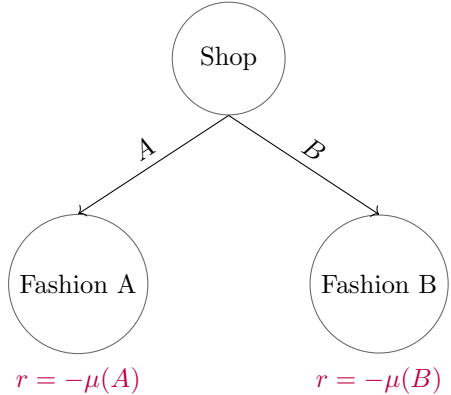

In the rest of this section, we will examine different types of game theoretic equilibria. These incorporate a notion of deviation: an equilibrium is only stable if no player has an incentive to deviate from its recommendations. These deviations are considered from the point of view of \emph{all} players. However, since all players are identical, it is enough to make sure that a given, randomly chosen player never has an incentive to deviate. If that player has no incentive to change behavior, then neither does the population. We will refer to this player as the \emph{representative player}.

\subsection{Mean Field Nash Equilibrium}\label{subsec:mfne}

The literature on Mean Field Games mostly (and almost only) focused so far on the notion of Nash equilibrium between the infinite number of agents within the population. As a generalization of Definition \ref{Nplayer_Nash}, it is naturally defined as follows:

\begin{definition}[Mean Field Nash Equilibrium, MFE]
Given $\epsilon>0$, a policy $\bar\pi\in\bar\Pi$ is an \textbf{$\epsilon$-Mean Field Nash Equilibrium} whenever
\begin{eqnarray*}
    \sup_{\pi'\in\bar\Pi} J(\pi', \mu^{\bar\pi})   
    &\leq& J(\bar\pi, \mu^{\bar\pi}) + \epsilon.
\end{eqnarray*}
It is a \emph{Mean Field Nash Equilibrium} whenever the previous relation holds for $\epsilon=0$. A Nash equilibrium is said to be \emph{pure} if it is deterministic.
\end{definition}

\begin{example}[Two-Actions Hipster game] \label{ex:hipster_game}
We give in Figure~\ref{fig:2_actions_hipster_game} an example of reward function in the Two-Actions Hipster Game: the goal for each player is to stand out from their peers by choosing the clothing item which is least frequent within the population. At a shop, agents choose either item A or item B, and are penalized for the non-uniqueness of their choice : if all agents choose Fashion A, Fashion B will grant the highest reward, and conversely. In this simplistic game, there is no pure Nash equilibrium, but only one mixed Nash (Agents choose Fashion A or B with probability $\frac{1}{2}$).
\end{example}

One of the most prominent properties of Mean Field Nash equilibria relies on their strong connection with equilibria in $N$ player games. We will later, in Section~\ref{sec:n_player_mf_eq}, explain in detail how one can use Mean-Field equilibria in N-player games. This process is at the core of the usefulness of Mean-Field games, since plugging a Mean Field Nash equilibrium in an $N$-player game yields an $\mathcal{O}\left( \frac{1}{\sqrt{N}} \right)$-approximate Nash equilibrium, which is known in the continuous time and continuous space setting in e.g.~\cite{MR2346927-HuangCainesMalhame-2006-closedLoop,Cardaliaguet-2013-notes}, and which we prove in this paper for discrete games.

Nevertheless, while the existence of Nash equilibria is a very straightforward property for MFGs to have, with clear and arguably non-restrictive conditions, deriving the uniqueness of such equilibria in general is a difficult and tedious task. One possible  approach relies on additional strong Lipschitz conditions leading to a contracting mapping  operator \cite{MR2346927-HuangCainesMalhame-2006-closedLoop}. 
Alternatively, the  so-called monotonicity condition introduced in \cite{lasry2007mean} intuitively provides players the incentive to behave differently than the full population and ensures uniqueness of the Nash Equilibrium. 
Whenever this well established condition is not satisfied, uniqueness of  Mean Field Nash equilibrium can be hard to enforce. 
A natural example for this is the converse of the Hipster game  (presented in Figure~\ref{fig:2_actions_hipster_game}) as described below.

\begin{example}[Suits Game] \label{ex:suits_game}
In the Suit game, rewards are inverted compared to the Hipster game (players are incentivized to act similarly to others). This game does not satisfy the monotonicity condition \cite{lasry2007mean} and  has 3 Nash equilibria:  all-in on Fashion A, B; or 50\% on each.
\end{example}

When the MFE is not unique, one possible option is to help the players synchronize using an extraneous noise or signal. Restoring uniqueness of MFE via the addition of vanishing common noise has been onserved in \cite{delarue2019restoring}. Alternatively, the addition of a common signal sent to the full population naturally calls for notions of correlated or coarse correlated equilibria \cite{aumann1974subjectivity, Aumann1987CorrelatedEA}. With the exception of \cite{DeglInnocenti-phdthesis,campi2020correlated}, Nash equilibria are surprisingly the only type of equilibrium considered in the MFG literature. 
This is in stark contrast with the literature on $N$-player games, where weaker notions of equilibria are well established and understood \cite{BlumInternalExternalRegret, blumregret, morrill2021efficient, morrill2020hindsight, Aumann1987CorrelatedEA, farina2022faster, farina2019coarse, gordon2008regret}. %
Specifically, it is understood that there exists a tight correspondence between no-regret dynamics and coarse correlated equilibria~\cite{monnot2017limits}. Moreover,  worst case analysis for Nash equilibria 
can sometimes automatically be extended without any further degradation of performance to worst case (coarse) correlated equilibria via what is known as robust Price of Anarchy analysis~\cite{roughgarden2015intrinsic}. %

\subsection{Intuition on Correlation Device and Correlated Equilibria}

We are now in position to generalize the concepts of correlation device and (coarse) correlated equilibrium to the Mean-Field setting by building on new formulations derived in Section~\ref{sec:n_player_anonymous_games}. Before doing so, let first provide relevant intuitions for these new concepts and facilitate their interpretation. 

\vspace{0.3cm}

\noindent \textbf{Correlation Device.} A correlation device makes a single policy recommendation to each player in the game. It coordinates the population's actions. In the well-known traffic lights example\footnote{In a hypothetical intersection where traffic laws would not hold.}, the correlation device sets the lights' colors, and lets agents (cars) decide whether to follow or not the lights' signals.

\vspace{0.3cm}

\noindent \textbf{Coarse-Correlated Equilibrium - Agent's perspective.} From the perspective of the agent, we can imagine that the correlation device is a mediator who is partially aligned to the agent's interests and has a bird's eye view of what the population is doing. In a coarse correlated equilibrium, the agent has two choices: either delegate all decisions to the mediator - despite the partial misalignment -, or take its own decisions, without the mediator's knowledge of what the rest of the population will be doing. If the agent has a larger incentive to use the services of the mediator on average, then the mediator's recommendations may be said to be a coarse correlated equilibrium.

\vspace{0.3cm}

\noindent \textbf{Correlated Equilibrium - Agent's perspective.} Keeping the mediator's analogy, in a correlated equilibrium situation, the agent has two choices: accept the mediator's suggested course of action, or refuse it and choose their own course. This case differs from the coarse correlated case by the fact that here, the agent sees which course of action the mediator has prepared, and, from it, can estimate what the other agents may be recommended by the mediator. Having more information - but not as much information as the mediator -, the agent may therefore take better-informed decisions. However, if despite this, the agent prefers to follow the mediator's suggestion, then the mediator's recommendations may be said to be a correlated equilibrium\footnote{On a philosophical note, a striking relationship between the concept of Correlated Equilibrium and that of Manager Efficiency developed by~\citet{macintyre2013after}. We hope to see such philosophical links developed more thoroughly in the future.}.\\

\vspace{0.3cm}

Whenever correlation devices are discrete probability distributions, a visualization of how correlation devices operate, for the homogeneous (only one recommended policy to the population) and non-homogeneous (heterogeneous, several deterministic policies may be recommended at once) cases, are respectively available in Figures~\ref{fig:homogeneous_diagram} and~\ref{fig:diagram}.

\subsection{Mean-Field Correlation Device}

Whenever several Nash equilibria exist, an equilibrium selection problem arises: the  population needs more guidance in order to be able to coordinate and synchronize. As noted before, in $N$-player games, the notions of correlated and coarse correlated equilibria bypass this issue through the use of a correlation device, which provides a signal allowing the population to synchronize; and so do they in Mean-Field games. 

\begin{definition}[Population distribution/recommendation] We introduce the following.
\begin{itemize}
    \item A \textbf{population distribution, or population recommendation $\nu \in \Delta(\Pi)$} is a distribution over the set of policies $\Pi$; %
    \item Given a {population distribution $\nu\in\Delta(\Pi)$}, each player receives an \textbf{individual recommendation} $\pi\in\Pi$ uniformly sampled from $\nu$, so that the distribution of all individual recommendations over the population is $\nu$.
\end{itemize} 
\end{definition}

As detailed in Section~\ref{sec:MFCE} below, correlated equilibria encompass an information asymmetry component: while the recommender knows the full population recommendation, the  players - the recommendees - only have access to their own recommendation, which can allow for complex cooperative behavior. Nevertheless, all players are also aware of the possible population distributions, together with their probability of occurrences. This information is contained into what we call correlation devices, whose definition in the Mean-Field setting is as follows.

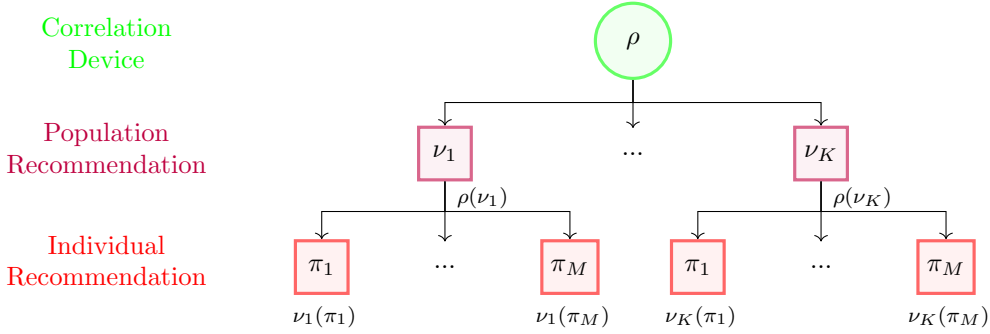
\begin{figure}[t!]
    \centering
    \begin{tikzpicture}[
    recommender/.style={circle, draw=green!60, fill=green!5, very thick, minimum size=10mm},
    invisible/.style={rectangle, draw=white!60, fill=white!0, very thin, minimum size=0mm, draw opacity=1]},
    population/.style={rectangle, draw=purple!60, fill=purple!5, very thick, minimum size=7mm},
    agent/.style={rectangle, draw=red!60, fill=red!5, very thick, minimum size=7mm},
    ]
    \node[recommender]      (recommender)            at (0, 0)                  {$\rho$};
    \node[invisible, label={[align=center, color=green]Correlation \\ Device}]        (rec_text)               at (-7, -0.6) {};
    \node[invisible]        (invis_rec)               at (0, -1.5) {...};

    \node[population]       (pop_left)       at (-2.5, -1.5)  {$\nu_1$};
    \node[invisible, label={[align=center]\footnotesize{$\rho(\nu_1)$}}] (rec_text) at (-2, -2.47) {};

    \node[population]      (pop_right)       at ( 2.5, -1.5) {$\nu_K$};
    \node[invisible, label={[align=center]\footnotesize{$\rho(\nu_K)$}}] (rec_text) at (3.05, -2.47) {};

    \node[invisible, label={[align=center, color=purple]Population \\ Recommendation}]        (rec_text2)               at (-7, -2) {};
    
    \node[invisible]       (invis_left)       [below=of pop_left] {...};
    \node[invisible]       (invis_right)       [below=of pop_right] {...};
    
    \node[agent]        (all)       [left=of invis_left] {$\pi_1$};
    \node[invisible, label={[align=center]\footnotesize{$\nu_1(\pi_1)$}}] (rec_text) at (-4.1, -4.05) {};

    \node[agent]        (alr)       [right=of invis_left] {$\pi_M$};
    \node[invisible, label={[align=center]\footnotesize{$\nu_1(\pi_M)$}}] (rec_text) at (-0.8, -4.05) {};

    \node[invisible, label={[align=center, color=red]Individual \\ Recommendation}]   (rec_text3)  at (-7, -3.5) {};
    
    \node[agent]        (arl)       [left=of invis_right] {$\pi_1$};
    \node[invisible, label={[align=center]\footnotesize{$\nu_K(\pi_1)$}}] (rec_text) at (0.9, -4.05) {};

    \node[agent]        (arr)       [right=of invis_right] {$\pi_M$};
    \node[invisible, label={[align=center]\footnotesize{$\nu_K(\pi_M)$}}] (rec_text) at (4.2, -4.05) {};
    
    \draw[->] ($(recommender.south) $) -- ++(0,-.3) -| ($ (pop_left.north) $);
    \draw[->] ($(recommender.south) $) -- ($ (invis_rec.north) + (0, 2mm) $);
    \draw[->] ($(recommender.south) $) -- ++(0,-.3) -| ($ (pop_right.north) $);

    \draw[->] ($(pop_left.south) $) -- ++(0,-.4) -| ($ (all.north) $);
    \draw[->] ($(pop_left.south) $) -- ($ (invis_left.north) + (0, 2mm) $);
    \draw[->] ($(pop_left.south) $) -- ++(0,-.4) -| ($ (alr.north) $);
    
    \draw[->] ($(pop_right.south) $) -- ++(0,-.4) -| ($ (arl.north) $);
    \draw[->] ($(pop_right.south) $) -- ($ (invis_right.north) + (0, 2mm) $);
    \draw[->] ($(pop_right.south) $) -- ++(0,-.4) -| ($ (arr.north) $);
    \end{tikzpicture}
    \caption{Structure of a discrete Mean-Field correlation device $\rho\in\mathcal{P}(\Delta(\Pi))$.}
    \label{fig:diagram}
\end{figure}

\begin{definition}[Correlation device]\label{def_Correlation_device}
    A \textbf{correlation device} is a distribution $\rho$ over $\Delta(\Pi)$. It encapsulates the possible population recommendations given to the population - we denote $\mathcal{P}(\Delta(\Pi))$ the set of correlation devices. %
\end{definition}

A Mean-Field correlation device is a distribution over population recommendations that synchronizes all individuals in the population. Its structure is presented in Figure \ref{fig:diagram}. The exogenous recommender picks a realization of a random variable with distribution $\rho\in\mathcal{P}(\Delta(\Pi))$ and gives each player its own individual recommendation $\pi\in\Pi$ as a signal. All players know $\rho$ together with their own individual recommendation $\pi\in\Pi$, but do not have access to the population recommendation $\nu\in\Delta(\Pi)$ sampled by the recommender. Whenever a player receives $\pi\in\Pi$ as recommendation, their belief about the possible population distributions  shifts to $\rho(\cdot\mid\pi)$ defined by: for $\nu\in\Delta(\Pi)$,
\begin{align}\label{eq:cond_distrib_d_rho_nu_pi}
    d\rho(\nu \mid \pi) \;:=\;  \frac{\nu(\pi) d\rho(\nu)}{\int_{\nu' \in \Delta(\Pi)} \nu'(\pi) d\rho(\nu')}.
\end{align}
This conditional distribution goes in pair with the distribution $\rho_\Pi$ over $\Pi$ induced by the correlation device $\rho\in\mathcal{P}(\Delta(\Pi))$ which is defined by 
\begin{align}\label{eq:cond_distrib_rho_nu_pi}
    \rho_\Pi(\pi)&:= \int_{\nu\in\Delta(\Pi)}\nu(\pi) d\rho(\nu)\;,
\end{align}
 so that 
 
 \[
    d\rho(\nu)=\sum\limits_{\pi\in\Pi} \rho_\Pi(\pi) d\rho(\nu\mid\pi).
 \]

By our definition, agents never observe $\nu$ : the whole stochasticity of the process resides in the centralized instance, which samples both $\nu$ and a policy from $\nu$ for each agent.
However, we could also imagine that $\rho$ would send $\nu$ to each agent, and lets agents sample their policy from $\nu$ for the duration of an episode. 
In this case, agents all play the same policy $\pi(\nu) \in \bar\Pi$, and all know what the other agents are playing. 
We call such $\rho$, which communicate $\nu$ to the players, \emph{\textbf{homogeneous} correlation devices}. 
We note that $\rho$ samples $\nu$ and transfers it to players, which then play $\pi(\nu)$.
We can therefore view $\rho$ as a distribution over the possible values of $\pi(\nu)$, \emph{i.e.} over $\bar\Pi$. We formalize this notion:

\begin{definition}[Homogeneous correlation device]\label{def:homogeneousCE}
 A \textbf{homogeneous correlation device} $\rho_h\in\mathcal{P}(\bar\Pi)$ is a special type of correlation device that samples stochastic policies, and only recommends one stochastic policy to all players in the population.
\end{definition}

Here is an example of a homogeneous correlation device.
\begin{example}
Let us consider again the Suits Game, defined in Example~\ref{ex:suits_game}, in which each player is incentivized to pick a fashion well represented in the population. A correlation device alternatively recommending all players to choose Fashion A and Fashion B (\emph{i.e.} 50\% of the time, it recommends Fashion A to all players; 50\% of the time, Fashion B) is a homogeneous correlation device, that happens to generate a Mean-Field correlated equilibrium,  as discussed in the next section.
\end{example}

Intuitively, since all players know what other players are playing, some homogeneous equilibria should find themselves very restricted. We show that this is indeed the case in Section~\ref{subsec:homogeneous_ce_props}.

\begin{figure}[t!]
    \centering
    \begin{tikzpicture}[
    recommender/.style={circle, draw=green!60, fill=green!5, very thick, minimum size=10mm},
    invisible/.style={rectangle, draw=white!60, fill=white!0, very thin, minimum size=0mm, draw opacity=1]},
    population/.style={rectangle, draw=purple!60, fill=purple!5, very thick, minimum size=7mm},
    agent/.style={rectangle, draw=red!60, fill=red!5, very thick, minimum size=7mm},
    ]
    \node[recommender]      (recommender)            at (0, 0)                  {$\rho$};
    \node[invisible, label={[align=center, color=green]Correlation \\ Device}]        (rec_text)               at (-7, -0.6) {};
    \node[invisible]        (invis_rec)               at (0, -1.5) {...};

    \node[population]       (pop_left)       at (-2.5, -1.5)  {$\nu_1$};
    \node[invisible, label={[align=center]\footnotesize{$\rho(\nu_1)$}}] (rec_text) at (-2, -2.47) {};

    \node[population]      (pop_right)       at ( 2.5, -1.5) {$\nu_M$};
    \node[invisible, label={[align=center]\footnotesize{$\rho(\nu_M)$}}] (rec_text) at (3.05, -2.47) {};

    \node[invisible, label={[align=center, color=purple]Population \\ Recommendation}]        (rec_text2)               at (-7, -2) {};
    
    \node[agent]        (all)      at (-2.5, -3.0) {$\pi(\nu_1)$};

    \node[invisible, label={[align=center, color=red]Individual \\ Recommendation}]   (rec_text3)  at (-7, -3.5) {};
    
    \node[agent]        (arl)      at (2.5, -3.0) {$\pi(\nu_M)$};

    \draw[->] ($(recommender.south) $) -- ++(0,-.3) -| ($ (pop_left.north) $);
    \draw[->] ($(recommender.south) $) -- ($ (invis_rec.north) + (0, 2mm) $);
    \draw[->] ($(recommender.south) $) -- ++(0,-.3) -| ($ (pop_right.north) $);

    \draw[->] ($(pop_left.south) $) -- ++(0,-.4) -| ($ (all.north) $);
    \draw[->] ($(pop_right.south) $) -- ++(0,-.4) -| ($ (arl.north) $);
    \end{tikzpicture}
    \caption{Structure of a discrete homogeneous Mean-Field correlation device $\rho\in\mathcal{P}(\Delta(\Pi))$.}
    \label{fig:homogeneous_diagram}
\end{figure}
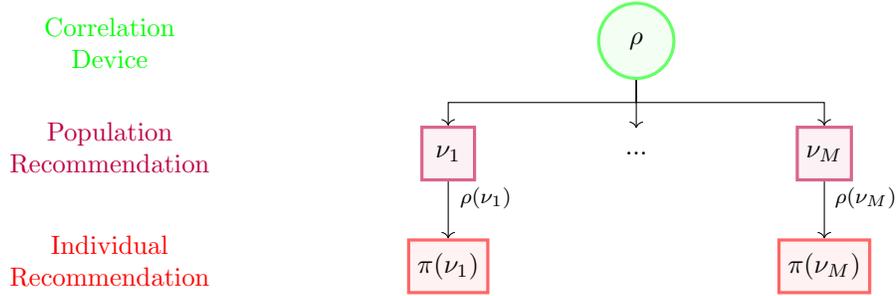

\subsection{Mean-Field Correlated Equilibrium}
\label{sec:MFCE}

We now turn to the definition of correlated equilibrium for Mean-Field games, which is built as a natural extension to the one considered in $N$-player games. 
We define Mean-Field correlated equilibria similarly to their anonymous $N$-player version derived in Definition \ref{def:NplayerCCE} above. %

\begin{definition}[Mean Field Correlated  Equilibrium, MFCE] Given $\epsilon>0$,
a correlation device $\rho$ is an \textbf{$\epsilon$- Mean Field Correlated Equilibrium} if, $\forall u \in \mathcal{U}_{CE}$
 \begin{eqnarray}\label{eq:def_MFCE}
    \mathbb{E}_{\nu \sim \rho, \, \pi \sim \nu}\left[J(u(\pi), \mu(\nu)) - J(\pi,\mu(\nu))\right] 
        \leq 
    \epsilon\;.
 \end{eqnarray}

It is called \textbf{Mean Field Correlated Equilibrium} whenever the previous relation holds for $\epsilon=0$.
\end{definition}

This definition of Mean Field Correlated equilibrium aligns naturally with the one developed in the Game theory literature \cite{Aumann1987CorrelatedEA}. Besides, we will verify in Section \ref{subsec:campi_Fischer} below that it also connects in an elegant fashion to the one introduced recently in  \cite{campi2020correlated}. 

The next result provides a geometric property of the set of Mean-Field correlated equilibria. 

\begin{proposition}\label{Prop:Convex_MF-CE}
    For all $\epsilon \geq 0$, the set of $\epsilon$-MFCEs is convex.
\end{proposition}

\begin{proof}
    Let $\epsilon \geq 0$, $\rho_0$, $\rho_1$ be two $\epsilon$-MFCE. Let $0 \leq \alpha \leq 1$ and let $\rho_\alpha$ be the barycentric correlation device $\alpha \rho_0 + (1-\alpha) \rho_1\in\mathcal{P}(\Delta(\Pi))$. 
    
    Let $u \in \mathcal{U}_{CE}$.
    \begin{align*}
        &\mathbb{E}_{\nu \sim \rho_\alpha, \, \pi \sim \nu}[J(u(\pi), \mu(\nu)) - J(\pi, \mu(\nu))] \\ = &\alpha \mathbb{E}_{\nu \sim \rho_0, \, \pi \sim \nu}[J(u(\pi), \mu(\nu)) - J(\pi, \mu(\nu))] 
        + (1 - \alpha) \mathbb{E}_{\nu \sim \rho_0, \, \pi \sim \nu}[J(u(\pi), \mu(\nu)) - J(\pi, \mu(\nu))] \\
        \leq &\epsilon
    \end{align*}
    
\end{proof}

The set of correlated equilibria is behaving as we expect. We now turn towards the set of \emph{homogeneous correlated equilibria}. There is a significant information difference between correlated equilibria and homogeneous correlated equilibria: while the former's agents only observe their own recommendation, the latter's observe the full population recommendation. This means that the deviations they consider will have more granularity than $\mathcal{U}_{CE}$: each population recommendation will correspond to one specific deviation, \emph{i.e.} homogeneous correlated equilibria's deviation functions are $\mathcal{U} = \{u \mid u: \bar\Pi \rightarrow \bar\Pi \}$. This concept can be linked with the notion of $\Phi$-regret introduced in \citet{piliouras2021phiregret}. We formally define homogeneous correlated equilibria, given their deviation set $\mathcal{U}^h_{CE} = \{ u \mid u: \bar\Pi \rightarrow \bar\Pi \}$,

\begin{definition}[Homogeneous Mean Field Correlated  Equilibrium, MFCE] Given $\epsilon>0$,
a homogeneous correlation device $\rho$ is an \textbf{$\epsilon$- Homogeneous Mean Field Correlated Equilibrium} if, 
 \begin{eqnarray}\label{eq:def_hom_ce}
    \mathbb{E}_{\nu \sim \rho}\left[J(u(\pi(\nu)), \mu(\nu)) - J(\pi(\nu),\mu(\nu))\right] 
        \leq 
    \epsilon \qquad \forall u \in \mathcal{U}^h_{CE}.
 \end{eqnarray}

It is called \textbf{Homogeneous Mean Field Correlated Equilibrium} whenever the previous relation holds for $\epsilon=0$.
\end{definition}

\subsection{Mean-Field Coarse Correlated  Equilibrium}

In $N$-player games, computing Correlated Equilibria can be very expensive \cite{morrill2021efficient}. Hereby, another set of equilibria, wider and easier to compute, was introduced in this setting: coarse correlated equilibria. Up to our knowledge, such notion has never been studied in teh framework of mean Field Games.  A coarse correlated equilibrium is a weaker notion of equilibrium, where each player may only choose to deviate from their recommendation before having observed it - though players are still assumed to have knowledge of the correlation device's behavior $\rho\in\mathcal{P}(\Delta(\Pi))$. This larger class of equilibria contains correlated equilibria and is more easily reachable by classical learning algorithms, as will be discussed in Section \ref{sec:learning}.

\begin{definition}[Mean Field Coarse Correlated equilibrium, MFCCE]\label{def:MFCCE}
Given $\epsilon>0$, a correlation device $\rho$ is an \textbf{$\epsilon$-Mean-Field Coarse Correlated Equilibrium} if
\begin{eqnarray}\label{eq:CE_deviations}
    \mathbb{E}_{\pi \sim \nu, \nu \sim \rho} \left[J\left(u(\pi), \mu(\nu)\right) - J\left(\pi, \mu(\nu)\right) \right] &\leq& \epsilon \;, \quad \forall u\in\mathcal{U}_{CCE}\,.
\end{eqnarray}

It is a \emph{Mean-Field Coarse Correlated Equilibrium} whenever the previous equation holds for $\epsilon=0$. 
\end{definition}

Recall that $\mathcal{U}_{CCE}$ denotes the set constant deviations over $\Pi$, i.e. the mappings from $\Pi$ to $\Pi$ which a fixed constant policy $\pi\in\Pi$.
MFCCEs can also be defined in an alternative way.
\begin{proposition}[MFCCE characterization using best-responses]\label{prop:MFCCE_characterisation}
    A correlation device $\rho$ is an $\epsilon$-MFCCE if and only if, %
    \begin{eqnarray*}
    \sup_{\pi'\in\Pi}  \mathbb{E}_{\nu\sim\rho}[J(\pi', \mu(\nu))]
    &\leq& 
    \mathbb{E}_{\pi\sim\nu,\nu\sim\rho}\left[J(\pi, \mu(\nu))\right]+\epsilon
\end{eqnarray*}
\end{proposition}

\begin{proof}  The proof follows from identifying $\Pi$ with $\{u(\pi),\, u\in\mathcal{U}_{CCE}\; \mbox{and} \; \pi\in\Pi\}$. 
\end{proof}

\begin{proposition}[MFCEs are MFCCEs]\label{prop:mfce_in_mfcce}
The set of $\epsilon$-MFCE is included in the set of $\epsilon$-MFCCE. 
\end{proposition}

\begin{proof} This property is a direct implication from the definition of MFCEs and Proposition~\ref{prop:MFCCE_characterisation}, when it is noted that $\mathcal{U}_{CCE} \subseteq \mathcal{U}_{CE}$.
\end{proof}

Inclusions between the sets of Nash, correlated and coarse correlated equilibria are represented in Figure \ref{fig:inclusion_visualization}. Besides, MFCCEs being much less restrictive than MFCEs, both sets rarely coincide. However, they can consistently coincide in very small games.

\begin{proposition}
    In two-action one-state Mean-Field games, the set of MFCEs and MFCCEs are equal.
\end{proposition}
\begin{proof}
    We already know that the set of MFCEs is included in the set of MFCCEs. The reverse inclusion is proven by observing that in this particular setting, unilateral deviation to either action is equivalent to deviating when being recommended the other action - thus being optimal for unilateral deviations is equivalent to being optimal for per-action deviations.
    
    Note that this does not imply that $\mathcal{U}_{CE} = \mathcal{U}_{CCE}$ - indeed, members of $\mathcal{U}_{CE}$ which switch both policies at the same time can not be members of $\mathcal{U}_{CCE}$.
\end{proof}

We note that this does not mean that $U_{CE} = U_{CCE}$ in these settings. Indeed, a deviation function which switches both actions is a member of $U_{CE}$ but not of $U_{CCE}$. However, if a payoff stands to be gained by deviating from one action to another, then it means that the other action is more profitable in general, and thus that only unilateral deviations towards it matter.

Just like MFCEs, the set of MFCCEs is also convex:

\begin{proposition}
 For all $\epsilon \geq 0$, the set of $\epsilon$-MFCCE is convex.
\end{proposition}
\begin{proof}
    Similar to the one of Proposition \ref{Prop:Convex_MF-CE}.
\end{proof}

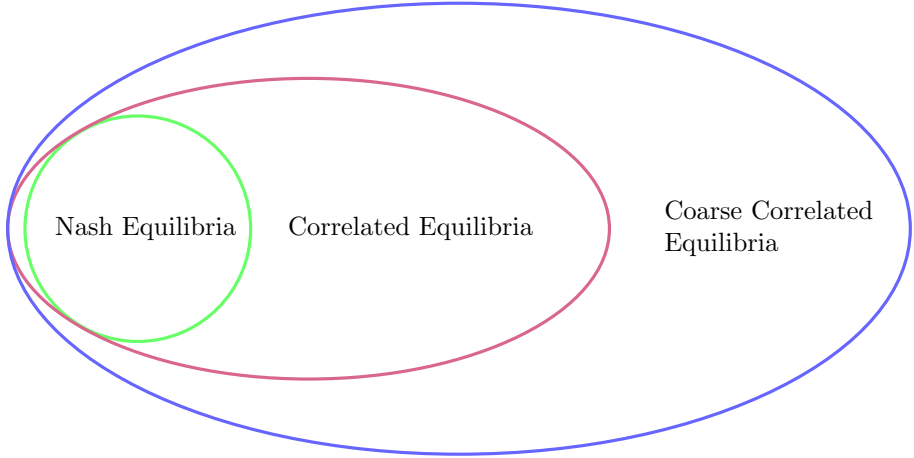
\begin{figure}[ht]
    \centering
    \begin{tikzpicture}
        \draw[color=green!60, very thick](-1,0) circle (1.5);
        \node[text width=3cm] at (-.6,0) {Nash Equilibria};
    
        \draw[color=purple!60, very thick](1.27,0) ellipse (4 and 2);
        \node[text width=8cm] at (5,0) {Correlated Equilibria};
    
        \draw[color=blue!60, very thick](3.27,0) ellipse (6 and 3);
        \node[text width=4cm] at (8,0) {Coarse Correlated \\ Equilibria};

    \end{tikzpicture}
    \caption{Visualization of the typical inclusion relationships between equilibrium sets.}
    \label{fig:inclusion_visualization}
\end{figure}

The definition of a homogeneous coarse correlated equilibrium is similar to that of a correlated equilibrium: indeed, coarse correlated equilibrium deviate before receiving any play information. More formally, with $ \mathcal{U}^h_{CCE} = \{ u \mid u: \bar\Pi \rightarrow \bar\Pi, \; u \text{ constant.} \}$ their deviation set,

\begin{definition}[Mean Field Coarse Correlated equilibrium, MFCCE]\label{def:hom_MFCCE}
Given $\epsilon>0$, a homogeneous correlation device $\rho$ is an \textbf{$\epsilon$-Homogeneous Mean-Field Coarse Correlated Equilibrium} if
\begin{eqnarray}\label{eq:hom_cce_deviations}
    \mathbb{E}_{\pi \sim \nu, \nu \sim \rho} \left[J\left(u(\pi), \mu(\nu)\right) - J\left(\pi, \mu(\nu)\right) \right] &\leq& \epsilon \;, \quad \forall u \in \mathcal{U}^h_{CCE} \,.
\end{eqnarray}

It is a \emph{Homogeneous Mean-Field Coarse Correlated Equilibrium} whenever the previous equation holds for $\epsilon=0$. 
\end{definition}

\subsection{Equilibrium Sets Visualization in a Toy Example}\label{subsec:eq_viz}

This section ambitions to highlight how vast the set of correlated equilibria can be in comparison to the set of Nash equilibria, and more strikingly how vast the set of coarse correlated equilibria is compared to the set of correlated equilibria. In a word, we illustrate the assertion depicted in Figure \ref{fig:inclusion_visualization}: 
\[
\text{ Nash Equilibria } \subseteq \text{ Correlated Equilibria } \subseteq \text{ Coarse Correlated Equilibria. }
\]
and evaluate the size of these sets in a simple game.

\begin{example}Let consider the following 3-actions (A, B and C) static Mean-Field Dominated-Action game:
$$r(A, \mu) = \mu(A) + \mu(C)\,,
    \qquad r(B, \mu) = \mu(B)\,,
    \qquad r(C, \mu) = \mu(A) + \mu(C) - 0.05\,,$$
where $\mu(X)$ abusively denotes the proportion of players picking action $X$ in the population (i.e. the state of a player reduces to their action).
A visualization of its Mean-Field Nash, correlated and coarse correlated equilibria is provided in Figure \ref{fig:mfcce_visualization}. \end{example}

In general, visualizing the sets of correlated equilibria is difficult. Indeed, each correlated equilibrium is a distribution over distribution of policies. Therefore,  a correlated equilibrium is in general composed by several different mixed policies at once. It is easy to see how to visualize one such equilibrium, but less obvious how to visualize their set, especially when the number of such mixed policies may be infinite. However, in our example, one of the three available actions is dominated: whenever an agent is recommended to play $C$, they know that they should play $A$ instead! 
Correlated equilibria are therefore restricted to recommending either $A$ or $B$. We know that any mixture between homogeneously recommending $A$ and $B$ to the population yields a CE, so that the set of CEs is the straight line between $A$ and $B$ in Figure \ref{fig:mfcce_visualization}.

Visualizing the set of coarse correlated equilibria is much harder, even more so in this simple game. Indeed, one can recommend action $C$ homogeneously and still get many coarse correlated equilibria, so we can not use the simplifying assumption used for CEs. We choose to restrict to the set of homogeneous CCEs, more precisely, we represent $(\alpha, \, \beta, \gamma)$ such that $\rho = \alpha \delta_A + \beta \delta_B + \gamma \delta_C$ is a CCE. We observe in Figure \ref{fig:mfcce_visualization} that the set of CCEs is represented by a very large triangle in the simplex, so that the correlation device can recommend the dominated action $C$. More strikingly, the set of CCEs reveals to be significantly larger than the set of CEs and Nash equilibria.

\begin{figure}[ht]
    \centering
    \includegraphics[scale=0.4]{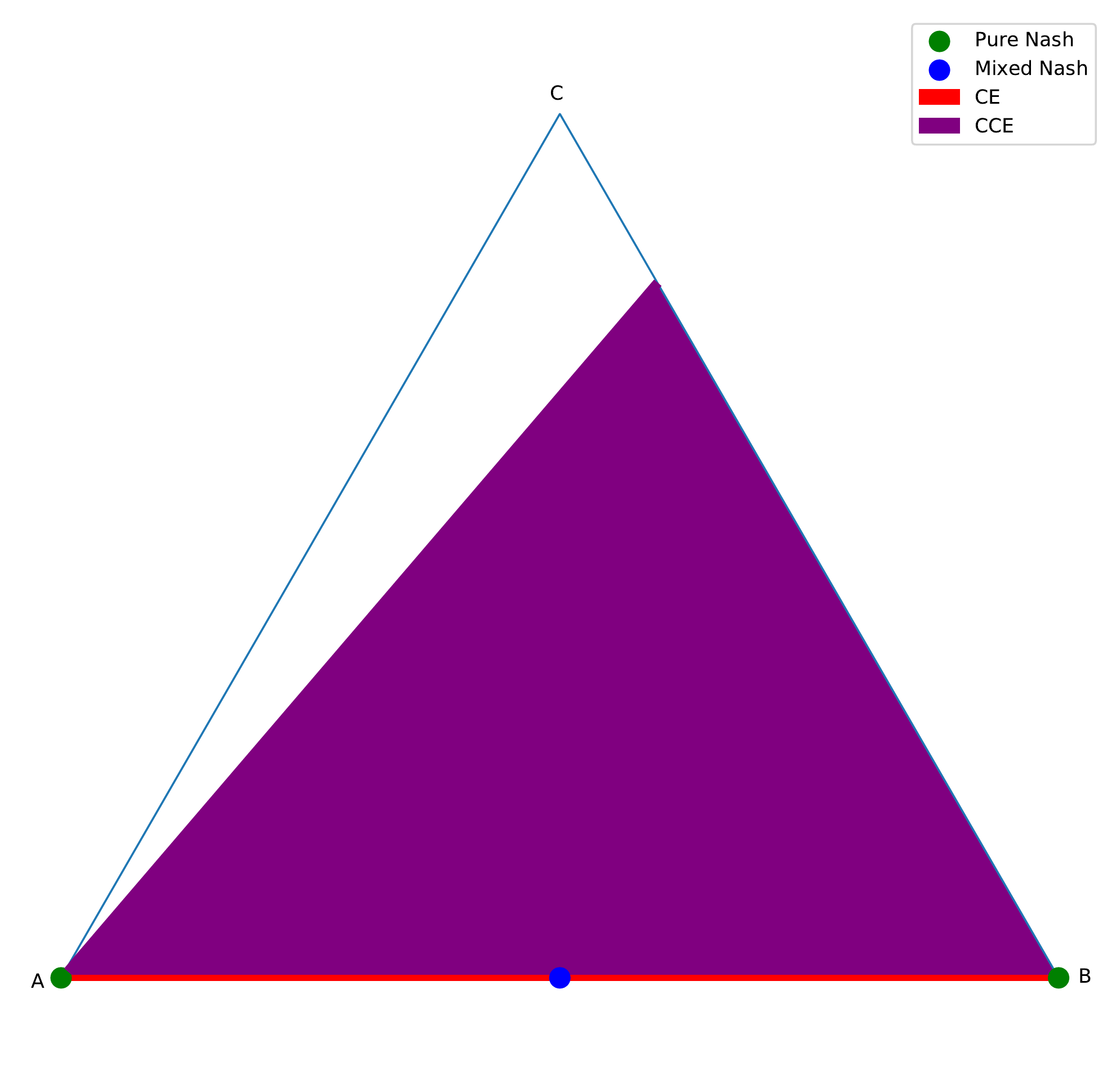}
    \caption{Visualization of Mean-Field Equilibria on the Dominated-Action game.}
    \label{fig:mfcce_visualization}
\end{figure}

\section{Properties of Mean Field (Coarse) Correlated Equilibria}\label{sec:properties_mean_field_eq}

In this section, we investigate several properties of our (coarse) correlated equilibrium framework. First, in Section~\ref{subsec:nash_c_ce_relationships}, we detail relationships between Nash equilibria and (coarse) correlated equilibria. Then, in Section~\ref{subsec:c_ce_existence}, we detail existence conditions for Mean-Field (coarse) correlated equilibria, and find surprising situations where \emph{no (coarse) correlated equilibrium exists}. This is mitigated by the existence, for all $\epsilon > 0$, of $\epsilon$-(coarse) correlated equilibria. Then, in Section~\ref{subsec:campi_Fischer}, we establish equivalence between our notion of correlated equilibrium and the one presented by Campi and Fischer~\cite{campi2020correlated}, thereby inheriting all their asymptotic properties. Finally, in Section~\ref{subsec:homogeneous_ce_props}, we characterize special properties of homogeneous Mean-Field correlated equilibria.

\subsection{Relationship Between (Coarse) Correlated Equilibria and Nash Equilibria}\label{subsec:nash_c_ce_relationships}

In 2-player zero-sum games, correlated equilibria are strongly linked to Nash equilibria: their marginalizations are Nash equilibria; a correlation device recommending a Nash equilibrium is also a correlated equilibrium, and a (coarse) correlated equilibrium which only recommends one (possibly stochastic) joint policy actually recommends a Nash equilibrium !

Mirroring these statements, we first show how any $\epsilon$-Nash equilibrium can be transformed into an $\epsilon$-correlated equilibrium; then that, given any $\epsilon$-correlated equilibrium recommending only one $\nu$, $\pi(\nu)$ is an $\epsilon$-Nash equilibrium. 
Finally, we analyze the question of (coarse) correlated equilibrium marginalizations, defining what they exactly are, when they exist, and conditions for them to be Nash equilibria. 

\subsubsection{From Nash Equilibria to Correlated Equilibria}\label{subsubsec:nash_to_ce}

We first start by showing how one can derive Correlated equilibria from Nash equilibria. 

\begin{proposition}[Nash-derived Correlated Equilibrium]\label{prop:nash_to_ce}
    Every $\epsilon$-Nash equilibrium can be transformed into a Correlated Equilibrium.
\end{proposition}
\begin{proof}
    Let $\pi^* \in \bar\Pi$ be a Nash equilibrium. We write $\nu^* = \nu_{\pi^*}$ for conciseness, and take $\rho = \delta_{\nu^*}$.
    
    $\rho$ is an $\epsilon$-correlated equilibrium: if there existed $u \in \mathcal{U}_{CE}$ such that 
    \[
        J(u(\pi(\nu^*)), \mu(\nu^*)) - J(\pi(\nu^*), \mu(\nu^*)) > \epsilon
    \]
    then, since $\pi(\nu^*) = \pi^*$ and $\mu(\nu^*) = \mu^{\pi^*}$, this would imply that $u(\pi^*)$ is a policy which has higher value against the Nash than the Nash policy $\pi^*$ plus $\epsilon$, which is strictly impossible by definition.
    
    Therefore every $\epsilon$-Nash equilibrium can be transformed into an $\epsilon$-correlated equilibrium.
\end{proof}

\subsubsection{From Coarse Correlated Equilibria to Nash Equilibria}

We now examine the converse of the above property - when can we extract an $\epsilon$-Nash equilibrium from an $\epsilon$-correlated equilibrium ? We show that this is at least possible when the correlated equilibrium is a single Dirac:

\begin{proposition}[Coarse correlated equilibrium-derived Nash equilibrium]\label{prop:cce_to_nash}
    Assume $\rho = \delta_\nu$, with $\nu \in \Delta(\Pi)$, is an $\epsilon$-coarse correlated equilibrium. Then $\pi(\nu)$ is an $\epsilon$-Nash equilibrium.
\end{proposition}
\begin{proof}
    We write the optimality condition of $\rho$ for all $u \in \mathcal{U}_{CCE}$:
    \[
        \mathbb{E}_{\nu\sim\rho, \pi\sim\nu} [J(u(\pi), \mu(\nu)) - J(\pi, \mu(\nu))] \leq \epsilon,
    \]
    
    \[
        \mathbb{E}_{\pi\sim\nu} [J(u(\pi), \mu(\nu)) - J(\pi, \mu(\nu))] \leq \epsilon,
    \]
    
    \emph{i.e.}, $\forall \pi' \in \Pi$,
    
    \[
        J(\pi', \mu(\nu)) - \mathbb{E}_{\pi\sim\nu} [J(\pi, \mu(\nu))] \leq \epsilon.
    \]
    
    Finally, we note that $\mathbb{E}_{\pi\sim\nu} [J(\pi, \mu(\nu))] = J(\pi(\nu), \mu(\nu))$, which concludes the proof:
    
    \[
        J(\pi', \mu(\nu)) - J(\pi(\nu), \mu(\nu)) \leq \epsilon,
    \]
    
    \emph{i.e.} $\pi(\nu)$ is an $\epsilon$-Nash equilibrium.
\end{proof}

We also show that, in certain classes of games, the marginalization - defined in Definition~\ref{def:marginalization} - of an $\epsilon$-Mean-Field coarse-correlated equilibrium yields an $\epsilon$-Nash equilibrium

We first define properly what the marginalization of a correlation device is:

\begin{definition}[Correlation Device Marginalization]\label{def:marginalization}
    The marginalization $\hat\pi$ of a correlation device $\rho$ is defined as the policy whose distribution is equal to $\int_\nu \mu(\nu) d\rho(\nu)$. 
\end{definition}

Note that it always exists when the dynamics do not depend on the distribution:

\begin{proposition}[Existence of the marginalization]
    In games where the dynamics do not depend on the mean field flow, the marginalization of a correlation device always exists, and is equal to
    
    \[
    \hat\pi_t(s, a) = \sum_{\pi\in\Pi} \frac{\int_\nu \nu(\pi) \mu_t^\pi(s) d\rho(\nu)}{\sum_{\pi'\in\Pi} \int_{\nu'} \nu'(\pi') \mu_t^{\pi'}(s) d\rho(\nu')} \pi(s, a). 
    \]
\end{proposition}
\begin{proof}
    Let first write the distribution evolution equation for $\hat\pi$:
    
    $$\mu_{t+1}^{\hat\pi}(x) = \sum_{x_t, a} p(x \mid x_t, a) \hat\pi(x_t, a) \mu_{t}^{\hat\pi}(x_t)\,.$$
    
    We prove by induction that $\mu_{t}^{\hat\pi} = \int_\nu \mu_t(\nu) d\rho(\nu)$ for all $t$.
    
    The result holds for $t = 0$ since $\mu_0$ is fixed. If this is true for $t$, then
    \begin{align*}
        \mu_{t+1}^{\hat\pi}(x) &= \sum_{x_t, a} p(x \mid x_t, a) \hat\pi(x_t, a) \mu_{t}^{\hat\pi}(x_t) \\
        &= \sum_{x_t, a} p(x \mid x_t, a) \int_{\nu'}\int_\nu \sum_\pi \frac{\nu(\pi) \mu_t^\pi(x_t) d\rho(\nu)}{\int_{\nu'} \mu_t(\nu') d\rho(\nu')} \mu_t(\nu') \pi(x_t, a) d\rho(\nu') \\
        &= \int_\nu \sum_\pi \nu(\pi) \underbrace{\sum_{x_t, a} p(x \mid x_t, a) \mu_t^\pi(x_t) \pi(x_t, a)}_{= \mu_{t+1}^\pi(x)} d\rho(\nu) \\
        &= \int_\nu \mu_{t+1}(\nu)(x) d\rho(\nu)\,,
    \end{align*}
    which concludes the induction argument.
\end{proof}

Finally, we will need to define what monotonicity, introduced by Lasry and Lions \cite{lasry2007mean} is:

\begin{definition}[Monotonicity]
    A mean field game is said to be monotonic if
    \[
        \langle \mu - \mu', r(\cdot, \mu) - r(\cdot, \mu') \rangle \leq 0\,,  \forall \mu, \mu'\in\mathcal{M}\;.
    \]
\end{definition}

We can now present cases where we can link the marginalization of a coarse correlated equilibrium with its optimality as a Nash equilibrium:

\begin{proposition}
    In monotonic games where the reward function is affine with respect to $\mu$, the marginalization of an $\epsilon$-Mean-Field-coarse correlated equilibrium, if it exists, is a $2\epsilon$-Mean-Field-Nash-equilibrium.
\end{proposition}
\begin{proof}
    Let $\rho$ be an $\epsilon$-MFCCE, and $\hat\pi$ its marginalization. 
    Let first observe that the monotonicity property implies:
    \begin{equation}\label{eq:monotonicity}
        \langle \mu - \mu', r(\cdot, \mu) \rangle \leq \langle \mu - \mu', r(\cdot, \mu') \rangle\,,  \forall \mu, \mu'\in\mathcal{M}\;.
    \end{equation}

    From there, we compute 
    \begin{align*}
        J(\pi, \mu^{\hat\pi}) - J(\hat\pi, \mu^{\hat\pi}) &= \langle \mu^\pi - \mu^{\hat\pi}, r(\cdot, \mu^{\hat\pi}) \rangle \\
        &= \sum_\nu \sum_{\nu'} \rho(\nu) \rho(\nu') \langle \mu^\pi - \mu(\nu), r\left(\cdot, \mu(\nu')\right) \rangle \\
        &= \sum_\nu \sum_{\nu'} \rho(\nu) \rho(\nu') \left( \langle \mu^\pi - \mu(\nu'), r\left(\cdot, \mu(\nu')\right) \rangle + \langle \mu(\nu') - \mu(\nu), r\left(\cdot, \mu(\nu')\right) \rangle \right) \\
        &\leq \sum_\nu \sum_{\nu'} \rho(\nu) \rho(\nu') \left( \epsilon + \langle \mu(\nu') - \mu(\nu), r\left(\cdot, \mu(\nu')\right) \rangle \right) \\
        &\leq \sum_\nu \sum_{\nu'} \rho(\nu) \rho(\nu') \left( \epsilon + \langle\mu(\nu') - \mu(\nu), r\left(\cdot, \mu(\nu)\right) \rangle\right) \\
        &\leq 2 \epsilon
    \end{align*}
    where the second line comes from the affine character of $r$ with respect to $\mu$, and $\hat\pi$ being the marginalization of $\rho$; the third and fifth lines come from $\rho$ being $\epsilon$-optimal, and the fourth line comes from Equation~\ref{eq:monotonicity}.
\end{proof}

\begin{remark}[Translation-invariance]
    We note that the above property also holds if a state-independent dependency on $\mu$ is added to the reward function. 
\end{remark}

\begin{remark}[Extension to $\epsilon$-monotonicity]
    If the game is $\epsilon'$-quasi-monotonic, i.e.     
    \[
        \langle \mu - \mu', r(\cdot, \mu) - r(\cdot, \mu') \rangle \leq \epsilon'\,,  \forall \mu, \mu'\in\mathcal{M}\;,
    \] 
    then the marginalization of an $\epsilon$-MFCCE, if it exists, is a $(2\epsilon+\epsilon')$-MFE.
\end{remark}

\begin{remark}[On the non existence of marginalization  in distribution-dependent settings]
    Consider the hole-trap game depicted in Figure~\ref{fig:hole_trap_game}. In this game, one initially chooses between going left or right. Once in the Left or Right node, the next state does not depend on the players' actions anymore: if every player is in the current node, then it transitions to its $+$ version (Left+ or Right+), otherwise all players are sent to the hole. 
    
    Taking a reward structure which makes Left+ and Right+ equivalent, and the Hole node very penalizing, we can take a Mean-Field Coarse Correlated Equilibrium which alternatively selects between Left and Right 50\% of the time.
    
    Its marginalized policy is a policy for which 50\% of players ends up in Left+ and 50\% of players ends up in Right+. However, this is strictly impossible, as this requires that 50\% of players be on the Left and Right nodes, which would automatically send all players to the hole, and none to Right+ and Left+. The marginalization of this correlated equilibrium is therefore impossible. %
\end{remark}

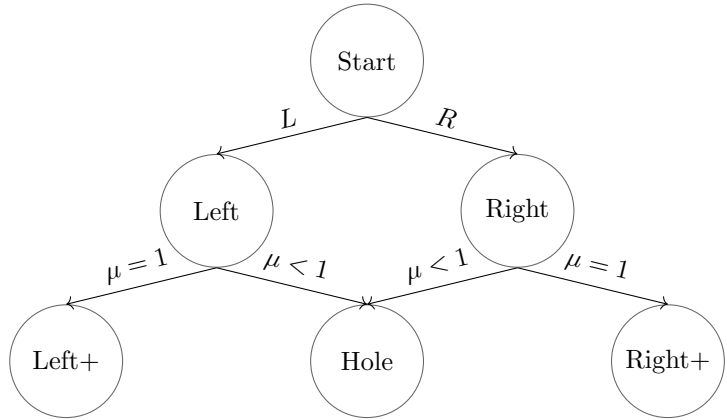
\begin{figure}[ht!]
    \centering
    \begin{tikzpicture}[
    state/.style={circle, draw=black!60, fill=white!5, thin, minimum size=15mm},]
    \node[state]      (start)            at (0, 0)                  {Start};

    \node[state]      (left)             at (-2, -2)                  {Left};
    \node[state]      (right)            at (2,  -2)                  {Right};

    \node[state]      (left_more)            at (-4, -4)                  {Left+};
    \node[state]      (hole)                 at (0, -4)                  {Hole};
    \node[state]      (right_more)           at (4, -4)                  {Right+};

    \draw[->] ($(start.south) $) -- ($ (left.north) $) node [midway, above, sloped] (TextNode) {$L$};
    \draw[->] ($(start.south) $) -- ($ (right.north) $) node [midway, above, sloped] (TextNode) {$R$};

    \draw[->] ($(left.south) $) -- ($ (left_more.north) $) node [midway, above, sloped] (TextNode) {$\mu = 1$};
    \draw[->] ($(left.south) $) -- ($ (hole.north) $) node [midway, above, sloped] (TextNode) {$\mu < 1$};

    \draw[->] ($(right.south) $) -- ($ (right_more.north) $) node [midway, above, sloped] (TextNode) {$\mu = 1$};
    \draw[->] ($(right.south) $) -- ($ (hole.north) $) node [midway, above, sloped] (TextNode) {$\mu < 1$};

    \end{tikzpicture}
    \caption{Hole-trap game}
    \label{fig:hole_trap_game}
\end{figure}

\subsection{Existence of (Coarse) Correlated Equilibria}\label{subsec:c_ce_existence}

We have not yet established conditions for correlated equilibria to exist. A set of conditions can be derived immediately from the fact that Nash equilibria can be used as correlated equilibria, as we proved in Proposition~\ref{prop:nash_to_ce}. Existence conditions for Nash equilibria, namely, continuity of the reward and dynamics functions with respect to $\mu$, hence also imply existence of correlated equilibria. Perhaps surprisingly, we find that the famous result derived by \citet{hart1989existence} that correlated equilibria (and therefore coarse correlated equilibria) exist in \emph{all finite $N$-player games} (\emph{i.e.} $N$-player games with finite $\mathcal{S}$, $\mathcal{A}$ and $\mathcal{T}$ but not necessarily with continuous reward and or dynamic functions) does not hold in Mean-Field games: Example~\ref{example:reward_for_the_few} shows a game where no exact correlated equilibrium exists. We summarize the existence relationships between different Mean-Field equilibria in Figure~\ref{fig:existence_implies_existence}, and visually represent them in Figure~\ref{fig:inclusion_visualization}. %

\begin{figure}[ht]
    \centering
    \begin{align*}
        \Aboxed{
        \text{MF-Nash}
        \raisebox{-3pt}{
            $\overset{\text{\normalsize$\overset{*}{\implies}$}}{\underset{\dagger}{\centernot\impliedby}}$
        }
        \text{MF-CE} \raisebox{-3pt}{
            $\overset{\text{\normalsize$\overset{**}{\implies}$}}{\underset{\dagger\dagger}{\centernot\impliedby}}$
        }
        \text{MF-CCE}
        }
    \end{align*}

    \begin{footnotesize}
        \begin{tabular}{c}
            $*$ : Proposition~\ref{prop:nash_to_ce}. \; $**$ : Proposition~\ref{prop:mfce_in_mfcce}. \; $\dagger$ : Example~\ref{example:ce_and_cce_no_nash}.  \;  $\dagger\dagger$ : Example \ref{example:cce_but_no_ce} \\
        \end{tabular}
    \end{footnotesize}

    \caption{Existence relationship between equilibrium concepts. $A \implies B$ means that the existence of $A$ implies the existence of $B$; whereas $A \centernot\impliedby B$ means that the existence of $B$ does not imply the existence of $A$.}
    \label{fig:existence_implies_existence}
\end{figure}

\begin{remark} Note that deriving a Mean-Field version of \citet{hart1989existence}'s proof of existence in the case of infinite players remains an open problem, due in part to the Mean-Field assumption that any finite set of players changing their policies would have no impact at all on the Mean-Field reward function - but \citet{hart1989existence}'s proof relies precisely on the fact that this isn't the case in their framework. 
\end{remark}

We begin by the following proposition, which will be the core argument for the existence proof.

With Proposition~\ref{prop:nash_to_ce} proven, we know that if the game admits a Nash equilibrium, then it admits a correlated equilibrium. Therefore, for correlated equilibria to exist, it suffices that Nash equilibria exist. A sufficient condition for their existence is the continuity of $r$ with respect to $\mu$. This has been proven in a very similar setting by \cite{muller2021learning}, for what they call a \emph{restricted game}. We straightforwardly adapt here their argument to our setting to prove Theorem~\ref{thm:cce-existence}, whose proof is postponed to Appendix~\ref{proof:existence_(c)ce}.

\begin{theorem}[(Coarse) Correlated equilibrium existence]\label{thm:cce-existence}
    If the reward function $r$ and the dynamics kernel function $p$ are continuous with respect to $\mu$, then the game admits at least one (coarse) correlated equilibrium.
\end{theorem}

Finally, we address the question of whether (coarse) correlated equilibria are always guaranteed to exist for Mean-Field games with finite state and action spaces. Theorem~\ref{thm:cce-existence} has already established the existence of such equilibria when the reward function $r$ is continuous in the population distribution $\mu$. The following example illustrates that equilibria do not necessarily exist when this continuity assumption does not hold, by highlighting a game where neither correlated nor coarse correlated equilibria exist !

\begin{example}[Reward for the few]\label{example:reward_for_the_few}
    We consider a stateless Mean-Field game with two actions, $a$ and $b$. The reward function is set up so as to reward the players who select the least popular action. More precisely, letting $\mu \in \mathscr{P}(\{a, b\})$ denote the population distribution over actions, we define
    \begin{align*}
        r(a, \mu) =
        \begin{cases}
            1 & \text{ if } \mu(a) < \nicefrac{1}{2} \\
            0 & \text{ if } \mu(a) = \nicefrac{1}{2} \\
            0 & \text{ if } \mu(a) > \nicefrac{1}{2}
        \end{cases}
     \, , \quad
        r(b, \mu) =
        \begin{cases}
            0 & \text{ if } \mu(a) < \nicefrac{1}{2} \\
            1 & \text{ if } \mu(a) = \nicefrac{1}{2} \\
            1 & \text{ if } \mu(a) > \nicefrac{1}{2}
        \end{cases} \, ,
    \end{align*}
    noting that in the case where the population is evenly split between actions $a$ and $b$, the players taking action $b$ are the one who are rewarded. Note that this payoff function is not continuous at $\mu = \nicefrac{1}{2} \delta_a + \nicefrac{1}{2} \delta_b$. Now, suppose $\rho$ is the correlation device of a coarse correlated equilibrium. The expected return of a representative player accepting the recommendation generated by this correlation device is
    \begin{align*}
        \int \left(\nu(a) \mathbbm{1}_{\{ \nu(a) < \nicefrac{1}{2} \}} + \nu(b) \mathbbm{1}_{\{ \nu(a) > \nicefrac{1}{2} \}} + \frac{1}{2} \mathbbm{1}_{\{ \nu(a) = \nicefrac{1}{2} \}}\right)  \rho( \mathrm{d} \nu) = \int \min(\nu(a), \nu(b)) \rho(\mathrm{d} \nu) \, .
    \end{align*}
    Now, the expected reward of a player that decides to deviate to action $a$ before seeing the recommendation generated by the correlation device is
        $\int \mathbbm{1} \{ \nu(a) < \nicefrac{1}{2} \} \rho(\mathrm{d} \nu)$ ,
    and similarly the expected reward for deviating to $b$ is
        $\int \mathbbm{1} \{ \nu(a) \geq \nicefrac{1}{2} \} \rho(\mathrm{d} \nu)$.
        
    In order for $\rho$ to encode a coarse correlated equilibrium, it must be the case that these expected rewards under deviation from the recommended play are no greater than the expected reward when following the recommendation:
    \begin{align*}
        \int \mathbbm{1} \{ \nu(a) < \nicefrac{1}{2} \} \rho(\mathrm{d} \nu) \, , \int \mathbbm{1} \{ \nu(a) \geq \nicefrac{1}{2} \} \rho(\mathrm{d} \nu) \leq \int \min(\nu(a), \nu(b)) \rho(\mathrm{d} \nu) \, .
    \end{align*}
    However, adding these two inequalities yields
    \begin{align*}
         1 \leq \int 2 \min(\nu(a), \nu(b)) \rho(\mathrm{d} \nu) \, .
    \end{align*}
    Since $2 \min(\nu(a), \nu(b)) \leq 1$, this inequality can only hold if $\nu(a) = \nu(b)$ $\rho$-almost surely, meaning that $\rho(\nicefrac{1}{2} \delta_a + \nicefrac{1}{2} \delta_b) = 1$. However, this is clearly not a coarse correlated equilibrium, since an individual player benefits from deviating to $b$ in this case. 
    
    We conclude no coarse correlated equilibrium (and hence no correlated equilibrium nor Nash equilibrium) exist for this Mean-Field game.
\end{example}

However, the following example below mitigates the previous one, by showing a game where, despite the lack of existence of Nash equilibria, correlated and coarse correlated equilibria do exist.

\begin{example}[Existence of Mean-Field games with a CE and a CCE but no Nash equilibrium.]\label{example:ce_and_cce_no_nash}
    Consider a Mean-Field variant of rock-paper-scissors. If there are at least two distinct actions in the population distribution, then rock wins, and scissor loses most. If there is only a single action taken in the population, then the payoffs to each individual player are as in the standard game. More precisely, when $\mu \in \mathscr{P}(\{\text{R}, \text{P}, \text{S}\})$ is not a Dirac, we have
    \begin{align*}
        r(\text{R}, \mu) = 1 \, , r(\text{P}, \mu) = -1, \, r(\text{S}, \mu) = -3 \, .
    \end{align*}
    Moreover, when $\mu$ is a Dirac, say $\delta_\text{R}$, we have the usual payoffs presented to the individual agent:
    \begin{align*}
        r(\text{R}, \delta_\text{R}) = 0 \, ,  r(\text{S}, \delta_\text{R}) = -1 \, , r(\text{P}, \delta_\text{R}) = 1 \, .
    \end{align*}
    Note that this reward function is not continuous at $\mu$ when $\mu$ is a Dirac. There is no Nash equilibrium in this game: a mixed policy $\pi$ cannot be a Nash equilibrium, since there is benefit in deviating to Rock, and a Dirac $\pi$ cannot be a Nash equilibrium, since there is benefit to an individual agent in deviating to the superior action. 
    
    Now, we argue that the correlation device $\rho \in \mathcal{P}(\Delta(\Pi))$ given informally by first selecting one of rock, paper, scissors uniformly at random, and then recommending this action to all players, is a coarse correlated equilibrium; mathematically, this is given by
    \begin{align*}
        \rho = \nicefrac{1}{3} \delta_{\delta_{\text{R}}} + \nicefrac{1}{3} \delta_{\delta_{\text{P}}} + \nicefrac{1}{3} \delta_{\delta_{\text{S}}} \, .
    \end{align*}
    The payoff when accepting this recommendation is $0$. The average payoff when deviating to a fixed action prior to seeing the recommendation is also $0$, hence we have a CCE. Note this is not a CE, since one can clearly deviate to a better action after seeing the recommendation. 
    
    However, the correlating device which alternates between everyone playing paper, and half the population playing paper while the other half plays rock is a mean field correlated equilibrium. More formally, 
    \[
        \rho = \nicefrac{1}{2} \delta_{\nicefrac{1}{2} \delta_{\text{P}} + \nicefrac{1}{2} \delta_{\text{R}}} + \nicefrac{1}{2} \delta_{\delta_{\text{P}}}
    \]
    is a Mean Field CE. To see this, let us consider each action's deviation incentive. When players are recommended to play rock, they always have an incentive to follow this recommendation. Players are never recommended to play scissors. Therefore, we must only examine the deviation payoffs from paper to rock on the one hand, and from paper to scissors on the other hand.
    \begin{align*}
        Payoff(S \mid P) &= 1 \, \mathbb{P}(\nu = \delta_{\text{P}} \mid P)  - 3 \, \mathbb{P}(\nu = \nicefrac{1}{2} \delta_{\text{P}} + \nicefrac{1}{2} \delta_{\text{R}} \mid P) = 1 \, \frac{2}{3} - 3 \, \frac{1}{3} = - \frac{1}{3}
    \end{align*}
    
    Similarly, we find that the expected deviation payoff when switching from paper to rock is $-\frac{1}{3}$. Finally, we see that the expected payoff when being recommended paper is $-\frac{1}{3}$. Players therefore never have an incentive to deviate from paper, and $\rho$ is thus a correlated equilibrium.
\end{example}

We have thereby provided an instance of a game where correlated and coarse correlated equilibria exist, but Nash equilibria do not. Hence, the set of correlated equilibria of all games is strictly larger than the set of Nash equilibria.

We also need to nuance the non-existence result: as we will see in Section~\ref{sec:regret_minimization_empirical_play}, although (coarse) correlated equilibria do not always exist as we have just shown, we can always find $\epsilon$-(coarse) correlated equilibria, with $\epsilon > 0$ as small as we like. We provide here a theorem stating this property, though its proof will be the entirety of Section~\ref{sec:regret_minimization_empirical_play}.

\begin{theorem}[Existence of $\epsilon > 0$-(coarse) correlated equilibria]\label{thm:existence_epsilon_c_ce}
    For all $\epsilon > 0$ small enough, there exists $\epsilon$-(coarse) correlated equilibria in \emph{all games}.
\end{theorem}
\begin{proof}
    All algorithms of Section~\ref{sec:regret_minimization_empirical_play} provably converge towards $\epsilon > 0$ (coarse) correlated equilibria, with $\epsilon \rightarrow 0$. 
\end{proof}

To illustrate Theorem~\ref{thm:existence_epsilon_c_ce}, we remark that in Example~\ref{example:reward_for_the_few}, although no exact equilibrium exists, one can easily design \emph{e.g.} an $\epsilon$-Nash equilibrium for all $\nicefrac{1}{2} > \epsilon > 0$. Indeed, taking $\nu_a = (\nicefrac{1}{2} + \epsilon) \delta_{a} + (\nicefrac{1}{2} - \epsilon) \delta_{b}$ and $\nu_b = (\nicefrac{1}{2} - \epsilon) \delta_{a} + (\nicefrac{1}{2} + \epsilon) \delta_{b}$, $\rho = \nicefrac{1}{2} \delta_{\nu_a} + \nicefrac{1}{2} \delta_{\nu_b} $ is a $4 \epsilon$-Nash equilibrium. However, a single policy will always be $> \nicefrac{1}{2}$-exploitable, thereby showing that $\epsilon$-Nash equilibria do not always exist for $\epsilon$ small enough. \\

At last, we exhibite a game where the existence of Mean Field CCE does not imply the existence of Mean Field CE. 

\begin{example}\label{example:cce_but_no_ce}
    Let consider the following (stateless) mean-field variant of rock-paper-scissors. Each member of the population selects an action from $\{\text{R}, \text{P}, \text{S}\}$, and the payoff structure is specified as:
    \begin{itemize}
        \item If $\mu(\text{P}) > 0$ (that is, a non-zero proportion of the population play paper), then $r(\text{S}, \mu) = 1$, $r(\text{P}, \mu) = 0$,  $r(\text{R}, \mu) = -1$.
        \item If $\mu(\text{P}) = 0$ but $\mu(\text{S}) > 0$ (that is, almost no one plays paper, but a non-zero proportion play scissors), then $r(\text{R}, \mu) = 1$, $r(\text{S}, \mu) = 0$, $r(\text{P}, \mu) = -1$.
        \item Finally, if $\mu = \delta_{\text{R}}$, then $r(\text{P}, \mu) = 1$, $r(\text{R}, \mu) = 0$, $r(\text{S}, \mu) = -1$.
    \end{itemize}
    Is there a correlated equilibrium in this game? No: if a player is ever recommended $\text{P}$, they realise that the sampled recommendation distribution puts mass on $\text{P}$, so they would benefit from deviating to $\text{S}$. So no MFCE can ever recommend $\text{P}$. But now similarly, any player recommended $\text{S}$ could similarly benefit from deviating to $\text{R}$, so $\text{S}$ cannot be recommended in a MFCE. This leaves only one possibility: that the MFCE always recommends $\text{R}$, but this is clearly also not an MFCE.
    
    We now claim that
        $\rho = \nicefrac{1}{3} \delta_{\delta_\text{S}} +\nicefrac{1}{3} \delta_{\delta_\text{P}} + \nicefrac{1}{3} \delta_{\delta_\text{R}}$ 
    is an MFCCE for this game. Following the recommendation leads to a payoff of $0$. However, playing a fixed action also clearly leads to an expected payoff of $0$, hence we have an MFCCE.
\end{example}

\subsection{Uniqueness of (Coarse) Correlated Equilibria}

The uniqueness of correlated and coarse correlated equilibria is less crucial than it is for Nash equilibria: indeed, when a game has a unique Nash, there can be no equilibrium selection problem, which is why Nash unicity is of interest for us. In counterpart, correlated equilibria do not suffer from equilibrium selection problems due to the correlation device's role in coordinating agents. However, we identified an important situation where correlated and coarse correlated equilibria are unique: the presence of a dominant strategy, which we define as follows: 

\begin{definition}[Strictly-dominant strategy]
    A strategy $\pi^* \in \bar\Pi$ is said to be strictly dominant if%
    \[
        J(\pi^*, \mu) > J(\pi, \mu) \,, \qquad \forall \pi \in \Pi, \, \mu \in \Delta(\mathcal{X})^\mathcal{T} \,.
    \]    
\end{definition}

Indeed, if a correlated or coarse-correlated equilibrium were to recommend any other action than the dominant one, the players would all have an incentive to play that dominant strategy instead, as we show here:

\begin{proposition}[(Coarse) Correlated equilibria uniqueness]
    If there exists a strictly dominant strategy in the game, then the game only admits a unique coarse correlated equilibrium, and therefore a unique correlated equilibrium, which only recommends $\nu^* \in \Delta(\Pi)$ such that $\pi(\nu^*) = \pi^*$.
\end{proposition}
\begin{proof}
    Let $\rho$ be a coarse correlated equilibrium of a game with strictly dominant strategy $\pi^*$.
    
    Then $\forall \nu \in \Delta(\Pi)$ such that $\pi(\nu) \neq \pi^*$,
    \[
        J(\pi^*, \mu(\nu)) > J(\pi(\nu), \mu(\nu))
    \]
    since $\pi^*$ is a strictly dominant strategy.
    
    Therefore, unless $\rho$ only recommends $\nu \in \Delta(\Pi)$ such that $\pi(\nu) = \pi^*$, $\pi^*$ is always a strictly-value-increasing deviation. For $\rho$ to be a coarse correlated equilibrium, it must therefore only recommend $\nu^* \in \Delta(\Pi)$ such that $\pi(\nu^*) = \pi^*$. Since there only exists one such $\nu^*$ (Otherwise two different deterministic policies would have equal state-action distribution, which is impossible), equilibrium uniqueness follows. Of course, equilibrium properties derive directly from the optimality of $\pi^*$.
\end{proof}

We provide an example of such a situation in the Mean-Field Prisoner's Dilemma:

\begin{example}[Mean-Field Prisoner's Dilemma]
    Consider the two-action normal-form mean-field game with actions C(ooperate) and D(efect), and reward function
    \begin{align*}
        r(C, \mu) &= 3 \mu(C) - \mu(D)\;, \\
        r(D, \mu) &= 4 \mu(C) - 0 \mu(D) 
    \end{align*}
    This game has a strictly dominating action, D.
\end{example}

\subsection{Connection to the Notion of Correlated Equilibrium Derived by Campi and Fischer } \label{subsec:campi_Fischer}

A notion of correlated solution in Mean-Field games has already been introduced by Campi and Fischer \cite{campi2020correlated}. The main difference between their framework and ours is that they chose to work with (policy, distribution flow) pairs $(\pi, \mu)$ instead of population recommendations, which led to difficulties in adapting their equilibrium concept from Mean-Field settings to $N$-player settings. In counterpart, the concept of population distribution adapted seamlessly to $N$-player games and allows us to provide deeper theoretical properties such as optimality bounds in the next sections of this work.

we investigate how our definition of MFCE coincides with their notion of correlated solution. Following our notations, $\mathcal{T}:=\{0,\ldots,T\}$ with $T\in\mathbb{N}$ in their framework, while the state space $\mathcal{X}$ and the action space $\mathcal{A}$ are finite. The set $\Pi$ is the finite set of deterministic strategies, that is
\[
    \Pi=\{\pi:\{0,\ldots,T-1\}\times\mathcal{X}\rightarrow\mathcal{A}\}\;.
\]

In our approach, the correlation device $\rho\in\mathcal{P}(\Delta(\Pi))$ is introduced in order to generate different distributions of policies over the full population and synchronise hereby the players actions. In the approach detailed in \cite{campi2020correlated}, the synchronisation between the representative player and the population is viewed as a constraint to which their correlation device must conform. In more details, their correlation device analogue recommends directly the representative individual policy $\pi\in\Pi$ together with the population Mean-Field flow $\mu\in\mathcal{M}=\Delta(\mathcal{X})^\mathcal{T}$. This gives rise to the notion of correlation flow  $\bar\rho$, a distribution over $\Pi\times\mathcal{M}$. The main drawback of this approach is that, written as such, there is no guarantee that the policies generated by a correlating flow $\bar\rho$ induces a Mean-Field flow consistent with the one sampled by $\bar\rho$. This additional property in \cite{campi2020correlated} identifies to a consistency condition on the correlating flow $\bar\rho$, which can be adapted from the one described in Definition 4.1 in \cite{campi2020correlated} and rewritten as follows.
\begin{definition}[Consistent correlating flow \cite{campi2020correlated}]
    A \textbf{consistent correlating flow} is a distribution $\bar\rho$ over $\Pi\times\mathcal{M}$ that satisfies the following consistency condition:
    \begin{eqnarray}\label{eq:consistency}
        \mu\left( \frac{\bar\rho(\cdot,\mu)}{\sum\limits_{\pi\in\Pi} \bar\rho(\pi,\mu) }  \right) = \mu \,, \quad \mbox{ for any $\mu$ in the support of $\bar\rho$\;.} 
    \end{eqnarray}
\end{definition}

The consistency condition indicates that, for a potentially recommended Mean-Field flow $\mu\in\mathcal{M}$, the population recommendation induced by the correlation flow $\bar\rho$ conditioned by $\mu$, that is 
\begin{equation} \label{eq:def_rho_bar_p} \bar\rho_p(\cdot\mid\mu):=\frac{\bar\rho(\cdot,\mu)}{\sum_{\pi\in\Pi} \bar\rho(\pi,\mu) },
\end{equation}  generates its own Mean-Field flow $\mu(\bar\rho_p(\cdot\mid\mu))$ that coincides with $\mu$. This condition is naturally inspired by the structure of Nash equilibria definition in MFGs and is required in order to properly define the notion of correlated solution of Mean Field Games in \cite{campi2020correlated}. Nevertheless, directly providing recommendations to the population when manipulating (C)CEs allows to  automatically satisfy this condition. This is the approach naturally followed by our notion of correlation device. 

We are now in position to establish a one to one correspondence between consistent correlation flows $\bar\rho$ considered in \cite{campi2020correlated} and correlation devices $\rho\in\mathcal{P}(\Delta(\Pi))$ as introduced in Definition \ref{def_Correlation_device}.

\begin{theorem}\label{prop:equivalence_of_def}
    For any consistent correlating flow $\bar\rho$ on $\Pi\times\mathcal{M}$, there exists a correlation device $\rho\in\mathcal{P}(\Delta(\Pi))$ that generates the same distribution over $\Pi\times\mathcal{M}$. The opposite result holds similarly. 
\end{theorem}

The derivation of this property first requires the following result. 
\begin{lemma}\label{lemma:d_mu_convex}
    For any $\mu\in\mathcal{M}$, the set $\mathcal{D}_\mu = \{ \nu \in \Delta(\Pi) \mid \mu(\nu) = \mu \}$ is convex.
\end{lemma}

\begin{proof}
    Define $\mu^{\pi}(\nu)$ the state distribution flow of agents playing $\pi$ when the population distribution is $\nu\in\mathcal{V}$, and $p^\pi(x' \mid x, \mu) = \sum\limits_{a \in \mathcal{A}} \pi(x', a) p(x' \mid a, x, \mu)$ the proportion of agents going from $x'$ to $x$ when playing $\pi$, under population distribution $\mu$.

    The state distribution of an agent playing $\pi$ in a population playing $\pi(\nu)$, %
    is by definition
    \[
    \mu_{t+1}^{\pi}(\nu)(x) = \sum_{x'} \mu_t^{\pi}(\nu)(x') p^\pi(x' \mid x, \mu_{t}(\nu)). 
    \]

    Fix $\mu \in \mathcal{M}$, $\nu_1, \, \nu_2 \in \mathcal{D}_\mu$ and define $\nu = \alpha \nu_1 + (1-\alpha) \nu_2$ with $\alpha \in [0, 1]$. %
    
    We will prove by induction on $t \in \mathcal{T}$ that, for each $\pi \in \Pi$, $\mu_{t}^{\pi}(\nu) = \mu_{t}^{\pi}(\nu_1) = \mu_{t}^{\pi}(\nu_2)$ and $\mu_t(\nu) = \alpha \mu_t(\nu_1) +  (1-\alpha) \mu_t(\nu_2)$ so that $\mu_t(\nu) = \mu_t$. We first observe that this is satisfied for $t = 0$, since the initial distribution is fixed. 
    
    Suppose that the result holds at time $t$. Then
    \begin{align*}
        \mu_{t+1}^{\pi}(\nu)(x) &= \sum_{x'} \mu_t^{\pi}(\nu)(x') p^\pi(x' \mid x, \mu_{t}(\nu)) \\
        &= \sum_{x'} \mu_t^{\pi}(\nu_1)(x') p^\pi(x' \mid x, \mu_{t}) \\
        &= \sum_{x'} \mu_t^{\pi}(\nu_1)(x') p^\pi(x' \mid x, \mu_{t}(\nu_1)) \\
        &= \mu_{t+1}^{\pi}(\nu_1)(x) \\
        \text{and similarly} \\
        &= \sum_{x'} \mu_t^{\pi}(\nu_2)(x') p^\pi(x' \mid x, \mu_{t}(\nu_2)) \\
        &= \mu_{t+1}^{\pi}(\nu_2)(x)\;.
    \end{align*}
    
    Besides, we observe that 
    
    \begin{align*}
        \mu_{t+1}(\nu)(x) &= \sum_\pi \mu_{t+1}^{\pi}(\nu)(x) \nu(\pi) \\
        &= \alpha \sum_\pi \mu_{t+1}^{\pi}(\nu)(x) \nu_1(\pi) + (1-\alpha) \sum_\pi \mu_{t+1}^{\pi}(\nu)(x) \nu_2(\pi) \\
        &= \alpha \sum_\pi \mu_{t+1}^{\pi}(\nu_1)(x) \nu_1(\pi) + (1-\alpha) \sum_\pi \mu_{t+1}^{\pi}(\nu_2)(x) \nu_2(\pi) \\
        &= \alpha \mu_{t+1}(\nu_1)(x) + (1-\alpha) \mu_{t+1}(\nu_2)(x) \\
        &= \mu_{t+1}(x)\;.
    \end{align*}
    
    The property is initialized and hereditary, which concludes the proof: $\mathcal{D}_\mu$ is convex.
\end{proof}

With Lemma~\ref{lemma:d_mu_convex} available, we are now in position to prove Theorem~\ref{prop:equivalence_of_def}, \textit{i.e.} the equivalence between our correlated equilibrium representation and the one presented in \cite{campi2020correlated}.
\begin{proof}[Proof of Theorem \ref{prop:equivalence_of_def}]
    Take any consistent correlation flow $\bar \rho$ in the sense of Campi and Fischer \cite{campi2020correlated}. It can be decomposed as a distribution $\bar\rho_{m}$ over $\mathcal{M}$ combined with a conditional distribution $\bar\rho_{p}$ over $\Pi$:
\[
\bar\rho(\pi,\mu)= \bar\rho_{p}(\pi\mid\mu)\bar\rho_{m}(\mu) \,.
\]
To any $\mu\in\mathcal{M}$, we associate the induced population distribution $\nu(\mu):=\bar\rho_{p}(.\mid\mu)$. Because the correlating flow is consistent, the Mean-Field flow induced by $\nu(\mu)$ coincides with $\mu$ - i.e. $\mu(\nu(\mu))=\mu$. Therefore the distribution over $\Pi\times\mathcal{M}$ induced by $\bar\rho$ is similar to the one generated by the following correlation device %
\begin{eqnarray*}
d\rho(\nu)&=& \int_{\mu\in\mathcal{M}} \mathbf{1}_{\bar\rho_p(.\mid\mu)=\nu} \; d\bar\rho_m(\mu)\,.
\end{eqnarray*}  
Indeed, the distribution over $\Pi\times\mathcal{M}$ generated by $\rho$ is given for $(\pi,\mu)\in\Pi\times\mathcal{M}$ by
\begin{eqnarray*}
    \int_{\nu\in\mathcal{V}}\mathbf{1}_{\mu(\nu)=\mu}\nu(\pi)d\rho(\nu) 
    &=& 
    \int_{\nu\in\mathcal{V}}\int_{\mu'\in\mathcal{M}}\mathbf{1}_{\bar\rho_p(.\mid\mu')=\nu}\mathbf{1}_{\mu(\nu)=\mu}\nu(\pi)d\bar\rho_m(\mu')\\
    &=& 
    \int_{\nu\in\mathcal{V}}\mathbf{1}_{\bar\rho_p(.\mid\mu)=\nu}\nu(\pi)d\bar\rho_m(\mu)\\ 
    &=& 
    d\bar\rho_p(\pi\mid\mu)d\bar\rho_m(\mu)\\ 
    &=& 
    d\bar\rho(\pi,\mu)\;, 
\end{eqnarray*}
where we used the consistency condition  in the second equality.

On the other hand, take now a correlation device $\rho\in\mathcal{P}(\Delta(\Pi))$. It induces on $\Pi\times\mathcal{M}$ the following correlation flow: 
\begin{eqnarray*}
d\bar\rho(\pi,\mu)&=& \int_{\nu\in\mathcal{V}} \mathbf{1}_{\mu(\nu)=\mu} \; \nu(\pi)d\rho(\nu)\;.  
\end{eqnarray*}

It remains to verify that the induced correlation flow is indeed consistent. By construction, we have that $d\bar\rho(\pi, \mu) \neq 0 \iff \exists \nu, \, \mu(\nu) = \mu, \, \nu(\pi) \neq 0, \, d\rho(\nu) \neq 0$.

Whenever there exists a unique $\nu\in\mathcal{D}$ such that $\mu(\nu) = \mu, \, d\rho(\nu) \neq 0$, then directly $\bar\rho_p(.\mid\mu) = \nu$ and  the consistency condition holds.

Otherwise, $\bar\rho_p(.\mid\mu)$ is a mixture of several population recommendations $\nu\in\mathcal{V}$ such that $\mu(\nu) = \mu$. But Lemma \ref{lemma:d_mu_convex} ensures that the set $\mathcal{D}_\mu = \{ \nu \in \Delta(\Pi) \mid \mu(\nu) = \mu \}$ is convex, so that $\mu(\bar\rho_p(.\mid\mu)) = \mu$ and the correlation flow induced by $\rho$ is also consistent.

\end{proof}

As our notion of correlating device connects now naturally to the notion of correlation flow considered in \citet{campi2020correlated}, we are now in position to draw connections between our notion of Correlated equilibria and the notion of correlated solution described in \cite{campi2020correlated}. Before doing so, let's turn to the definition of correlated solution introduced by \citet{campi2020correlated} which requires the following notion of expected return when using a deviation mapping $u\in\mathcal{U}_{CE}$ in the presence of a correlating flow $\bar\rho$.

\begin{definition}[Correlated solution, Definition 4.1 in \cite{campi2020correlated}]
A consistent correlation flow $\bar\rho$ is a correlated solution to the Mean Field Game whenever the following optimality condition holds:
\begin{eqnarray*}
     \mathbb{E}_{(\pi,\mu)\sim\bar\rho}\left[J(u(\pi), \mu) - J(\pi, \mu) \right] \leq 0\;,\mbox{ for any }u\in\mathcal{U}_{CE}\;.
\end{eqnarray*}
\end{definition}

\begin{proposition}
A correlating flow $\bar\rho$ is a correlated solution in the Campi-Fischer \cite{campi2020correlated} sense if and only if a corresponding correlation device $\rho$ - which exists by Proposition~\ref{prop:equivalence_of_def} - is a Mean Field Correlated Equilibrium according to our definition.
\end{proposition}

\begin{proof} Let $\bar\rho$ be a consistent correlation flow generating the same distribution over $\Pi\times\mathcal{M}$ than the correlation device $\rho\in\mathcal{C}$, see Proposition \ref{prop:equivalence_of_def}. The consistent correlation flow $\bar\rho$ is a correlated solution of the MFG if and only if
\begin{eqnarray*}
    \mathbb{E}_{(\pi,\mu)\sim\bar\rho}\left[J(u(\pi),\mu)-J(\pi,\mu)\right] \le 0\;, \qquad u\in\mathcal{U}_{CE}\;.
\end{eqnarray*}
On the other hand, the correlation device is a correlated equilibrium if and only if 
\begin{eqnarray*}
    \mathbb{E}_{\pi\sim\nu,\nu\sim\rho}\left[J(u(\pi),\mu(\nu))-J(\pi,\mu(\nu))\right] \le 0\;, \qquad u\in\mathcal{U}_{CE}\;.
\end{eqnarray*}
The proof is complete, recalling that $\bar\rho$ and $\rho$ induce the same distribution on $\Pi\times\mathcal{M}$.
\end{proof}

\subsection{Homogeneous Correlated Equilibrium Characterization}\label{sec:hom_ce} \label{subsec:homogeneous_ce_props}

 Homogeneous correlation device presented in Definition \ref{def:homogeneousCE} are such that any agent knows what any other agent is playing, since everyone is playing the same policy. Therefore, a homogeneous $\epsilon$-correlated equilibrium should intuitively only recommend $\epsilon'$-Nash equilibria, with some $\epsilon' \geq 0$ to be specified. In this section, we clarify the relationship between Nash equilibria and homogeneous correlated equilibria.

We first start with linking the components of a homogeneous $\epsilon$-Mean-Field correlated equilibrium to $\epsilon$-nash-equilibria.

\begin{proposition} \label{prop:homogeneous_ce_nash}
    Let $\epsilon \geq 0$ and $\rho$ be a homogeneous $\epsilon$-MFCE. Then all $\pi \in \Pi$ atoms of $\rho$ are $\frac{\epsilon}{\rho(\pi)}$-MFE.
\end{proposition}

\begin{proof}
    Let $\epsilon \geq 0$, $\rho$ be a homogeneous $\epsilon$-MFCE and $\pi^* \in \Pi$ such that $\rho(\pi^*) > 0$. 
    
    Since $\rho$ is an $\epsilon$-homogeneous Mean-Field correlated equilibrium, we have 
    \begin{align*}
        \int_{\pi\in\bar\Pi} (J(u(\pi), \mu^\pi) - J(\pi, \mu^\pi)) \rho(d\pi) \leq \epsilon \qquad \forall \pi' \in \Pi,\, u: \bar\Pi \rightarrow \bar\Pi.
    \end{align*}
    
    Given $\pi^* \in \bar\Pi$ an atom of $\rho$ and $\pi' \in \bar\Pi$ any policy, we select $u$ such that $\forall \pi \in \bar\Pi, \, \pi \neq \pi^*, \, u(\pi) = \pi$ and $u(\pi^*) = \pi'$. Plugging this in the above equation, we get

    \begin{align*}
        \rho(\pi^*) (J(\pi', \mu^{\pi^*}) - J(\pi^*, \mu^{\pi^*})) &\leq \epsilon \\
        J(\pi', \mu^{\pi^*}) - J(\pi^*, \mu^{\pi^*}) &\leq \frac{\epsilon}{\rho(\pi^*)}\;.
    \end{align*}
    Which means that $\pi^*$ is an $\frac{\epsilon}{\rho(\pi^*)}$-MFE.
\end{proof}

We now know that the components of homogeneous $\epsilon$-correlated equilibria are necessarily $\epsilon'$-Nash equilibria. This shows that, at least in Mean-Field games, only homogeneous correlation devices recommending solely Nash equilibria can have no $\Phi$-regret.

Finally, we answer the converse question - if a homogeneous correlated equilibrium only recommends $\epsilon$-Mean-Field Nash equilibria, is it an $\epsilon$-Mean-Field correlated equilibrium?

\begin{proposition}\label{prop:mfce_convex_mixture_eps_eq}
    Any homogeneous correlation device recommending only $\epsilon$-Mean-Field Nash equilibria is an $\epsilon$-MFCE.
\end{proposition}
\begin{proof}
Let $\rho$ be a homogeneous correlation device with support only over $\epsilon$-Nash equilibria.

 For all $u \in \mathcal{U}^h_{CE}$, we compute
 \begin{align*}
        \mathbb{E}_{\pi \sim \rho} [J(u(\pi), \mu^\pi) - J(\pi, \mu^\pi)] &\,=\, \int\limits_{\pi \in \bar\Pi}  \rho(d\pi) \left(J(u(\pi), \mu^\pi) - J(\pi, \mu^\pi)\right) \\
        &= \int\limits_{\pi \in \epsilon\text{-Nash}} \rho(d\pi) \underbrace{\left(J(u(\pi), \mu^\pi) - J(\pi, \mu^\pi)\right)}_{\leq \epsilon} \\
        &\leq \epsilon,
 \end{align*}
 
 hence $\rho$ is an $\epsilon$-MFCE.
\end{proof}

\section{Connections Between \texorpdfstring{$N$}{\textit{N}}-Player and Mean-Field Equilibria}\label{sec:n_player_mf_eq}

In this section, we explore the connections between $N$-player and Mean-Field equilibria. We first properly define how to use mean-field equilibria in $N$-player games in section \ref{subsec:mean_field_n_player_corresponding}. We then build  in section \ref{subsec:n_to_mean_field} on the correspondence between our approach and the one in \citet{campi2020correlated} to investigate the behavior of $N$-player equilibria as N tends to infinity. We show that they converge towards Mean-Field equilibria. Finally, in Section \ref{subsec:mean_field_(c)ce_n_player}, we derive a key practical property by computing optimality bounds whenever using a mean-field equilibrium in an $N$-player game.

\subsection{Mean-Field Games to \texorpdfstring{$N$}{\textit{N}}-Player Games}\label{subsec:mean_field_n_player_corresponding}

Before we use Mean-Field correlation devices in $N$-player games, we must first define how we can do so. 

The population recommendation framework is very straightforward to use in $N$-player games : just like in Mean-Field games, we first sample a population recommendation $\nu \sim \rho$, and then, for each player, sample a policy from $\nu$. Since there are now only N players, sampling N policies from $\nu$ yields $\nu_N \in \Delta_N(\Pi)$, a random variable with a law determined by $\nu$ and N. This means that we can view $\rho$ as a distribution over $\Delta_N(\Pi)$, \emph{i.e.} $\rho \in \mathcal{P}(\Delta_N(\Pi))$: $\rho$ is an $N$-player correlation device !

When sampling an $N$-player population recommendation $\nu_N$ from a Mean-Field population recommendation $\nu \in \Delta(\Pi)$, we will use the abusive notation $\nu_N \sim \nu$. The discussions above yield the following property:

\begin{proposition}[Mean-Field to $N$-player equilibria]\label{prop:mf_to_n_cce}
Taking $\rho$ a Mean-Field correlation device, and $\rho_N$ its $N$-player version, we have that
\[
    \mathbb{E}_{\nu\sim\rho, \, \nu_N \sim \nu, \, \pi \sim \nu_N}\left[ \mathbb{E}_{\mu^N \sim \mu(\nu)} \left[ J(u(\pi), \mu_N) \right]  \right] = \mathbb{E}_{\nu_N \sim \rho_N, \, \pi \sim \nu_N}\left[ J(u(\pi), \mu_N) \right] \quad \forall u : \Pi \rightarrow \Pi.
\]    
\end{proposition}

However, a Mean-Field correlation device can only be used in an $N$-player game if it makes sense to do so, \emph{i.e.} if the $N$-player game corresponds to the Mean-Field game. We define this notion more precisely:

\begin{definition}[Corresponding $N$-player game]
    Given a Mean-Field game with payoff function $J$ and deterministic policies $\Pi$, its corresponding $N$-player game is the $N$-player game where all N players play the Mean-Field game as independent agents, and the Mean-Field population distribution is replaced by the $N$-players' distribution.
    
    In other words, taking $\mu_N$ the state distribution of all N players, replace $r(x, a, \mu)$ by $r(x, a, \mu_N)$ and $p(x' \mid x, a, \mu)$ by $p(x' \mid x, a, \mu_N)$.
\end{definition}

To rephrase the definition, players play a modified version of the Mean-Field game where their distribution flow is considered to be the game's Mean-Field flow as far as rewards and dynamics are concerned.

\subsection{\texorpdfstring{$N$}{\textit{N}}-Player to Mean-Field Equilibria}\label{subsec:n_to_mean_field}

Given the equilibrium equivalence shown in section~\ref{subsec:campi_Fischer} between \citet{campi2020correlated}'s concepts and ours, we inherit their convergence proofs going from $N$-player games to the Mean-Field case: any sequence of $N$-player (coarse) correlated equilibria converges towards a Mean-Field (coarse) correlated equilibrium as N increases, given some conditions.

\begin{theorem}[$N$-player CEs to Mean-Field CEs]
    Let $(\rho_N)_{N}$ be a sequence of $\epsilon_N$-correlated equilibria in the corresponding $N$-player game. 
    If the reward function and state transition functions are continuous in $\mu$, and if $\epsilon_N \rightarrow 0$, then the limit of the sequence $(\rho_N)_{N}$ is a Mean-Field correlated equilibrium.
\end{theorem}
\begin{proof}
    The result follows from a direct application of Theorem 6.1 in \cite{campi2020correlated} and we only need to verify that the 5 required assumptions (A1)-(A5) \cite{campi2020correlated} in are satisfied.
    Assumption (A1) holds since the state transition function is continuous in $\mu$. Assumption (A2) follows from the continuity of the reward function with respect to $\mu$. Assumption (A3) and (A4) are valid as $(\rho_N)_{N}$ is a sequence of $\epsilon_N$-correlated equilibria and  $\epsilon_N \rightarrow 0$. Finally, Assumption (A5) holds by virtue of $(\rho_N)_N$ being correlated equilibria of the corresponding $N$-player game, so that $\mu^N_0 = \mu_0$ for all N. 
\end{proof}

\begin{theorem}[$N$-player CCEs to Mean-Field CCEs]
    A similar statement holds for coarse correlated equilibria.
\end{theorem}
\begin{proof}
    The proof follows the line of argument of the one of Theorem 6.1 in \cite{campi2020correlated} and simply requires to restrict the set of deviations $\mathcal{U}_{CE}$ to the more restrictive $\mathcal{U}_{CCE}$.
\end{proof}

We now explore the converse of these properties: which population behavior is induced by plugging  Mean-Field equilibria policy in $N$-player games?

\subsection{Mean-Field Equilibria in \texorpdfstring{$N$}{\textit{N}}-Player Games} \label{subsec:mean_field_(c)ce_n_player}

Spending resources computing Mean-Field equilibria can be reasonably justified whenever we can use these equilibria in real-world situations, where, typically, agents aren't infinite, but present in very large numbers. It is therefore useful to be sure that our Mean-Field-generated equilibria work reasonably well in the large-$N$ $N$-player games of interest. The purpose of this section is to provide conditions for which using a Mean-Field $\epsilon$-(coarse) correlated equilibrium in $N$-player games provides an $N$-player $\mathcal{O} \left (\epsilon + \frac{1}{\sqrt{N}} \right )$-(coarse) correlated equilibrium !

We first consider in Theorem~\ref{thm:optimality_no_dynamics_dependence} the simple situation, where transitions do not depend on $\mu$, then ramp up to transition functions that are Lipschitz with respect to $\mu$, first with $\rho$ as sums of diracs in Theorem~\ref{thm:full_optimality_theorem_discrete_rho}, then for all correlating devices in Theorem~\ref{thm:full_optimality_theorem_continuous_rho}.

\begin{theorem}\label{thm:optimality_no_dynamics_dependence}
    Let $\rho$ be an $\epsilon \geq 0$-Mean-Field (coarse) correlated equilibrium. If
    \begin{itemize}
        \item the reward function is $\gamma_r$-lipschitz in $\mu$ for the $L_2$ norm, and
        \item the transition function does not depend on $\mu$,
    \end{itemize}
    then $\rho$ is an $\epsilon + \frac{2 \gamma_r T \left( 1 + \sqrt{\frac{1}{2 N}} \right) }{\sqrt{N}} $-(coarse) correlated equilibrium of the corresponding $N$-player game. 
\end{theorem}

\begin{proof} We consider correlated equilibria, but dealing with coarse correlated ones simply requires to replace the set of deviations $\mathcal{U}_{CE}$ by $\mathcal{U}_{CCE}$.
    An $\epsilon$-Mean-Field correlated equilibrium $\rho$ in the Mean-Field's corresponding $N$-player game is characterized by, according to Proposition~\ref{prop:mf_to_n_cce}, %
    \[
    \mathbb{E}_{\nu \sim \rho, \pi \sim \nu}\left[\mathbb{E}_{\mu^N \sim \mu(\nu)} \left[J(u(\pi), \mu^N_{-\pi, u(\pi)}) - J(\pi, \mu^N)\right]\right] \leq \epsilon\,, \quad \forall u \in \mathcal{U}_{CE}\,.
    \]
    
    Fix $u \in \mathcal{U}_{CE}$. The outline of the proof is the following: We first control  the difference between $J(u(\pi), \mu^N_{-\pi, u(\pi)})$ and $J(u(\pi), \mu^N)$, and then bound the difference between $J(u(\pi), \mu^N)$ and $J(u(\pi), \mu(\nu))$, both using the Lipschitz property of $r$, and therefore of $J$. 
    
    We write $\delta_\mu$ the indicator function of a player's position and time: if a given player $i$ is in state $x$ at time $t$, then $\delta^i_\mu(x, t) = 1$, and it is 0 for all other states at time $t$. Directly, we have that $\mu^N = \frac{1}{N} \sum_i \delta^i_\mu$. We overload the notation to write $\delta^\pi_\mu$ the indicator function of the location of a given player playing $\pi$. Observe that, since $\mu^N = \sum_i \delta^i_\mu$, we can separate this sum following
    
    \[
        \mu^N = \frac{1}{N} \sum_{i \neq j} \delta^i_\mu + \frac{1}{N} \delta^j_\mu
    \]
    
    \emph{i.e.}
    \[
        \mu^N = \mu^N_{-j} + \frac{1}{N} \delta^j_\mu.
    \]
    
    Since this is true for all $j$, we can exclude the player which deviated from playing $\pi$ to $u(\pi)$ from the sum:
    
    \[
    \mu^N_{-\pi, u(\pi)} = \frac{N-1}{N} \mu^{N-1}_{-\pi} + \frac{1}{N} \delta^{u(\pi)}_\mu 
    \quad \mbox{and} \quad 
    \mu^N = \frac{N-1}{N} \mu^{N-1}_{-\pi} + \frac{1}{N} \delta^\pi_\mu \,.
    \]
    
    Therefore
    \[
    \mu^N_{-\pi, u(\pi)} - \mu^N = \frac{1}{N} \left( \delta^{u(\pi)}_\mu - \delta^\pi_\mu \right)\,.
    \]
    
    We will now prove that $J$ is $T \gamma_r$-Lipschitz w.r.t. $\mu$. Take $\mu_1$, $\mu_2 \in \mathcal{M}$ and $\pi \in \Pi$.
    
    \begin{align*}
        J(\pi, \mu_1) - J(\pi, \mu_2) &= \sum\limits_{t \in \mathcal{T}} \sum\limits_{x \in \mathcal{X}} \mu^\pi_t(x) \left( r^\pi(x, \mu_{1, t}) - r^\pi(x, \mu_{2, t}) \right) \\
        &\leq \sum\limits_{t \in \mathcal{T}} \sum\limits_{x \in \mathcal{X}} \mu^\pi_t(x) \gamma_r \| \mu_{1, t} - \mu_{2, t} \|_2 \\
        &\leq  \gamma_r \sum\limits_{t \in \mathcal{T}} \| \mu_{1, t} - \mu_{2, t} \|_2 \\
        &\leq  \gamma_r \sum\limits_{t \in \mathcal{T}} 1 \sqrt{\sum_x (\mu_{1, t}(x) - \mu_{2, t}(x))^2} \\
        &\leq  \gamma_r \sqrt{\sum\limits_{t \in \mathcal{T}} \sum_x (\mu_{1, t}(x) - \mu_{2, t}(x))^2} \sqrt{\sum\limits_{t \in \mathcal{T}} 1^2} \\
        &\leq \sqrt{T} \gamma_r \| \mu_1 - \mu_2 \|_2
    \end{align*}
    where the first line is true because $\mu^\pi$ doesn't depend on $\mu_1$ or $\mu_2$, since dynamics are independent of distribution.
    
    Since $J$ is $\sqrt{T} \gamma_r$-Lipschitz w.r.t. $\mu$, we deduce 
    \[
    \lvert J(u(\pi), \mu^N_{-\pi, u(\pi)}) - J(u(\pi), \mu^N) \rvert \leq \frac{\sqrt{T} \gamma_r}{N} \| \delta^{u(\pi)}_\mu - \delta^\pi_\mu \|_2\,.
    \]
    Because the number of states in the game is finite, $\| \delta^{u(\pi)}_\mu - \delta^\pi_\mu \|_2$ is bounded.
    The maximum value of this difference is reached in the hypothetical situation where $\pi$ and $u(\pi)$ never reach the same state at the same time. Hence, we have  
    \begin{align*}
        \| \delta^{u(\pi)}_\mu - \delta^\pi_\mu \|_2 &= \sqrt{\sum_{t} \underbrace{\sum_s \left(\delta^{u(\pi)}_\mu(s, t) - \delta^\pi_\mu(s, t) \right)^2}_{\leq 2}} \; \leq  \; \sqrt{2T} \,,
    \end{align*}
    so that 
    \begin{equation}\label{eq:mu_n_to_mu_nu}
        \lvert J(u(\pi), \mu^N_{-\pi, u(\pi)}) - J(u(\pi), \mu^N) \rvert \leq \frac{T \gamma_r \sqrt{2}}{N}.
    \end{equation}

    Note that the above is true for any realization of the random variables $\delta^\pi_\mu$.
    
    We have bounded the difference between $J(u(\pi), \mu^N_{-\pi, u(\pi)})$ and $J(u(\pi), \mu^N)$; now let us bound the difference between $J(\pi, \mu^N)$ and $J(\pi, \mu(\nu))$ for all $\pi$, as we will use the following equality later on:
    
    \begin{equation}\label{eq:main_eq_theorem_n_player}
        \mathbb{E}_{\mu^N \sim \mu(\nu)} \left[J(\pi, \mu^N)\right] = \mathbb{E}_{\mu^N \sim \mu(\nu)} \left[J(\pi, \mu^N) - J(\pi, \mu(\nu)) \right] + J(\pi, \mu(\nu)).        
    \end{equation}
    
    We start with the Lipschitz property of J: 
    \[
    \lvert J(\pi, \mu^N) - J(\pi, \mu(\nu)) \rvert \leq T \gamma_r \| \mu^N - \mu(\nu) \|_2.
    \]
    
    By the Jensen inequality, we have
    \[
        \mathbb{E}[\| \mu^N - \mu(\nu) \|_2] 
        = \mathbb{E}\left[ \sqrt{\sum_{x, t} \left|\mu^N(x, t) - \mu(\nu)(x, t) \right|^2} \right] 
        \leq \sqrt{\sum_{x, t} \mathbb{E}\left[ \left|\mu^N(x, t) - \mu(\nu)(x, t) \right|^2 \right]}\,.
    \]
    
    Recall that $\mu^N$ is the Mean-Field flow resulting from an N players independently sampling \emph{and} playing their policies from $\nu$. Since the policy sampling \emph{and} the state sampling via policy playing are independent of other players, the \emph{expected} distribution of all players is the Mean-Field distribution of a population playing $\nu$, \emph{i.e.} $\mathbb{E}[\mu^N] = \mu(\nu)$ (Though their \emph{actual} state distribution will of course typically differ from their expected state distribution).
    
    Therefore $\forall x \in \mathcal{X}, \, t \in \mathcal{T}, \; \mathbb{E}[\mu^N(x, t)] = \mu(\nu)(x, t)$ and therefore, $\forall x \in \mathcal{X}, \, t \in \mathcal{T}$, 
    
    \[ \mathbb{E}\left[ \left|\mu^N(x, t) - \mu(\nu)(x, t) \right|^2 \right] = \text{Var}(\mu^N(x, t)) \]
    
    The term $\mu^N(x, t) = \frac{1}{N} \sum\limits_{i=1}^N \delta^{i}_\mu(x, t)$ is the empirical mean of N independent Bernoulli random variables with mean $\mu(\nu)(x, t)$, and therefore has variance $\frac{1}{N} \mu(\nu)(x, t)(1-\mu(\nu)(x, t))$.
    
    We notice that whatever the value of $\mu$, $\mu (1-\mu) \leq \mu$ since $1-\mu \leq 1$.
    Therefore $\forall t \leq T, \, \forall \mu \in \mathcal{M}$, 
    \[
    \sum_x \mu(x, t) (1-\mu(x, t)) \leq \sum_x \mu(x, t) = 1 
    \]
    
    Which yields 
    \[
    \mathbb{E}\left[\| \mu^N - \mu(\nu) \|_2\right] \leq \sqrt{\frac{T}{N}},  
    \]
    
    and finally gives us
    \begin{equation}\label{eq:n_player_j_inequality_to_nu}
        \mathbb{E}\left[\lvert J(u(\pi), \mu^N) - J(u(\pi), \mu(\nu)) \rvert \right] \leq \frac{T \gamma_r}{\sqrt{N}}.
    \end{equation}

    Plugging this property into Equation~\ref{eq:main_eq_theorem_n_player}, we obtain
    \begin{align*}
        \mathbb{E}_{\mu^N \sim \mu(\nu)} \left[J(\pi, \mu^N)\right] = \mathbb{E}_{\mu^N \sim \mu(\nu)} \left[J(\pi, \mu^N) - J(\pi, \mu(\nu)) \right] + J(\pi, \mu(\nu)) \\        
        J(\pi, \mu(\nu)) - \frac{T \gamma_r}{\sqrt{N}} \leq \mathbb{E}_{\mu^N \sim \mu(\nu)} \left[J(\pi, \mu^N)\right] \leq J(\pi, \mu(\nu)) + \frac{T \gamma_r}{\sqrt{N}},        
    \end{align*}
    where the second line comes from Equation~\ref{eq:n_player_j_inequality_to_nu} and the fact that $-\mathbb{E}[\lvert X \rvert] \leq \mathbb{E}[X] \leq \mathbb{E}[\lvert X \rvert]$.

    We recall Equation~\ref{eq:mu_n_to_mu_nu}:

    \[ \mathbb{E}_{\nu \sim \rho, \pi \sim \nu}\left[\mathbb{E}_{\mu^N \sim \mu(\nu)} \left[J(u(\pi), \mu^N_{-\pi, u(\pi)}) \right]\right] \leq \frac{T \gamma_r \sqrt{2}}{N} + \mathbb{E}_{\nu \sim \rho, \pi \sim \nu}\left[\mathbb{E}_{\mu^N \sim \mu(\nu)} \left[J(u(\pi), \mu^N) \right]\right] \quad \forall u: \Pi \rightarrow \Pi.
    \]

    Combining all these equations, we have
    \begin{align*}
        \mathbb{E}_{\nu \sim \rho, \pi \sim \nu}\Big[\mathbb{E}_{\mu^N \sim \mu(\nu)} &\left[J(u(\pi), \mu^N_{-\pi, u(\pi)}) - J(\pi, \mu^N)\right]\Big] \\
        &\leq \mathbb{E}_{\nu \sim \rho, \pi \sim \nu}\Big[\mathbb{E}_{\nu_N \sim \nu} \left[J(u(\pi), \mu(\nu_N)) - J(\pi, \mu(\nu_N))\right]\Big] + \frac{T \gamma_r \sqrt{2}}{N} \\
        &\leq \mathbb{E}_{\nu \sim \rho, \pi \sim \nu}\left[ J(u(\pi), \mu(\nu)) - J(\pi, \mu(\nu)) \right] + \frac{T \gamma_r \sqrt{2}}{N} + 2 \frac{\gamma_r T}{\sqrt{N}} \\
        &\leq \epsilon + \frac{2 \gamma_r T \left(1 + \sqrt{\frac{1}{2 N}} \right)}{\sqrt{N}}\,.
    \end{align*}
    where the last inequality comes from the fact that $\rho$ is an $\epsilon$-Mean-Field (coarse) correlated equilibrium.
\end{proof}

\begin{remark} 
We see that in this case, equilibrium approximation accuracy decreases with the time horizon, however, it does so at speed $\mathcal{O}(T)$ - surprisingly, the inaccuracy is not just not exponential in the time horizon, but it is linear ! It also linearly depends on $\gamma_r$: the lower the Lipschitz coefficient, the more accurate the approximation. There is no dependency on state space size $\lvert \mathcal{X} \rvert$, however we infer that it is hidden within the $L_2$-Lipschitz condition.
\end{remark} 

We now tackle the more complex case of $\mu$-dependent transitions. In $N$-player games, sampling recommendations from a Mean-Field correlation device $\rho$ induces a sampling noise: when $\rho$ has sampled population distribution $\nu$, although the N players sample their distributions from $\nu$, their population distribution will not be equal to $\nu$.

Moreover, the $N$-players' action choices, and the game's intrinsic stochasticity will also render players' trajectories different from their expected values.

This adds a third expectation in the computation of a (coarse) correlated equilibrium's payoff, which we abusively write $\mu^N \sim \mu(\nu)$ as the distribution of N players who sampled their policies from $\nu$.

Finally, we write $\mu^N_{\pi}$ the distribution flow associated with all $N_\pi$ players playing policy $\pi \in \Pi$. Similarly, we write $\mu_\pi(\nu)$ the Mean-Field flow associated with players playing policy $\pi$ when the population distribution is $\nu$. %

We first tackle the distribution-dependent dynamics in the particular case where $\rho$ is a finite sum of diracs.

\begin{theorem}\label{thm:full_optimality_theorem_discrete_rho}
    Let $\rho$ be an $\epsilon \geq 0$-Mean-Field (coarse) correlated equilibrium. If
    \begin{itemize}
        \item the reward and transition functions are lipschitz in $\mu$ for the $L_2$ norm, and
        \item $\rho$ is a finite sum of diracs,
    \end{itemize}
    then $\rho$ is an $\epsilon + O \left(\frac{1}{\sqrt{N}} \right)$ (coarse) correlated equilibrium of the corresponding $N$-player game.
\end{theorem}
\begin{proof}
    We provide here a proof outline to introduce the reader to the main arguments in the full proof, which can be found in Appendix~\ref{proof:full_optimality_theorem_discrete_rho}.
    
    Using Lipschitz arguments, we bound the (coarse) correlated equilibrium equation in the $N$-player game by the same equation in the mean-field game, with the addition of a distance term between the mean-field distribution and the $N$-player distribution, \emph{e.g.}
    
    \begin{align*}
        \mathbb{E}_{\mu^N \sim \mu(\nu)} \left[J(u(\pi), \mu^N_{-\pi, u(\pi)}) - J(\pi, \mu^N)\right]& \leq J(u(\pi), \mu(\nu)) - J(\pi, \mu(\nu)) \\
        & + \gamma_r \mathbb{E}_{\mu^N \sim \mu(\nu)} \left[ \| \mu(\nu) - \mu^N \|_2 + \| \mu(\nu) - \mu^N_{-\pi, u(\pi)} \|_2 \right].        
    \end{align*}
        
    The rest of the proof is focused on finding bounds for the $\mathbb{E}_{\mu^N \sim \mu(\nu)} \left[ \| \mu(\nu) - \mu^N \|_2 \right]$ term, which can be straightforwardly extended to the $\mathbb{E}_{\mu^N \sim \mu(\nu)} \left[ \| \mu(\nu) - \mu^N_{-\pi, u(\pi)} \|_2 \right]$ term.
    
    The dependence of $p$ on $\mu$ forces us to consider every sampled policy's state distributions separately, as they influence one another: it is difficult otherwise to know policy state distributions, and thus which mixed policy is being played at which state and time.
    
    To bound the difference between $\mu^N$ and $\mu(\nu)$, we proceed by induction over game time using a lemma which reconciles per-policy correctness (Closeness to $\mu_\pi$ for every $\pi$) with global correctness (Closeness to $\mu$ for every $\mu^N$).
    
    Finally, we conclude the proof by summing over the finite number of atoms of $\rho$ to recover the first expectation.
\end{proof}

We see that we still keep a bound in $ \mathcal{O} \left( \frac{1}{\sqrt{N}} \right) $ for this more complex case, though deriving it was much more difficult. Unfortunately, allowing $\rho$ to be any type of distribution degrades the bounds, given our proof technique, to $\mathcal{O} \left( \frac{1}{N^{\frac{1}{4}}} \right)$, the square root of the former one, as we see in the following theorem. %

\begin{theorem}\label{thm:full_optimality_theorem_continuous_rho}
    Let $\rho$ be an $\epsilon \geq 0$-Mean-Field (coarse) correlated equilibrium. If
    \begin{itemize}
        \item the reward and transition functions are lipschitz in $\mu$ for the $L_2$ norm
    \end{itemize}
    then $\rho$ is an $\epsilon + O \left(\frac{1}{N^{\frac{1}{4}}} \right)$ (coarse) correlated equilibrium of the corresponding $N$-player game.
\end{theorem}
\begin{proof}
    The line of arguments is similar to that of Theorem~\ref{thm:full_optimality_theorem_discrete_rho}, with a few alterations to the end, that are developed in Appendix~\ref{proof:full_optimality_theorem_continuous_rho}.
    
    Indeed, the end of the proof of Theorem~\ref{thm:full_optimality_theorem_discrete_rho}   requires summing over a finite number of values of $\frac{\rho(\nu)}{\sqrt{\nu_m N}}$, which is always finite and is indeed $\mathcal{O} \left( \frac{1}{\sqrt{N}} \right)$. However, when $\rho$ is not finite, it could well be that it puts mass on a sequence for which $\nu_m$ tends to 0, and this bound therefore diverges.
    
    To counter this, we introduce a new scalar, $\alpha$,that we use to filter out policies with selection probabilities $\leq \alpha$. We prove that, while $\alpha \leq \mathcal{O} \left( \frac{1}{\sqrt{N}} \right)$, the policies that weren't filtered out will still have their state distribution $\mu_\pi \leq \mathcal{O} \left( \frac{1}{\sqrt{N}} \right)$, and so will the global state distribution. Once this is proven, we search for the best value of $\alpha$ leading to the best bound on $N$. We find that $\alpha = \frac{1}{\sqrt{N}}$ yields the bound of $\mathcal{O} \left( \frac{1}{N^{\frac{1}{4}}} \right)$.
\end{proof}

\begin{remark} We note that these proofs are much more difficult than for Nash equilibria because of Mean-Field correlated equilibria's induced stochasticities: they provide deterministic policy recommendations, and in $N$-player games, the number of players playing a given policy \emph{is a random variable}. What this means is that we cannot consider that the whole population plays a policy $\pi(\nu)$, which greatly complexifies the proof.
\end{remark}

It is unclear whether the bound $\mathcal{O}\left( \frac{1}{N^{\frac{1}{4}}}\right)$ is ever reached, or if non-discrete MF(C)CEs have tighter bounds; we leave this question for future work.

However, before closing this section, we would like to make the remark that, since Nash equilibria can be cast as correlated equilibria, the above bounds also apply to Nash equilibria. Surprisingly, this is the first result of the sort of which we are aware in the fully discrete setting:

\begin{remark}[Mean-Field Nash Equilibrium $N$-player $\epsilon$-optimality]
    This development, since it applies to coarse correlated and correlated equilibria, also straightforwardly applies to Nash equilibria by Proposition~\ref{prop:nash_to_ce}, which, given the conditions of the above theorem, are thus $\epsilon = \mathcal{O}\left( \frac{1}{\sqrt{N}} \right)$-Nash equilibria in their corresponding $N$-player games since Proposition~\ref{prop:cce_to_nash} can be adapted to $N$-player games). 
\end{remark}

To the best of our knowledge, this is the first time that optimality bounds have been provided for Mean-Field Nash equilibria's optimality in $N$-player games for the fully discrete setting. %

\section{Regret Minimization and Empirical Play}\label{sec:regret_minimization_empirical_play}

There are strong connections between game-theoretic equilibria and regret minimisation in online learning. A core result \cite{blumregret} states that if all players follow a regret-minimizing algorithm to select their strategy, then the (learning-)time average of their joint behaviour converges to the set of coarse correlated equilibria.
This connection provides a means of computing approximate equilibria which has been fundamental to recent advances in the state-of-the-art of games such as heads-up no-limit poker \cite{brown2017libratus,moravvcik2017deepstack}.

Regret minimization has surprisingly been understudied in the Mean Field Games literature. In this section, we describe a corresponding connection between regret-minimizing algorithms and Mean-Field coarse correlated equilibria, which serve as the basis for deriving convergence results of learning equilibria in Section~\ref{sec:learning}.

\subsection{Empirical Play}

A continuous-time learning algorithm generates a continuous-time, measurable sequence of policies $(\pi_s)_{0 \leq s \leq t}$. A correlation device is extracted from this sequence by recommending a policy from a uniformly-selected moment of play: it is the \emph{empirical play}.

\begin{definition}[Empirical Play]
    The empirical play $\hat\rho \in \mathcal{P}(\Delta(\Pi))$ of the sequence of policies $(\pi_s)_{0 \leq s \leq t}$ is the correlation device resulting from uniformly recommending each deterministic component of one stochastic policy selected at random among $(\pi_s)_{0 \leq s \leq t}$.
\end{definition}

More formally, in the continuous case, this yields
\begin{align*}
    \hat{\rho}(A) = \frac{1}{t}\int_0^t \mathbbm{1}\{ \nu \in A \mid \pi_s = \pi(\nu) \} \mathrm{d}s \, .
\end{align*}

In the discrete case, this yields
\begin{align*}
    \hat{\rho}(\nu) = \frac{1}{t} \sum_{s=1}^t \delta_{ \pi_s = \pi(\nu) } \, ,
\end{align*}

The motivation for introducing the notion of empirical play is that several key results in this section establish that if each member of a population that played the sequence of policies $(\pi_s)_{0 \leq s \leq t}$ is relatively happy with their choice of policies in hindsight, in a sense made precise below, then the corresponding empirical play correlation device is an approximate equilibrium for the Mean-Field game under consideration. 

To evaluate how close to optimal the empirical play is, we define policy alterations which characterize the expected deviation payoffs when one follows it.

\begin{definition}[Policy Alterations]
    The set of \textbf{Policy Alterations} $\mathcal{U}_{A}$ is the set of functions $\bar\Pi \rightarrow \bar\Pi$ such that $u \in \mathcal{U}_A$ is a policy alteration if there exists a function $u' \in \mathcal{U}_{CE}$ such that for all $\bar\pi = \sum\limits_{\pi \in \Pi} \alpha_\pi \pi$,  $u(\bar\pi) = \sum\limits_{\pi \in \Pi} \alpha_\pi u'(\pi)$
    
    Informally, a policy alteration of $\bar\pi \in \bar\Pi$ is a function that swaps around deterministic policies' mass in the composition of $\bar\pi$.
    
    The set of \textbf{Coarse Policy Alterations} $\mathcal{U}_{CA}$ is the subset of $\mathcal{U}_A$ composed only of constant functions. 
\end{definition}

The remainder of this section is devoted to formalising the relationship between regret minimization and both correlated equilibria and coarse correlated equilibria, and the question of how such sequences of policies can be generated algorithmically is addressed in Section~\ref{sec:learning}.

\subsection{External Regret and Coarse Correlated Equilibria}

Consider a representative agent in a Mean-Field game, using policy $\pi_s$ at time $s$, against a population distribution $\mu_s$. The cumulative return of the agent over a time interval $[0,t]$ is given by
\begin{align*}
    \int_0^t J(\pi_s, \mu^s) \mathrm{d} s \, .
\end{align*}
A natural question to consider is how better the agent could have done in hindsight by sticking with a fixed policy $\pi$ throughout the interval $[0,t]$, in contrast to using the sequence $(\pi_s)_{0 \leq s \leq t}$. The increase in payoff that the agent could have received is referred to as the \emph{regret} of not having played $\pi$. The \emph{external regret} of a policy sequence codifies the worst-case regret against a fixed policy.

\begin{definition}[External regret]
    Given a sequence of population distributions $(\mu_s)_{0 \leq s \leq t}$, the \emph{external regret} of a policy sequence $(\pi_s)_{0 \leq s \leq t}$ is given by  
    \begin{align*}
        \extreg((\pi_s)_{0 \leq s \leq t}, (\mu_s)_{0 \leq s \leq t}) = \sup_{\pi \in \Pi} \int_0^t J(\pi, \mu^s) \mathrm{d} s - \int_0^t J(\pi_s, \mu^s) \mathrm{d} s \, .
    \end{align*}
    
    Alternatively, an equivalent definition is 
    \begin{align*}
        \extreg((\pi_s)_{0 \leq s \leq t}, (\mu_s)_{0 \leq s \leq t}) = \sup_{u \in \mathcal{U}_{CA}} \int_0^t J(u(\pi_s), \mu^s) \mathrm{d} s - \int_0^t J(\pi_s, \mu^s) \mathrm{d} s \, ,
    \end{align*} 
    where the equivalence is immediate when equating $\mathcal{U}_{CA}$ to $\Pi$.
\end{definition}

For a bounded reward function $J$, an immediate upper bound on the external regret of a policy sequence $(\pi_s)_{0 \leq s \leq t}$ given a population sequence $(\mu_s)_{0 \leq s \leq t}$ is $O(t)$. Of particular interest are methods for selecting policy sequences $(\pi_s)_{s \geq 0}$ in the presence of a population sequence $(\mu_s)_{s \geq 0}$ such that the external regret grows as $o(t)$; such a policy sequence is said to be \emph{no-regret}, or \emph{regret-minimising}. This interest, in the context of game theory, stems from the close connection between external regret and coarse correlated equilibria; both notions encode the value of deviation to a fixed policy in certain circumstances. 

This connection is well-known in non-Mean-Field game theory, and forms the basis for many algorithms for computing equilibria. The following result makes this connection precise in the case of Mean-Field games, and serves as a key motivation for the use of regret-minimisation algorithms for computing coarse correlated equilibria in Mean-Field games, following similar results in non-Mean-Field game theory.

\begin{proposition}\label{epsilon_cce_regret}
Let $\epsilon > 0$ and $(\pi_s)_{0 \leq s \leq t}$ be a sequence of policies. Then the following two propositions are equivalent.
\begin{enumerate}
    \item $\frac{1}{t} \extregret((\pi_s)_{0 \leq s \leq t},(\mu^{\pi_s})_{0 \leq s \leq t}) \leq \epsilon$
    \item The Empirical Play of $(\pi_s)_{0 \leq s \leq t}$ is an $\epsilon$-Mean Field Coarse Correlated Equilibrium.
\end{enumerate}
\end{proposition}

\begin{proof}
    Let select $\epsilon > 0$ and $(\pi_s)_{0 \leq s \leq t}$, and name $\hat\rho$ the correlation device recommending the empirical play of $(\pi_s)_{0 \leq s \leq t}$. Observe that
\begin{align*}
    \frac{1}{t} \extregret((\pi_s)_{0 \leq s \leq t},(\mu^{\pi_s})_{0 \leq s \leq t})  
    &= \sup_{\pi' \in \Pi} \mathbb{E}_{\nu \sim \hat\rho, \pi \sim \nu}[J(\pi', \mu(\nu)) -  J(\pi, \mu(\nu))]\,,   
\end{align*} 
following the definition of $\hat\rho$ as recommending the empirical play uniformly: each $\nu$ recommended by $\rho$ is derived from uniformly recommending $\nu_{\pi_s}$ over $s$. We deduce that
\begin{align*}
    \frac{1}{t} \extregret((\pi_s)_{0 \leq s \leq t},(\mu^{\pi_s})_{0 \leq s \leq t})  
    &= \sup_{u \in \mathcal{U}_{CCE}} \mathbb{E}_{\nu \sim \hat\rho, \pi \sim \nu}[J(u(\pi), \mu(\nu)) -  J(\pi, \mu(\nu))]\,,
\end{align*} 
hence providing the connection with the coarse correlated equilibrium characterization stated in Definition~\ref{def:MFCCE}.

Hence, $\hat \rho$ is an $\epsilon$-Mean-Field Coarse Correlated Equilibrium if and only if the Average External Regret of $(\pi_s)_{0 \leq s \leq t}$ is smaller than $\epsilon$.
\end{proof}

The correspondence between $\epsilon$-external regret and $\epsilon$-coarse correlated equilibria is now established. However, in general, algorithms never really reach $0$ regret, and we now wonder: Does an asymptotically no-regret algorithm indeed get closer to the set of coarse correlated equilibria as it minimizes regrets, or could it actually remain ``away'' from this set? The following proposition proves that no-external-regret learners do approach the set of CCEs!

\begin{proposition}\label{prop:cce_conv_regret}
Let $(\pi_s)_{0 \leq s \leq t}$ be such that $\lim_{t \rightarrow \infty} \frac{1}{t} \extregret((\pi_s)_{0 \leq s \leq t}, (\mu^{\pi_s})_{0 \leq s \leq t}) = 0$, and assume the reward function $r$ is bounded and the set of coarse correlated equilibria is non-empty. Then the empirical play of $\pi$, $\hat\rho_\pi^t$, converges to the set of coarse correlated equilibria $\mathcal{C}$, \emph{i.e.} $\inf\limits_{\rho_0 \in \mathcal{C}} d_{\mathcal{W}_2}(\hat\rho_\pi^t, \rho_0) \rightarrow 0$, where $d_{\mathcal{W}_2}$ is the Wasserstein-2 distance.
\end{proposition}

\begin{proof}

First, notice that, since we are in a finite-time, finite-state setting, $r$ being bounded implies directly that $J$ is bounded.
Let denote by $\mathcal{C}_\epsilon$ is the set of $\epsilon$-CCE, while $\mathcal{C}$ is the set of CCE.

We will prove by contradiction that 
\begin{equation} \label{eq:holds}
    \forall \alpha > 0, \, \exists \epsilon > 0, \forall \rho \in \mathcal{C}_\epsilon, \; \inf\limits_{\rho_0 \in \mathcal{C}} d_{\mathcal{W}_2}(\rho, \rho_0) < \alpha\,.
\end{equation}

Let us suppose that 
\begin{equation} \label{eq:contradict}
    \exists \alpha > 0, \, \forall \epsilon > 0, \exists \rho \in \mathcal{C}_\epsilon, \; \inf\limits_{\rho_0 \in \mathcal{C}} d_{\mathcal{W}_2}(\rho, \rho_0) \geq \alpha\,.
\end{equation}

We take a sequence $(\rho_n)_n$ such that
\[
    \forall n, \, \rho_n \in \mathcal{C}_{\frac{1}{2^n}}, \, d_{\mathcal{W}_2}(\rho_n, \rho_0) \geq \alpha\,.
\]

Correlation devices are distributions over distributions over $\lvert \Pi \rvert$ elements. The set of distributions over $\lvert \Pi \rvert$ elements is the set of vectors in $\mathbb{R}_+^{\lvert \Pi \rvert}$ which sum to 1. It is compact as a closed and bounded subset of $\mathbb{R}^{\lvert \Pi \rvert}$. All measures over the set of population distributions are therefore, by definition, tight. Since their set is tight, Theorem 5.1 in \citet{billingsley1999convergence} indicates that the set of correlation devices is relatively compact. 

Hence, there exists a subsequence of $(\rho_n)_n$, denoted $(\rho'_n)_n$ converging weakly towards a point $\bar\rho$. Since $\mathbb{R}^{\lvert \Pi \rvert}$ is Polish, $(\rho'_n)_n$ converges towards $\bar\rho$ with respect to the Wasserstein distance $d_{\mathcal{W}_2}$. 

We note that the deviation-payoff function $\rho \rightarrow \max\limits_{\pi \in \Pi} \int_\nu \rho(d \nu) \left( J(\pi, \mu(\nu)) - J(\pi(\nu), \mu(\nu)) \right)$ is continuous (It is the max over the integral over the finite set $\Pi$ of continuous functions of $\rho$) provided $J$ is bounded.
Hence, since $\rho^n\in\mathcal{C}_{\frac{1}{2^n}}$, $\bar\rho$ must be a coarse correlated equilibrium. This contradicts \eqref{eq:contradict} so that  \eqref{eq:holds} holds.

Moreover, \eqref{eq:holds} directly implies that $$\forall \alpha > 0, \, \exists \epsilon_\alpha > 0, \, \forall \epsilon \leq \epsilon_\alpha, \, \forall \rho \in \mathcal{C}_\epsilon, \; \inf\limits_{\rho_0 \in \mathcal{C}} d_{\mathcal{W}_2(\rho, \rho_0)} < \alpha\;,$$
since the sets $(\mathcal{C}_\epsilon)_{\epsilon \leq \epsilon_\alpha}$ are included into the set of $\epsilon_\alpha$-coarse correlated equilibria.

We define a sequence $\alpha_n$ which converges to 0, and a subsequence $\phi(n)$ such that, $\forall n,\, \phi(n)$ is the first $n$ from which $(\hat\rho_\pi^n)_n$ is an $\epsilon_{\alpha_n}$-CCE and after which it never becomes a worse equilibrium. We know that $\forall t \geq \phi(n)$, $\hat\rho_\pi^t$ is also an $\epsilon_{\alpha_n}$-CCE, and therefore $\forall t \geq \phi(n)$, $ \inf\limits_{\rho_0 \in \mathcal{C}} d_{\mathcal{W}_2}(\hat\rho_\pi^t, \rho_0) < \alpha_n$ as well.

Thus $\forall \epsilon > 0$, $\exists N \geq 0, \, \forall t \geq N, \, \inf\limits_{\rho_0 \in \mathcal{C}} d_{\mathcal{W}_2}(\hat\rho_\pi^t, \rho_0) < \epsilon$, and thus $\inf\limits_{\rho_0 \in \mathcal{C}} d_{\mathcal{W}_2}(\hat\rho_\pi^t, \rho_0) \rightarrow 0$.

\end{proof}

\subsection{Swap Regret and Correlated Equilibria}

A second naturally arising question is: given the output of an algorithm over several timesteps, had the agent swapped its policies for other policies (\textit{i.e.}, every time it was recommended to play $\pi_1$, it chose to play $\pi_2$ instead), could they have received a higher payoff ? This is the definition of Swap Regret~\cite{gordon2008regret, blumregret}: given a policy alteration $u$, what is the difference between our received payoff and the maximal payoff, were we to have altered our play using the best possible $u$? %

More formally, 
\begin{definition}[Swap Regret]
    Given a sequence of policies $(\pi_s)_{1 \leq s \leq t}$ and a sequence of population distributions $(\mu_s)_{1 \leq s \leq t}$, we define swap regret as 
    
    \[
        \swapreg((\pi_s)_{1 \leq s \leq t}) = \sup_{u \in \mathcal{U}_A} \int_s J(u(\pi_s), \mu_s) - J(\pi_s, \mu_s) ds
    \] 
\end{definition}

\begin{proposition}
Let $\epsilon > 0$ and $(\pi_s)_{0 \leq s \leq t}$ be a sequence of policies. Then the following two propositions are equivalent.
\begin{enumerate}
    \item $\frac{1}{t} \swapreg((\pi_s)_{0 \leq s \leq t},(\mu^{\pi_s})_{0 \leq s \leq t}) \leq \epsilon$\,;
    \item The Empirical Play of $(\pi_s)_{0 \leq s \leq t}$ is an $\epsilon$-Mean Field Correlated Equilibrium.
\end{enumerate}
\end{proposition}

\begin{proof}
Let $\epsilon > 0$ and $(\pi_s)_{0 \leq s \leq t}$ a history of policies. We begin this proof by noting that for all $0 \leq s \leq t$, all policies $\pi$ and Mean-Field flow $\mu$, 

\[
J(\pi, \mu) = \sum\limits_{\tilde\pi \in \Pi} \nu_\pi(\tilde\pi) J(\tilde\pi, \mu),
\]
where we recall that $\forall \bar\pi \in \bar\Pi$, $\pi(\nu_{\bar\pi}) = \bar\pi$.

We thus have
\begin{align*}
    \frac{1}{t} \swapreg((\pi_s)_{0 \leq s \leq t}, (\mu^{\pi_s})_{0 \leq s \leq t}) &= \frac{1}{t} \sup_{u \in \mathcal{U}_A} \int_s J(u(\pi_s), \mu^{\pi_s}) - J(\pi_s, \mu^{\pi_s}) ds \\
    &= \frac{1}{t} \sup_{u \in \mathcal{U}_{CE}} \int_s \sum\limits_{\pi \in \Pi} \nu_{\pi_s}(\pi) \left( J(u(\pi), \mu^{\pi_s}) - J(\pi, \mu^{\pi_s})  \right) ds \\
    &= \sup_{u \in \mathcal{U}_{CE}} \mathbb{E}_{\pi \sim \nu_{\pi_s}, \, \nu_{pi_s} \sim Uniform((\pi_t)_t)} [J(u(\pi), \mu^{\pi_s}) - J(\pi, \mu^{\pi_s})] \\
    &= \sup_{u \in \mathcal{U}_{CE}} \mathbb{E}_{\pi \sim \nu, \nu \sim \hat\rho} [J(u(\pi), \mu^{\pi_s}) - J(\pi, \mu^{\pi_s})],
\end{align*}
with $\hat\rho$ the empirical play of $(\pi_s)_{0 \geq s \geq t}$, which concludes the proof.

\end{proof}

Once again, we may wonder what happens when a no-regret algorithm learns: does it go closer to the set of correlated equilibria ? The following proposition answers this question positively.

\begin{proposition}
Let $(\pi_s)_{0 \leq s \leq t}$ be such that $\lim_{t \rightarrow \infty} \frac{1}{t} \swapreg((\pi_s)_{0 \leq s \leq t}, (\mu_s)_{0 \leq s \leq t}) = 0$. Then the empirical play of $(\pi_s)_{0 \leq s \leq t}$ converges to the set of correlated equilibria, \emph{i.e.} $\min\limits_{\rho_0 \in \mathcal{C}} d_{\mathcal{W}_2}(\rho_\nu, \rho_0) \rightarrow 0$.
\end{proposition}
\begin{proof}
    The proof follows the same steps as that of Proposition~\ref{prop:cce_conv_regret}, the only change being the set of deviations considered and the deviation payoff function. Since the deviation payoff function remains continuous, the proof remains unchanged.
\end{proof}

\section{Learning Weak Equilibria in Mean Field Games}\label{sec:learning}

Now that we have introduced new equilibrium concepts for Mean-Field games, a new question must be asked: how can they be algorithmically reached ? This section provides new insights on various learning algorithms that are known to efficiently learn Nash equilibria in Mean Field Games under certain conditions, including Nash unicity.

More specifically, we focus on three algorithms, which we apply to Mean Field games that do not necessarily satisfy monotonicity or contractivity properties. We study Online Mirror Descent~\cite{perolat2021scaling}'s convergence properties without assuming monotonicity; we also present a new version of Fictitious Play~\cite{perrin2020fictitious}, \textit{Joint Fictitious Play}, and prove that both Online Mirror Descent and Joint Fictitious Play are no-external-regret. As we proved in Section~\ref{sec:regret_minimization_empirical_play}, this means that their empirical plays converge towards the set of coarse correlated equilibria. Finally, we provide a brief overview of Mean-Field PSRO, recently introduced by \citet{muller2021learning}, which can converge to both correlated and coarse-correlated equilibria. 

We remark once again that these results do not require \emph{any condition} on the games played - while they fit within our framework -, in particular, they do not require any monotonicity or contractivity properties to be true.

\subsection{Mean-Field Joint Fictitious Play}

Using Fictitious play algorithms to learn Nash equilibria in games dates back to the seminal papers of Brown \cite{Brown1951} and Robinson \cite{Robinson1951}. Its extension to Mean Field games has been considered in \cite{cardaliaguet2017learning, Hadikhanloo-phdthesis}, while its rate of convergence has been discussed in \cite{perrin2020fictitious} when learning in continuous time  and in \cite{geist2021concave} when learning in discrete time. We focus here on frameworks of games for which several Nash equilibria may exist and present a variant of Fictitious Play in continuous learning time.

\subsubsection*{Continuous time Joint Fictitious Play Algorithm}

In Joint Fictitious Play (Joint FP), at every step, the agents all play simultaneously the same policy which is sampled from the past best responses. In continuous time, at time $s$, each best response is computed as: 
\begin{equation*}
    \pi^{BR}_\tau = \argmax_{\pi' \in \Pi} \int \limits_{s=0}^\tau \langle \mu^{\pi'}, r^{\pi'}(\cdot, \mu^{\pi_{s}}) \rangle ds\;,
\end{equation*}

\begin{equation*}
    \mu^{\pi_\tau}(x) = \frac{1}{\tau} \int \limits_0^\tau \mu^{\pi^{BR}_s}(x) ds\;.
\end{equation*}

 \begin{remark}
An intuition for the reason why we need a different algorithm for no-regret-learning while in $N$-player games, traditional fictitious play is no-regret comes from the Mean-Field non-linearity problem, highly highlighted by~\cite{muller2021learning}: while Joint FP and FP are the same in the $N$-player setting - in those, the reward against an averaged policy is the same as the average reward against each policy -, they are different here, and only Joint FP directly minimizes external regret. It is unclear whether FP also minimizes external regret, or if there are cases where it would not.
 \end{remark}

\subsubsection*{Regret minimization}

The convergence of continuous-time FP to the set of mixed Nash equilibria in the context of monotone Mean Field Games has been derived in \cite{perrin2020fictitious}. It allows to encompass the presence of common noise in the dynamics and the derived convergence rate is of order $O(1/\tau)$. This convergence property requires the consideration of Mean Field Games satisfying the classical monotonicity condition, ensuring in particular the uniqueness of Nash equilibrium. 

Whenever the monotonicity condition is not satisfied, we verify that a small alteration to continuous-time FP, continuous-time Joint FP, converges to a coarse correlated equilibrium. This is proven from the external regret minimization property of Joint FP.

Following a similar line of argument as in \cite{ostrovski2013payoff}, we now demonstrate that continuous time JFP converges to a MF-CCE (observe that the monotonicity assumption is not required). 

\begin{proposition}
For continuous time JFP, at time $\tau$, the regret $\extregret\left((\pi(s))_{0 \leq s \leq \tau}, (\mu^{\pi(s)})_{0 \leq s \leq \tau}\right)$ of the continuous time FP policy  
is of order $O(1/t)$.
\end{proposition}
\begin{proof}
For $\tau > 0$  and by definition of $\pi^{BR(\tau)}$, the envelope theorem \cite{ekeland1999convex} ensures  that
\begin{equation*}
    \frac{d}{d\tau} \left[ \max_{\pi'}  \int \limits_{s=0}^\tau \langle \mu^{\pi'}, r^{\pi'}(., \mu^{\pi^s})\rangle ds  \right] 
    = \langle \mu^{\pi^{BR(\tau)}}, r^{\pi^{BR(\tau)}}(., \mu^{\pi^\tau})\rangle.%
\end{equation*}

\begingroup
\allowdisplaybreaks
Integrating between an arbitrary time $\tau_0 > 0$ and $T$, this directly implies
\begin{align*}
    & \max_{\pi'} \int \limits_{s=0}^T \langle \mu^{\pi'}, r^{\pi'}(., \mu^{\pi^s})\rangle ds - \max_{\pi'} \int \limits_{s=0}^{\tau_0} \langle \mu^{\pi'}, r^{\pi'}(., \mu^{\pi^s})\rangle ds \\
    &= \int \limits_{\tau_0}^T \frac{d}{d\tau} \left[ \max_{\pi'} \int \limits_{s=0}^\tau \langle \mu^{\pi'}, r^{\pi'}(., \mu^{\pi^s})\rangle ds \right] d\tau \\
    &= \int \limits_{\tau_0}^T \langle \mu^{\pi^{BR(\tau)}}, r^{\pi^{BR(\tau)}}(., \mu^{\pi^{\tau}}) \rangle d\tau 
    \\
    &= \int \limits_{\tau=0}^T  \langle \mu^{\pi^{BR(\tau)}}, r^{\pi^{BR(\tau)}}(., \mu^{\pi^\tau}) \rangle d\tau - 
    \int \limits_{\tau=0}^{\tau_0} \langle \mu^{\pi^{BR(\tau)}}, r^{\pi^{BR(\tau)}}(., \mu^{\pi^\tau})\rangle d\tau \,.
\end{align*}
\endgroup

Finally, we deduce that  
\begin{align*}
    &\max_{\pi'} \int \limits_0^T \langle \mu^{\pi'}, r^{\pi'}(.,\mu^{\pi^s})\rangle ds - \int \limits_0^T \langle \mu^{\pi^{BR(s)}},  r^{\pi^{BR(s)}}(.,\mu^{\pi^s})\rangle ds 
    \\
    &= \max_{\pi'} \int \limits_0^{\tau_0} \langle \mu^{\pi'}, r^{\pi'}(.,\mu^{\pi^s})\rangle ds - \int \limits_0^{\tau_0} \langle \mu^{\pi^{BR(s)}},  r^{\pi^s}(.,\mu^{\pi^{BR(s)}})\rangle ds \, .
\end{align*}

Hence, the previous left hand side expression is $O(1)$ implying that the external regret is $O(1/t)$. \end{proof}

\subsubsection*{Discrete time Joint Fictitious Play algorithm}

We describe here a discretization of the above continuous algorithm in Algorithm~\ref{alg:JFP}, whose empirical convergence properties are illustrated in Section~\ref{sec:exp_results}.

\begin{algorithm2e}[ht!]
\SetKwInOut{Input}{input}\SetKwInOut{Output}{output}
\Input{Start with an initial policy $\pi_0$, an initial distribution $\mu_0$ and define $\bar{\pi}_0 = \pi_0$}
\caption{Joint Fictitious Play in Mean Field Games \label{alg:JFP}}
\For{$j=1,\dots,J$:}{
      compute $\pi^{BR}_j$ a best response maximizing $\bar r^j = \sum_{i=0}^j \langle \mu^{\pi^{BR}_j}, r^{\pi^{BR}_j}(\cdot, \mu^{\pi^{i}}) \rangle $\;
      compute $\mu^{\pi^{BR}_{j}}$, the distribution induced by $\pi^{BR}_j$\;
      compute $\bar\mu_j = \frac{j-1}{j} \bar\mu_{j-1} + \frac{1}{j} \mu^{\pi^{BR}_{j}}$ and $\pi_j$ the policy whose distribution is $\bar\mu_j$.  
}
\Return{$\hat\rho = \text{Uniform distribution over } ((\nu_{\pi_j})_j)$}
\end{algorithm2e}

\subsubsection*{Dominated strategy exclusion}

Finally, we investigate the relationship between Joint FP's empirical play and dominated strategies: do we have guarantees that Joint FP's computed equilibrium will not include dominated strategies ? How about pre-asymptotic behavior, how quickly are dominated strategies eliminated from play ? 

\begin{proposition}[Fictitious Play Pareto-Optimality]\label{prop:jfp_pareto_optimality}
    Let $(\pi_t)_{t \in ]0;T]}$ be the policies produced by Fictitious Play by time $T > 0$. Then a policy sampled from this set will asymptotically almost-surely never be dominated as $T \rightarrow \infty$, and the probability of sampling a dominated strategy is $\leq \frac{1}{T}$. 
\end{proposition}
\begin{proof}
    We begin the proof by recalling the definition of a dominated policy: $\pi \in \Pi$ is dominated if there exists $\pi' \in \Pi, \; \forall \mu, \; J(\pi', \mu) > J(\pi, \mu))$.

    We note that $\forall t > 0, \; \pi^{BR}(t)$ can by definition not be dominated, since it is defined as $\argmax_{\pi'} \int_{s=0}^t J(\pi', \mu_s) ds$: if $\pi'$ dominated $\pi^{BR}(t)$, then $\int_{s=0}^t J(\pi', \mu_s) ds > \int_{s=0}^t J(\pi_t, \mu_s) ds$, which is contradictory. 
    
    Therefore, the only potentially dominated strategy among the mixture that defines $\pi^t$ is $\pi_0$ : the probability that $\pi^t$ plays according to a dominated strategy is therefore at most the probability that $\pi^t$ plays $\pi_0$.
    
    The policy-mixing distribution is continuous, so this probability is null for all $t > 1$, and potentially equal to 1 for $t \in [0; 1]$. We therefore have $\mathbb{P}\left(\text{Sampling actions following } \pi_0 \text{ from } \pi_t \; \mid \; t \right) = \frac{1}{T}$ if $t \geq 1$. %
    
    All in all, we have 
    \begin{align*}
    \mathbb{P}\left(\text{Playing dominated strategy in a game}\right) &\leq \mathbb{P}\left(\text{Playing according to } \pi_0 \text{ in a game.} \right) \\
    &\leq \frac{1}{T}
    \end{align*}
    
\end{proof}

\subsection{Mean-Field Online Mirror Descent}

We now turn to Online Mirror descent algorithms for mean field games as studied in \cite{perolat2021scaling}.

\subsubsection*{Continuous time Mean Field Online Mirror Descent}

\begin{algorithm2e}[ht!]
    \SetKwInOut{Input}{Input}\SetKwInOut{Output}{Output}
    \Input{$N$ number of actions, $\eta > 0$ learning rate, $\tau_{max}$ max learning steps.}
    \caption{Discrete-Time Online Mirror Descent \label{alg:discrete_omd}}
    $\tau = 0$, $y_0 = 0$, $\pi_0 = \text{Uniform policy}$ \;
    \While{$\tau \leq \tau_{max}$}{
        Observe current Q-value $Q^{\pi_{\tau}}(x, \cdot, t) \quad \forall x,\, t$ \;
        Set $ y_{\tau+1}(x, \cdot, t) = y_\tau(x, \cdot, t) + \eta Q^{\pi_{\tau}}(x, \cdot, t) \quad \forall x,\, t $ \;
        Set $\tau = \tau + 1$ \;
        Compute $\pi_\tau(x, \cdot, t) = \text{softmax}(y_\tau(x, \cdot, t))$ \;
    }
    \Return{Collection of policies $(\pi_\tau)_{\tau \geq 0}$}
\end{algorithm2e}

For the Online Mirror Descent algorithm, \cite{perolat2021scaling} introduce a regularizer $h: \Delta(\mathcal{A}) \rightarrow \mathbb{R}$, that is assumed to be $\rho$-strongly convex for some constant $\rho>0$. 
Its conjugate $h^*: \mathbb{R}^{\mathcal{A}} \rightarrow \mathbb{R}$ is defined as $h^*(y) = \max \limits_{\pi \in \Pi} [<y,\pi> - h(\pi)]$.  When $h$ has good properties  we have 
\begin{equation}
\label{eq:Dhstar-Gamma}
    \Gamma(y) := \nabla h^*(y) = \argmax \limits_{\pi} [<y,\pi> - h(\pi)] \;.
\end{equation}

The continuous-time Online Mirror Descent dynamics are defined as %
\begin{align}
    &y_t(x,a,\tau) = \int \limits_{0}^{\tau} Q^{\pi(s), \mu^{\pi(s)}}_t(x,a) ds, \quad t \in \mathcal{T}
    \label{eq:OMD-y}
    \\
    &\pi_t(.\mid x,\tau) = \Gamma(y_n(x,.,\tau)), \quad t \in \mathcal{T}
    \label{eq:OMD-pi}
\end{align}
where we define $Q^{\pi, \mu} = (Q^{\pi, \mu}_t)_{t \in \mathcal{T}}$ and, with $T = \max\limits_{t \in \mathcal{T}} t$: 
\begin{align*}
    \begin{cases}
    Q^{\pi, \mu}_{T}(x,a) = 0
    \\
    Q^{\pi, \mu}_t(x,a) = r(x,a,\mu_t) + \sum \limits_{x' \in \mathcal{X}} p(x'\mid x,a,\mu_t) \sum_{a'} \pi_t(x, a') Q^{\pi, \mu}_{t+1}(x',a'), 
    \\
    \qquad t=T-1,T-2,\dots,0.
    \end{cases}
\end{align*}
where we assume, without loss of generality, that $\mathcal{T}$ is the sequence 0, ..., $T$.

\subsubsection*{Convergence properties}

We characterize the regret-minimizing properties of Online Mirror Descent.

\begin{theorem}\label{prop:omd_no_regret}
 Online Mirror Descent is a regret minimizing strategy in Mean Field games (no monotonicity required):
\[
\frac{1}{\tau} \extregret((\pi(s))_{0 \leq s \leq \tau}; (\mu^{\pi(s)})_{0 \leq s \leq \tau}) = O(\frac{1}{\tau})
\]
\end{theorem}

We prove Proposition~\ref{prop:omd_no_regret} in Appendix~\ref{appendix:omd_no_regret_proof}. 

We note that we obtain convergence bounds for the external regret of Online Mirror Descent, which is in stark contrast with its exploitability, for which, even in the monotonic case, we do not have such bounds. This is due to the fact that what makes average external regret converge is the averaging, and external regret being strictly bounded thanks to it being the sum of past Online Mirror Descent plays: whereas a single Online Mirror Descent policy may be more or less exploitable in ways that are difficult to evaluate, its sequence of policies is difficult to exploit "all at the same time", leading to bounded external regret.

\subsubsection*{Dominated strategy exclusion}

Similarly to Joint FP, we investigate OMD's exclusion of dominated startegies and its speed in doing so. Just like Joint FP, OMD's elimination of dominated strategies in its empirical play is $\mathcal{O}\left ( \frac{1}{T} \right )$-quick due to the empirical play's uniform average over all previous timesteps. 

\begin{proposition}[Online-Mirror Descent Optimality]
    As $t$ tends to infinity, a policy $\pi$ uniformly sampled from $(\pi_t)_{t \in [0;T]}$ produced by OMD with entropy regularizer almost-surely never takes $\epsilon > 0$-dominated actions.
\end{proposition}

\begin{proof}
    Let $x$ be a state, $a_1$ an action $\epsilon$-dominated by $a_2$, \emph{i.e.} $\forall \mu \in \mathcal{P}(\mathcal{X})$, $\forall \pi \in \Pi$, $Q^{\pi, \mu}(x, a_1) \leq Q^{\pi, \mu}(x, a_2) - \epsilon$ with $\epsilon > 0$.
    We have that $\pi_t(x) = \text{softmax}(y)$, and $y = \int_0^T Q^{\pi_s, \mu^{\pi_s}} ds$.
    Directly, $y(x, a_1, t) \leq - \epsilon t + y(x, a_2, t)$, thus $\pi_t(x, a_1) \leq e^{- t \epsilon} \pi_t(x, a_2)$.
    
    Whether $a_2$ keeps being selected or not, we have necessarily that $\pi_t(x, a_1) \rightarrow 0$.
    
    Let $\epsilon' > 0$, $t' > 0$ such that $\forall t > t', \; \pi_t(x, a_1) < \frac{1}{2} \epsilon'$. Finally, take $T$ such that $\frac{t'}{T} \leq \frac{1}{2} \epsilon'$, and randomly sample $\pi_t$ from $(\pi_t)_{t \in [0;T]}$. 
    \begin{align*}
        \mathbb{P}(\text{$\pi_t$ plays $a_1$}) &= 
        \underbrace{\mathbb{P}(\text{$\pi_t$ plays $a_1$} \; \mid \; t < t')}_{\leq 1} \mathbb{P}(t < t') + \underbrace{\mathbb{P}(\text{$\pi_t$ plays $a_1$} \; \mid \; t \geq t')}_{< \frac{1}{2} \epsilon'} \mathbb{P}(t \geq t') \\        
        &\leq \frac{t'}{T} + \frac{\epsilon'}{2} \underbrace{\frac{T-t'}{T}}_{\leq 1} \\
        &\leq \epsilon'
    \end{align*}
    There are only a finite amount of states and actions, thus there is only a finite amount of dominated actions. Taking a sup over all possible times $T$, we have that for all $\epsilon, \epsilon' > 0, \; \exists T' > 0$ such that $\forall T \geq T'$, $\mathbb{P}(\text{Sampled }\pi_t \text{ from } (\pi_t)_{t \in [0;T]}\text{ plays } \epsilon\text{-dominated action}) \leq \epsilon'$, which concludes the proof.
\end{proof}

Neither algorithm presented above converges towards a Mean-Field correlated equilibrium, and one could legitimately wonder whether such an algorithm does exist. Mean-Field PSRO, introduced by Muller et al. \cite{muller2021learning}, and presented below, answers this question by the affirmative.

\subsection{Mean-Field PSRO and Derived Algorithms}

This section presents Mean-Field PSRO, introduced by~\citet{muller2021learning}, starting from its main component, Mean-Field Adversarial Regret Minimization, and then moving on to the full algorithm, which uses Mean-Field Adversarial Regret Minimization as its core component. It is provided here for sake of completeness.

\subsubsection*{Mean-Field Adversarial-Regret Minimizers}

The central idea of Mean-Field PSRO's argument for convergence to (coarse) correlated equilibria rests upon adversarial regret minimization theory, which we can describe like this: imagine we are tasked with finding a probability distribution over $N$ actions to maximize a given return. We have to do this $T$ times; while another player chooses an observed reward vector $r_t \in \mathcal{R} \subseteq \mathbb{R}^N$ at each time $t$, with $\mathcal{R}$ a compact subset of $\mathbb{R}^N$.

This vector reward can be anything, and, in particular, we are trying to minimize regret - external or internal - in the \emph{worst possible case}, hence the name: adversarial regret minimization. We notice that, given a distribution over policies $\nu \in \Delta(\Pi)$, $\langle \nu, J(\cdot, \mu(\nu)) \rangle \geq \min\limits_{r \in \mathcal{R}} \langle \nu, r \rangle$: the distribution-dependent Mean-Field expected return cannot be worse than the worst possible reward vector - assuming $\mathcal{R}$ has been chosen so as to encompass all possible values of $J$, which is possible if \textit{e.g.} $J$ is bounded, which is true when $r$ is bounded.

Given this property, one can run a regret-minimizing algorithm on the following bandit problem: 

At each round, given regret-minimizing algorithm $A$, 
\begin{enumerate}
    \item Get probability distribution $\nu \in \Delta(\Pi)$ over deterministic policies from $A$.
    \item Observe reward vector $J(\cdot, \mu(\nu))$.
    \item Return reward vector $J(\cdot, \mu(\nu))$ to $A$.
\end{enumerate}

This procedure is then looped over until the algorithm has reached low-enough average regret. In this paper, we show examples of two external regret-minimizing algorithms, Polynomial Weights, and Regret Matching, presented in Algorithm~\ref{alg:poly_weights} and ~\ref{alg:regret_matching}. These two algorithms have, after $T$ iterations, $O\left(\frac{1}{\sqrt{T}}\right)$ regret.

Note that, using the technique introduced by~\citet{BlumInternalExternalRegret}, these algorithms may also be used to minimize internal regret, thereby leading to correlated equilibria.

We note that this is close to the problem of finding no-regret policies in adversarial MDPs, where the adversary is able to control both the dynamics and the reward of the MDP, for which other algorithms~\cite{abbasi2013adversarial} exist that reach $O(\frac{1}{\sqrt{T}})$ average regret as well.  

\begin{algorithm2e}[ht!]
    \SetKwInOut{Input}{Input}\SetKwInOut{Output}{Output}
    \Input{$N$ number of actions, $\epsilon > 0$ precision threshold, $\eta > 0$ learning rate}
    \caption{Polynomial Weights\label{alg:poly_weights}}
    Set $t = 0$, $w_0^i = 0 \; \forall i$. \;
    \While{Average Regret $< \epsilon$}{
        Observe reward vector $r_t$ and average reward $\langle \frac{w_t}{\sum_i w_t^i}, r_t \rangle $\;
        Set $w_{t+1}^i = w_{t}^i (1 + \eta r_t^i)$ for all i \;
        $t += 1$ \;
    }
    \Return{Collection of probabilities $\left(\frac{w^i_t}{\sum_j w_t^j} \right)_{t, i}$}
\end{algorithm2e}

\begin{algorithm2e}[ht!]
    \SetKwInOut{Input}{Input}\SetKwInOut{Output}{Output}
    \Input{$N$ number of actions, $\epsilon > 0$ precision threshold}
    \caption{Regret Matching\label{alg:regret_matching}}
    $t = 0$, $R_t^i = 0, \, p_0^i = \frac{1}{N} \; \forall i$\;
    \While{Average Regret $< \epsilon$}{
        Observe reward vector $r_t$ and average reward $\langle \frac{w_t}{\sum_i w_t^i}, r_t \rangle $\;
        Update regret vector $R_{t+1}^i = R_t^i + r_t^i - \langle \frac{w_t}{\sum_i w_t^i}, r_t \rangle $ \;
        $t += 1$ \;
        \If{$\sum_i \lfloor R_{t}^i \rfloor_+ = 0$}{
            Play $p^i_{t} = \frac{1}{N} \quad \forall i$}
        \Else{
            Play $p^i_{t} = \frac{\lfloor R_t^i \rfloor_+}{\sum_i \lfloor R_{t}^i \rfloor_+} \quad \forall i$}
    }
    \Return{Collection of probabilities $(p_t^i)_{t, i}$}
\end{algorithm2e}

\subsubsection*{Mean-Field PSRO}

We now move on to Mean-Field PSRO. The algorithm, an adaptation of PSRO \cite{lanctot2017psro, marrismuller2021ce} to the Mean-Field case, proceeds following the usual PSRO scheme, described in Algorithm~\ref{alg:Mean-Field-psro}. Briefly, the algorithm proceeds as follows:

\begin{enumerate}
    \item Compute a best-response / best-responses to the current distribution / correlation device in a certain way,
    \item Integrate this best-response / these best responses to the policy pool,
    \item Recompute a new distribution / correlation device,
    \item Loop until no new best-response is found / no value improvement is found.
\end{enumerate}

\begin{algorithm2e}[ht!]
    \SetKwInOut{Input}{Input}\SetKwInOut{Output}{Output}
    \Input{Initial Policy $\pi_0$, Meta-Solver $\sigma$ and Best-Response Computer $BR$}
    \caption{Mean-Field PSRO\label{alg:Mean-Field-psro}}
    $ \Pi_0 = \{ \pi_0 \} $, $ \sigma_0 = \sigma(\Pi_0) $, $n = 0$ \;
    \While{$BR(\Pi_n, \sigma_n) \not\in \Pi_n$ }{
        Add best-response(s) to the pool: $\Pi_{n+1} = \Pi_n \cup \{ BR(\Pi_n, \sigma_n) \} $ \;
        Recompute optimal distribution $\sigma_{n+1} = \sigma(\Pi_{n+1})$ \;
        $n += 1$ \;
    }
    \Return{$ \Pi_n, \; \sigma_n $}
\end{algorithm2e}

The exact best-response computation type in step 1 depends on whether we compute a correlated equilibrium - in which case a best-response is computed for every policy with non-zero play probability, quantifying the payoff for deviating from being recommended that policy - or a coarse correlated equilibrium / a Nash equilibrium - where a unique best-response is computed against the correlation device / the distribution.

The difficulty in bringing PSRO to the Mean-Field case, as noted by \citet{muller2021learning}, but also \citet{ywang2021empirical}, who also introduced Mean-Field PSRO, but only for computing Nash equilibria, is the sudden non-linearity of the game with respect to the distribution of policy played by agents: the Mean-Field reward term can now be heavily non-linear. This renders payoff-table-based equilibrium computations impossible, unless one creates an entry for every possible strategy mixture - which is of course impossible. \citet{muller2021learning} and \citet{ywang2021empirical} go around this issue by computing equilibria without using payoff tables:

\begin{itemize}
    \item To compute Nash equilibria, they use either existing Nash-converging algorithms (\citet{ywang2021empirical}) or black-box approaches (\citet{muller2021learning}),
    \item To compute correlated or coarse correlated equilibria, \citet{muller2021learning} use a distribution generated from the play of adversarial regret minimizers. 
\end{itemize}

The intuition behind using adversarial regret minimizers is the following: if a no-adversarial-regret algorithm minimizes regret against any type of adversary, then it can minimize regret against its own state-distribution-induced reward changes, and thus converge towards coarse correlated and correlated equilibria. This is the core of the algorithm's idea.

Both approaches are shown to work theoretically, \textit{i.e.} Mean-Field PSRO converges towards Mean-Field Nash, correlated and coarse correlated equilibria both theoretically and empirically, as we see in the following theorems derived from \citet{muller2021learning}:

\begin{theorem}[Adaptation of Theorem 8 in \citet{muller2021learning} - MFCE]\label{theorem:mf_psro_convergence_ce}
    When using a no-swap-regret meta-solver with average regret threshold $\epsilon$, and the right best-response concept, Mean-Field PSRO converges to an $\epsilon$-MFCE.
\end{theorem}

\begin{theorem}[Adaptation of Theorem 8 in \citet{muller2021learning} - MFCCE]\label{theorem:mf_psro_convergence_cce}
    When using a no-external-regret meta-solver with average regret threshold $\epsilon$, and the right best-response concept, Mean-Field PSRO converges to an $\epsilon$-MFCCE.
\end{theorem}

We let the reader consult~\cite{muller2021learning} for more details, \emph{e.g.} what \emph{the right} best response concepts are.

\subsubsection*{Dominated strategy exclusion}

Just like for Joint FP and OMD, we examine the relationship between Mean-Field PSRO and dominated strategies. Perhaps surprisingly, we find that PSRO does not necessarily eliminate dominated strategies, at least when computing coarse-correlated equilibria. The only guarantee we find is that, when computing correlated equilibria, it always asymptotically eliminates them. 
To counteract this undesirable property, we propose two different alterations of the algorithm which guarantee that Mean-Field PSRO \textbf{never} recommends a dominated strategy \emph{at any time during training}.

\begin{proposition}[Mean-Field PSRO's CE-optimality]
    Mean-Field PSRO used to compute Mean-Field correlated equilibria can never recommend a dominated strategy at convergence.
\end{proposition}
\begin{proof}
    The proof results from the fact that a correlated equilibrium can, by definition, never recommend a strictly dominated strategy (If it did, then deviating to the strategy which dominates the dominated strategy would always yield payoff improvements, and therefore the correlation device in question would not be a correlated equilibrium). At convergence, Mean-Field PSRO has found a correlated equilibrium, and hence cannot recommend strictly dominated strategies. 
\end{proof}

However, we note that PSRO could potentially recommend strictly dominated strategies when \textit{e.g.} computing Mean-Field coarse correlated equilibria (Which can contain dominated strategies, as shown in Section~\ref{subsec:eq_viz}), or in the process of computing a Mean-Field correlated equilibrium. This is due to the initial policies present in the initialization pool of PSRO, of which we cannot guarantee optimality / non-suboptimality. It is however possible to slightly modify the algorithm to obtain an optimality-guaranteeing result:

\begin{proposition}[Mean-Field PSRO: Optimality Modification]
    Either of the following two PSRO modifications ensures that PSRO \textbf{never} recommends strictly dominated strategies, while keeping PSRO's convergence guarantees:
    \begin{itemize}
        \item Ensure that all of PSRO's initial policies are not strictly dominated.
        
        or

        \item After PSRO's first iteration, remove all initial policies from the pool and only keep the best-responses.  (Only PSRO's first step can then contain strictly-dominated strategies)
    \end{itemize}
\end{proposition}
\begin{proof}
    Mean-Field PSRO can never add to its pool a strictly dominated strategy, since it only adds best-responses and best-responses can never be strictly dominated. Only the initial policies present in PSRO's pool could potentially be. If they are not (First modification), then PSRO's pool never contains dominated strategies, and therefore PSRO never recommends strictly dominated strategies.
    If we cannot be certain that they aren't, we note that the best response against them can never be strictly dominated; hence, removing them from the pool and only keeping these best-responses empties the pool from potentially strictly dominated strategies, thus preventing PSRO from recommending strictly dominated strategies
\end{proof}

We show examples of OMD's, JFP's and Mean-Field PSRO's behavior in different games in section~\ref{sec:exp_results}.

\section{Experimental Results}\label{sec:exp_results}

The following section presents several experimental results of the algorithms presented so far in this paper, Online Mirror Descent (OMD), Joint Fictitious Play (JFP), and Mean-Field PSRO (MF-PSRO), on normal form games. Openspiel~\cite{lanctot2020openspiel} was used to produce all the figures.

\subsection{Games of Interest}

In order to illustrate the approximation of coarse correlated equilibria by the 3 algorithms described above, we focus our attention on three normal-form, 3-actions (A, B and C) Mean-Field games:
\begin{itemize}
    \item The dominated-action game, with reward structure
    \begin{align*}
        r(A, \mu) &= \mu(A) + \mu(C)\,, \\
        r(B, \mu) &= \mu(B)\,, \\
        r(C, \mu) &= \mu(A) + \mu(C) - 0.05 \mu(B)\,; \\
    \end{align*}
\end{itemize}    

We will use this game to characterize how action C, which is strictly dominated by action A, will be eliminated by different algorithms. It is also interesting to see conditions for algorithms to converge towards playing A only vs. playing B only. We will see that all algorithms eliminate action C, but in different ways and with different speeds.

\begin{itemize}
    \item The almost-dominated-action game, with reward structure
    \begin{align*}
        r(A, \mu) &= \mu(A) + \mu(C)\,, \\
        r(B, \mu) &= \mu(B)\,, \\
        r(C, \mu) &= \mu(A) + \mu(C) - 0.05 \mu(B) \,; \\
    \end{align*}
\end{itemize}

We use this game as an example which shows that an action needs to be \emph{strictly} dominated to be eliminated - action C is dominated by action A whenever $\mu(B) > 0$, but this domination goes to 0 as $\mu(B)$ tends to 0. We will see that, while Joint FP and Mean-Field PSRO eliminate C in this setting, OMD doesn't.

\begin{itemize}
    \item The biased rock-paper-scissors game, with reward structure
    \begin{align*}
        r(A, \mu) &= 0.5 * \mu(B) - 0.3 * \mu(C)\;, \\
        r(B, \mu) &= 0.3 * \mu(C) - 0.7 * \mu(A)\;, \\
        r(C, \mu) &= 0.7 * \mu(A) - 0.5 * \mu(B)\;; \\
    \end{align*}
\end{itemize}

This game is an example of a non-monotonic game where OMD and Joint FP do not converge to a single point, but instead \emph{cycle}: it shows that pointwise convergence is not always obtained with these two algorithms, and cycling is possible.

\subsection{Online Mirror Descent}

Figures~\ref{fig:omd_dom} and~\ref{fig:omd_dom_up} show OMD on the almost-dominated action game with different initializations, Figure~\ref{fig:omd_almost_dom} shows OMD on the dominated-action game, Figure~\ref{fig:omd_rps} shows OMD on the Biased RPS game, and Figure~\ref{fig:omd_dom_many} shows OMD's final policies (determined by a color) as a function on its initial policy (the color's position) on the almost-dominated game. Figures~\ref{fig:omd_dom}, ~\ref{fig:omd_dom_up}, ~\ref{fig:omd_almost_dom} and ~\ref{fig:omd_rps} show OMD's current policy at different learning steps, one red circle per step. Heavily red areas are areas where OMD spent a lot of learning time; very light-red areas are areas not much visited by OMD. Each circle in Figure~\ref{fig:omd_dom_many} represents the initial policy played by OMD via its position, and the final policy played by OMD via its color. Colors are computed as $\pi(A) \text{B} + \pi(B) \text{G} + \pi(C) \text{R}$, where $\pi$ is the final policy, and R, G and B are the primary colors.

We see that, on the dominated-action game of Figures~\ref{fig:omd_dom} and~\ref{fig:omd_dom_up}, OMD eliminates action $C$, which is $0.05$-dominated by action $A$, and converges to either $A$ or $B$ depending on its initialization.
However, we also see what happens when there exists a 0-dominated action: in the almost-dominated-action game, Figure~\ref{fig:omd_almost_dom}, $C$ is 0-dominated by $A$, and we do see that once OMD has eliminated action B from its distribution of play, it finds an equilibrium where it does not eliminate $C$, since $A$ and $C$ are in that case equivalent. This empirically shows that the condition $\epsilon > 0$-dominated condition must be true for a dominated action to be systematically eliminated by OMD.

On the biased rock-paper-scissors game, Figure~\ref{fig:omd_rps}, which is not a monotonic game, we see that the last iterate of OMD does not actually converge to a fixed policy, and instead cycles, yielding an approximate coarse correlated equilibrium. We note that, since its last iterate is only proven to converge in the monotonic case, this does not contradict the theory behind OMD, and instead enriches it with cases where OMD does reach a Mean-Field coarse correlated equilibrium without last-iterate convergence.

Figure~\ref{fig:omd_dom_many} provides a lower-granularity view of OMD's behavior when varying its starting points through initial q-value change on the almost-dominated action game. We see that when $\mu(B) > 50\%$, OMD converges towards B; whereas its behavior is much more nuanced when the probability of playing $B$ is lower than $50\%$: in this case, OMD converges towards a location-dependent, continuous-looking mixture between $A$ and $C$.

\begin{figure}
    \centering
    \begin{subfigure}[b]{0.4\textwidth}
        \centering
        \includegraphics[scale=0.3]{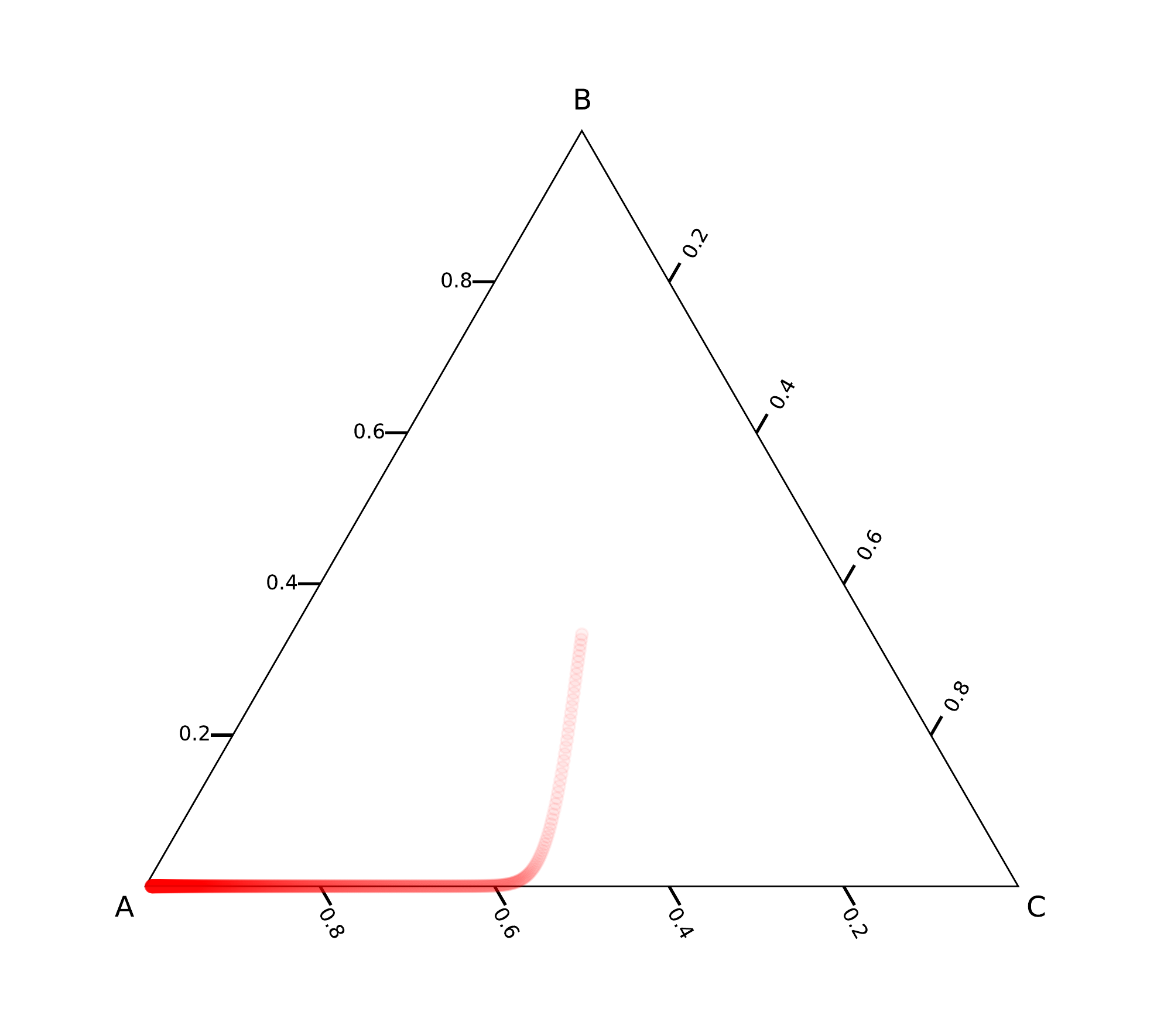}
        \caption{Online Mirror Descent (OMD) on the Dominated Strategy Game - Center start.}
        \label{fig:omd_dom}
    \end{subfigure}
    \hfill
    \begin{subfigure}[b]{0.4\textwidth}
        \centering
        \includegraphics[scale=0.3]{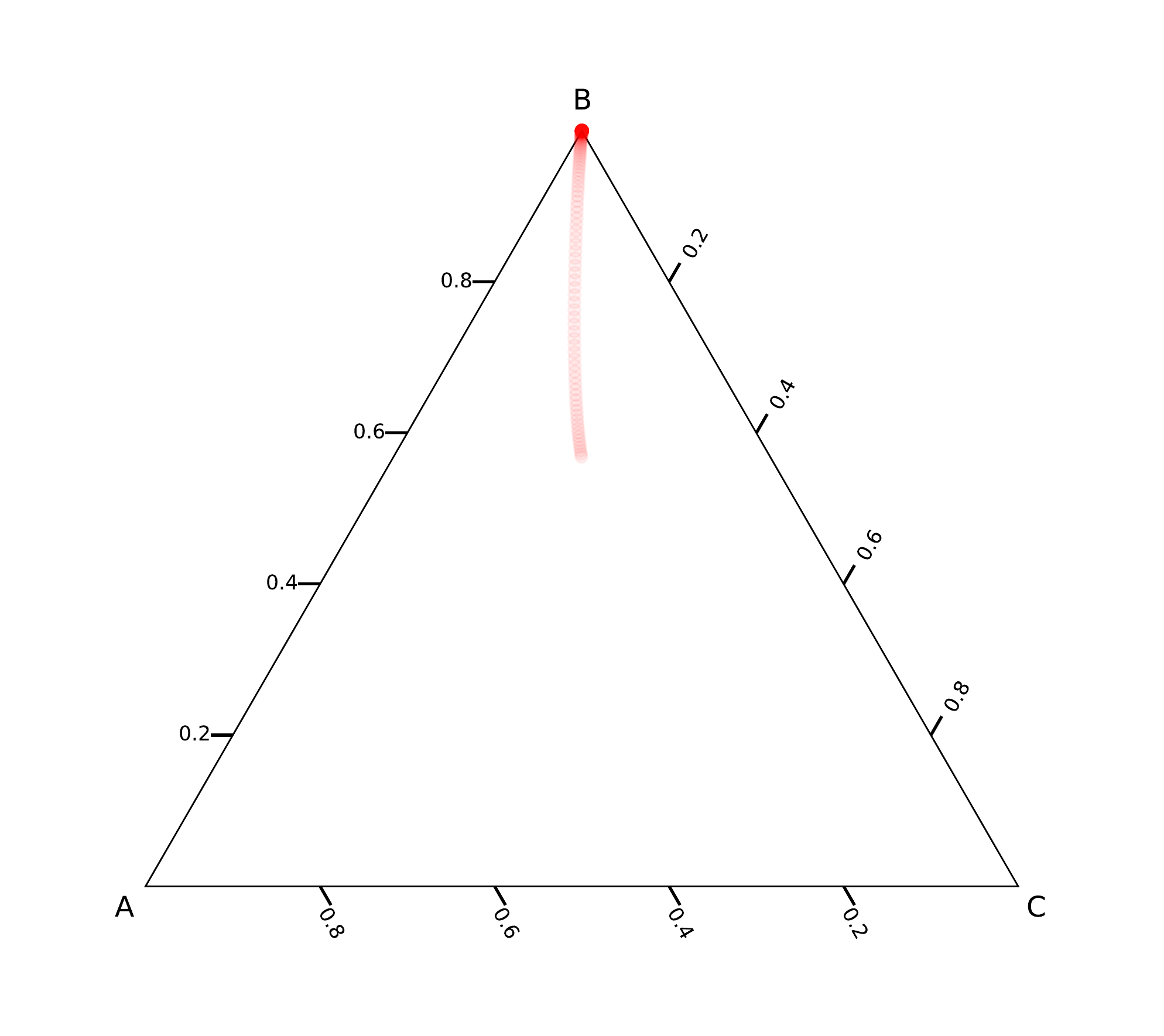}
        \caption{OMD on the Dominated Strategy Game - Biased start towards B.}
        \label{fig:omd_dom_up}
    \end{subfigure}
    \vfill
    \begin{subfigure}[b]{0.4\textwidth}
        \centering
        \includegraphics[scale=0.3]{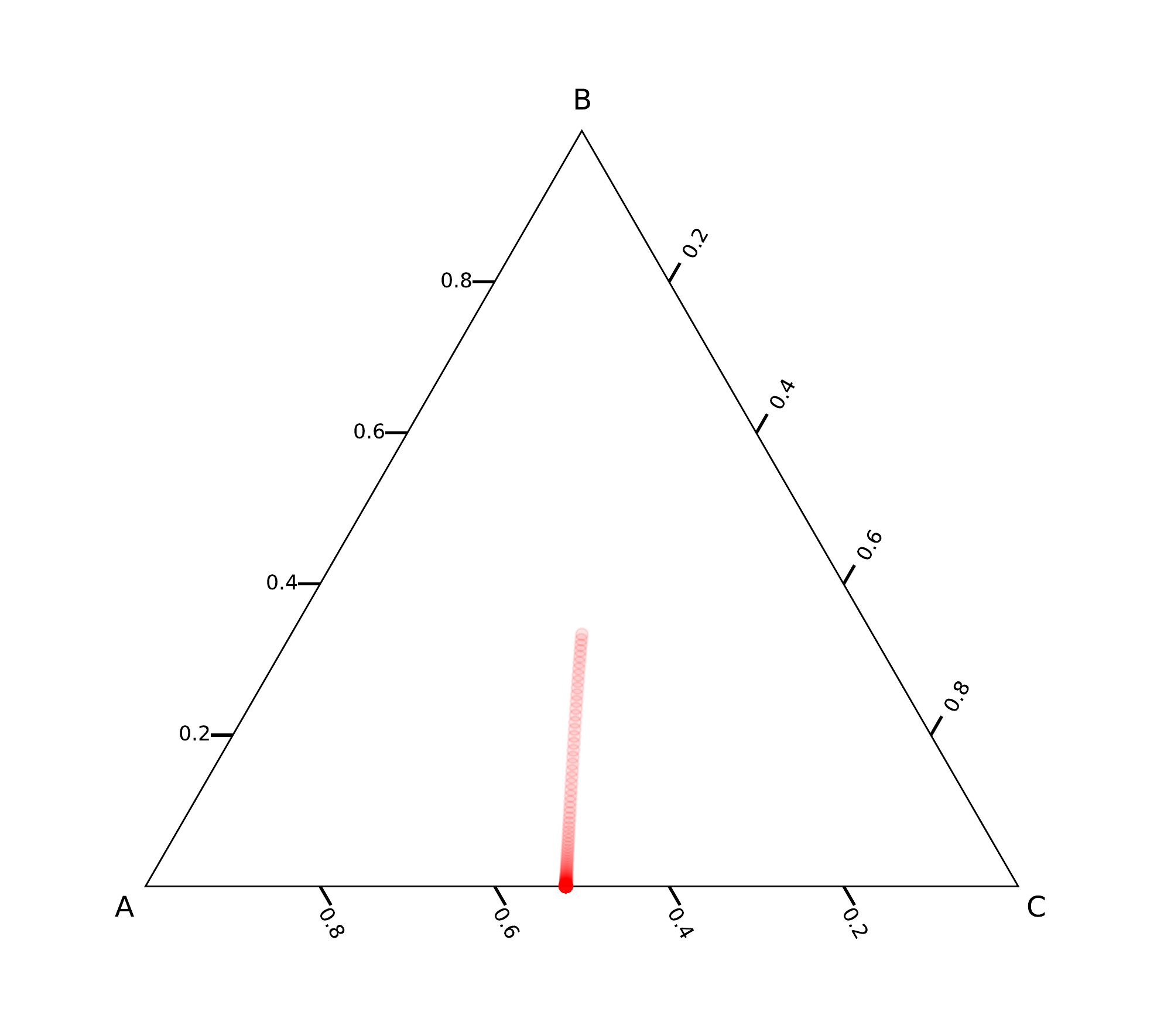}
        \caption{OMD on the Almost-Dominated Strategy Game.}
        \label{fig:omd_almost_dom}
    \end{subfigure}
    \hfill
    \begin{subfigure}[b]{0.4\textwidth}
        \centering
        \includegraphics[scale=0.3]{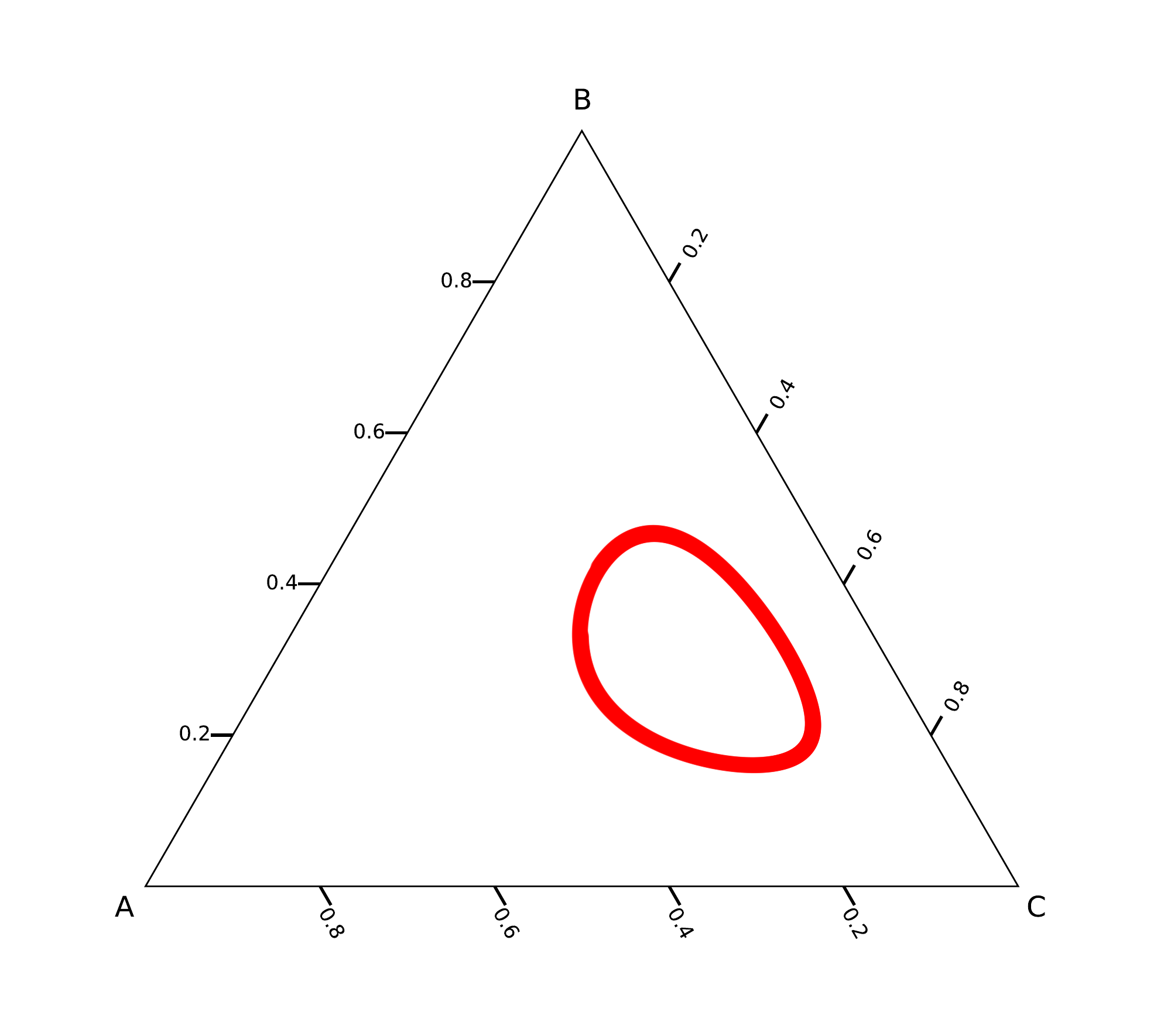}
        \caption{OMD on the Biased Rock-Paper-Scissors Game.}
        \label{fig:omd_rps}
    \end{subfigure}
    \caption{Online Mirror Descent (OMD) on several Normal-Form Mean-Field Games. Each red circle represents OMD's policy at a given step.}
\end{figure}

\begin{figure}
    \centering
    \includegraphics[scale=0.6]{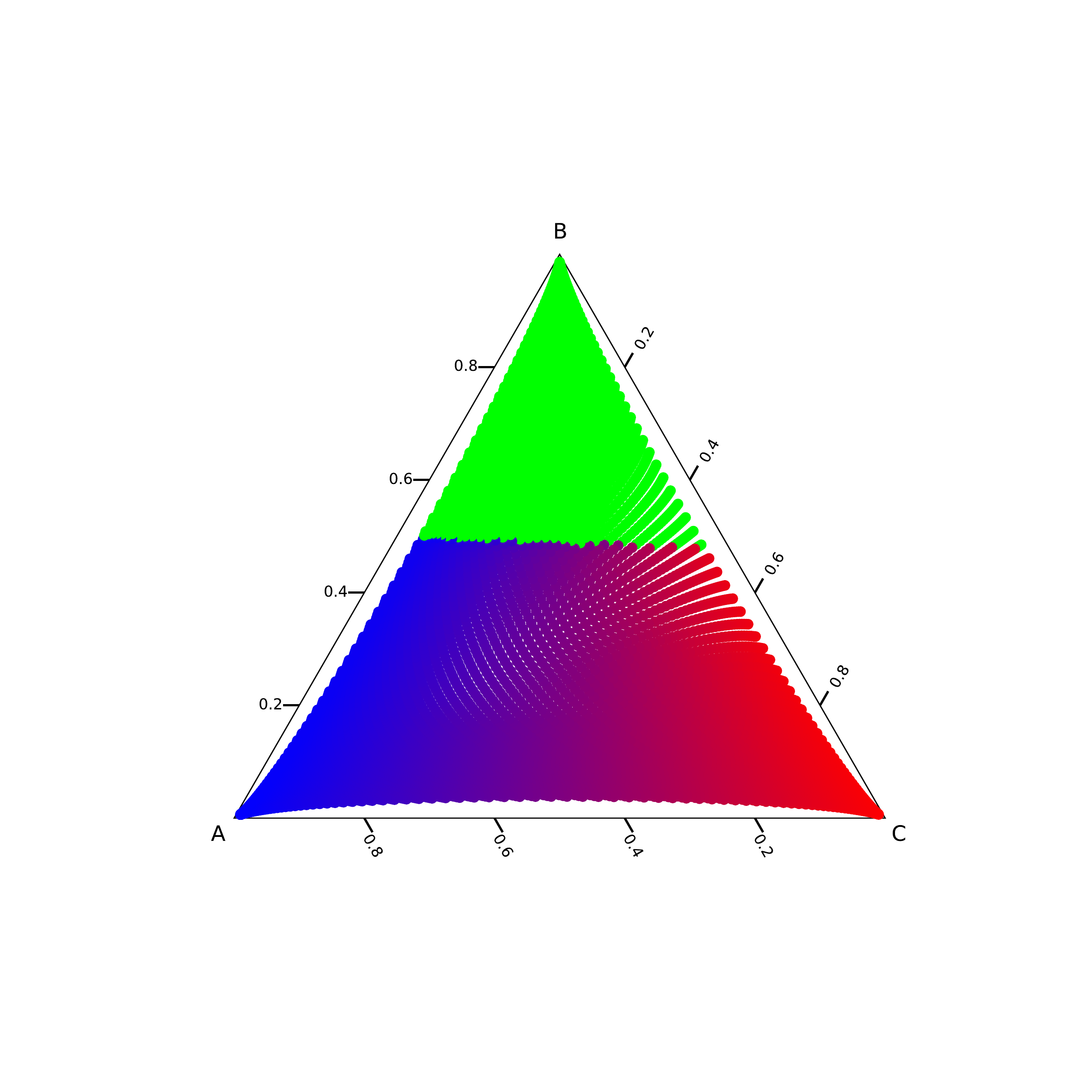}
    \caption{Online Mirror Descent (OMD) map of starting-points to converged-points on the Almost-dominated Strategy game. Each point represents the starting policy of OMD, and each color, its final policy. Colors are computed as $\pi(A) \text{B} + \pi(B) \text{G} + \pi(C) \text{R}$, where $\pi$ is the final policy, and R, G and B are the primary colors.}
    \label{fig:omd_dom_many}
\end{figure}

\subsection{Joint Fictitious Play}

Figures~\ref{fig:jfp_dom} and~\ref{fig:jfp_dom_up} show Joint FP on the almost-dominated action game with different initializations, Figure~\ref{fig:jfp_almost_dom} shows Joint FP on the dominated-action game, Figure~\ref{fig:jfp_rps} shows Joint FP on the Biased RPS game, and Figure~\ref{fig:jfp_dom_many} shows Joint FP's converged-to policies (shown by a color) as a function of its initial policies (the color's position) on the almost-dominated game. Figures~\ref{fig:jfp_dom}, ~\ref{fig:jfp_dom_up}, ~\ref{fig:jfp_almost_dom} and ~\ref{fig:jfp_rps} show Joint FP's current policy at different learning steps, one red circle per step. Heavily red areas are areas where Joint FP spent a lot of learning time; very light-red areas are areas not much visited by Joint FP. Each circle in Figure~\ref{fig:jfp_dom_many} represents the initial policy played by Joint FP via its position, and the final policy played by Joint FP via its color. Colors are computed as $\pi(A) \text{B} + \pi(B) \text{G} + \pi(C) \text{R}$, where $\pi$ is the final policy, and R, G and B are the primary colors.

Figures~\ref{fig:jfp_dom} and~\ref{fig:jfp_almost_dom} demonstrate that Joint Fictitious Play is much faster and harsher in eliminating dominated actions: indeed, action C is never even considered by the algorithm - it is eliminated directly. However, we note that if the algorithm had started in a region where A and C were equivalent (\emph{i.e.} where $\mu_B = 0$), it would indeed have kept their proportions equal. As expected and shown in Figure~\ref{fig:jfp_dom_up}, Joint FP converges to action B when it starts close enough to it.

On Biased Rock-Paper-Scissors, Figure~\ref{fig:jfp_rps}, we notice that Joint FP behaves similarly as OMD: Joint FP does not manage to converge, but instead cycles around the optimal policy, yielding an approximate coarse correlated equilibrium. Very interestingly, but also unsurprisingly, JFP walks "in straight lines", because its new policies are always best responses; its decreasing speed is due to the $\frac{1}{N}$ factor in its update. 

Figure~\ref{fig:jfp_dom_many} provides a lower-granularity view of JFP's behavior when varying its starting policy on the Almost-dominated action game. We see that as soon as the proportion of population playing $B$ exceeds $50\%$, JFP will converge towards B, whereas, contrarily to OMD and in accordance with Proposition~\ref{prop:jfp_pareto_optimality}, it will completely eliminate action $C$ and only focus on action $A$ - since $A$ is everywhere better than $C$.

\begin{figure}
    \centering
    \begin{subfigure}[b]{0.4\textwidth}
        \centering
        \includegraphics[scale=0.3]{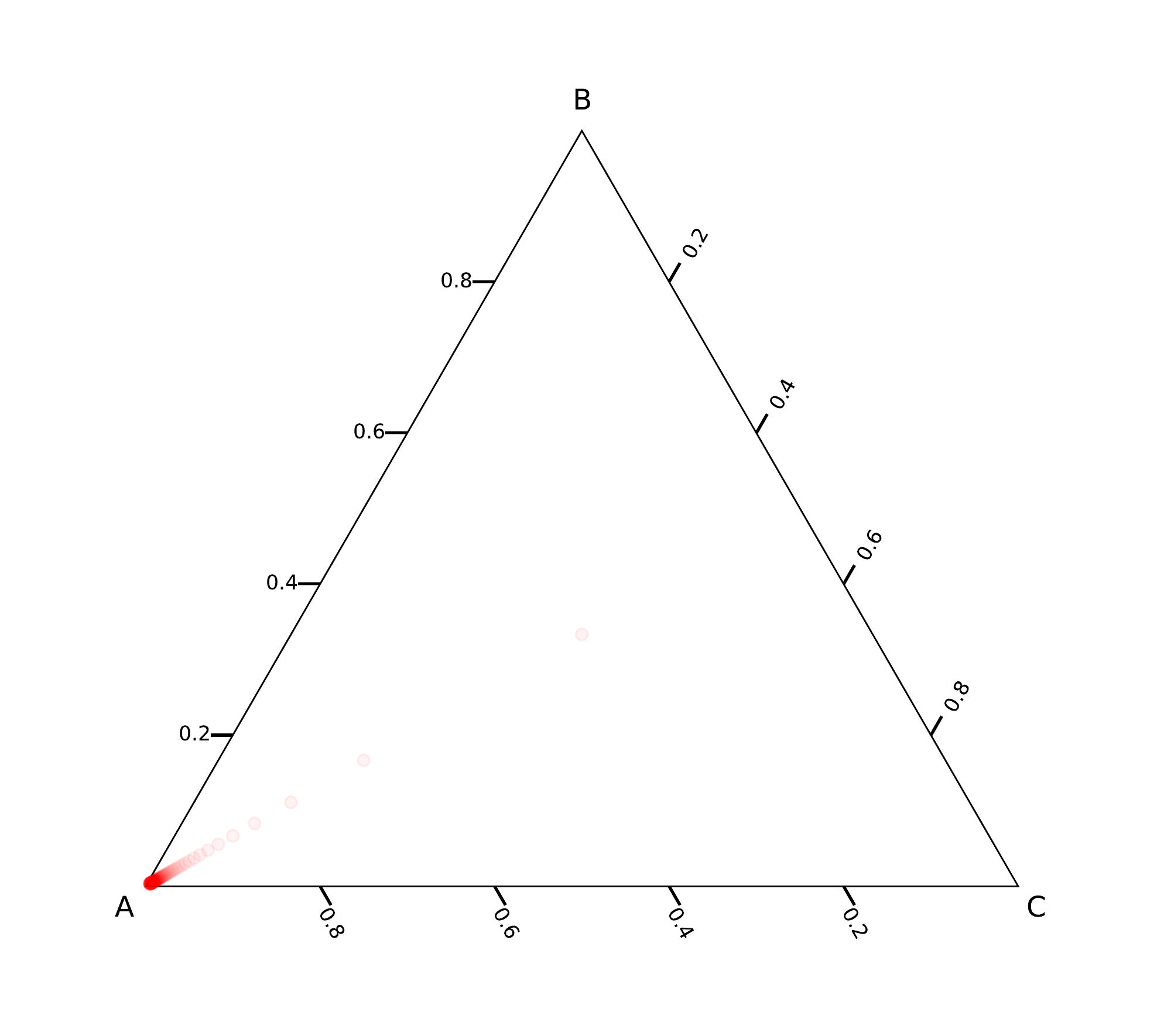}
        \caption{Joint Fictitious Play on the Dominated Strategy Game - Center start.}
        \label{fig:jfp_dom}
    \end{subfigure}
    \hfill
    \begin{subfigure}[b]{0.4\textwidth}
        \centering
        \includegraphics[scale=0.3]{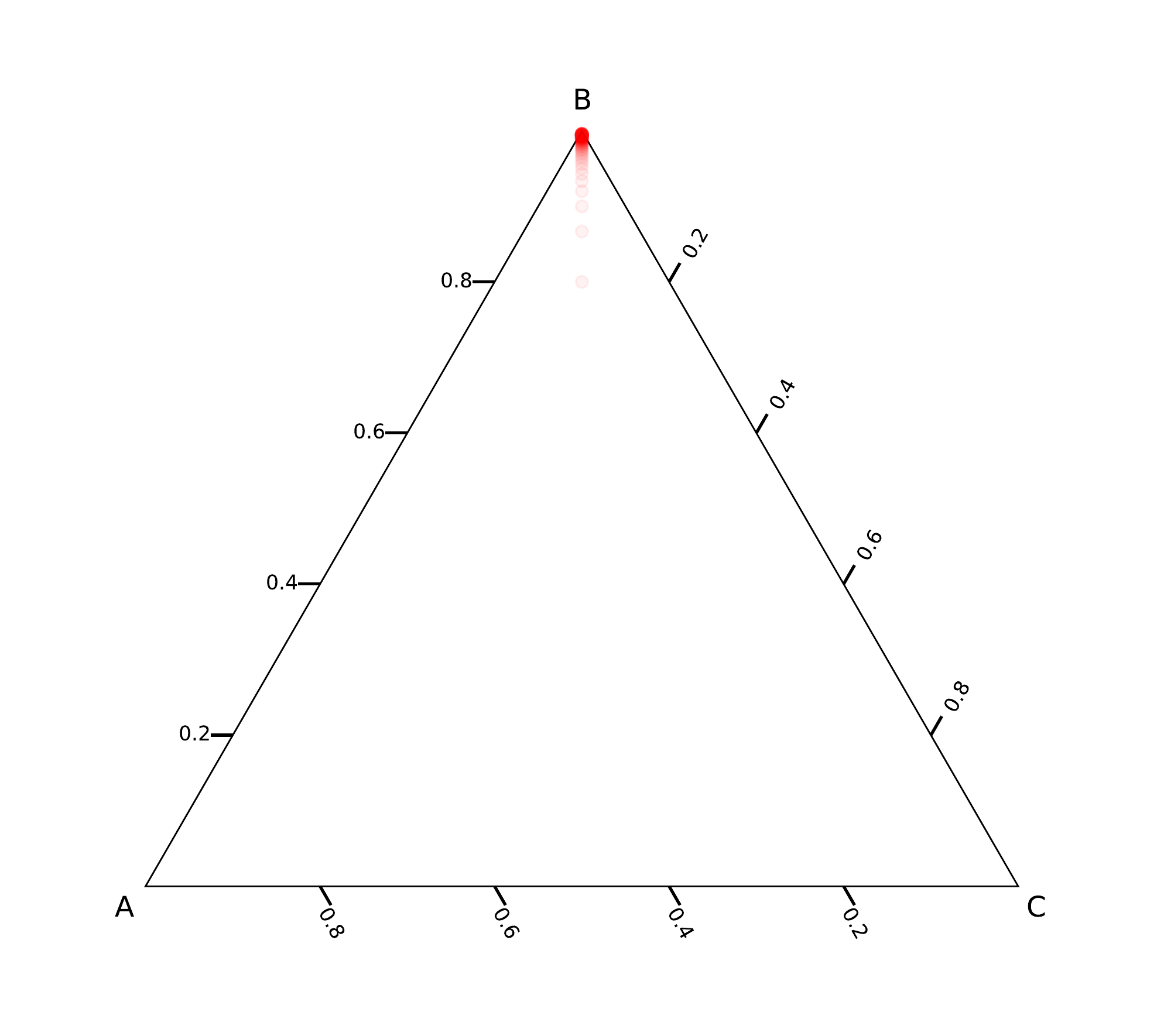}
        \caption{Joint Fictitious Play on the Dominated Strategy Game - Biased start towards B.}
        \label{fig:jfp_dom_up}
    \end{subfigure}
    \vfill
    \begin{subfigure}[b]{0.4\textwidth}
        \centering
        \includegraphics[scale=0.3]{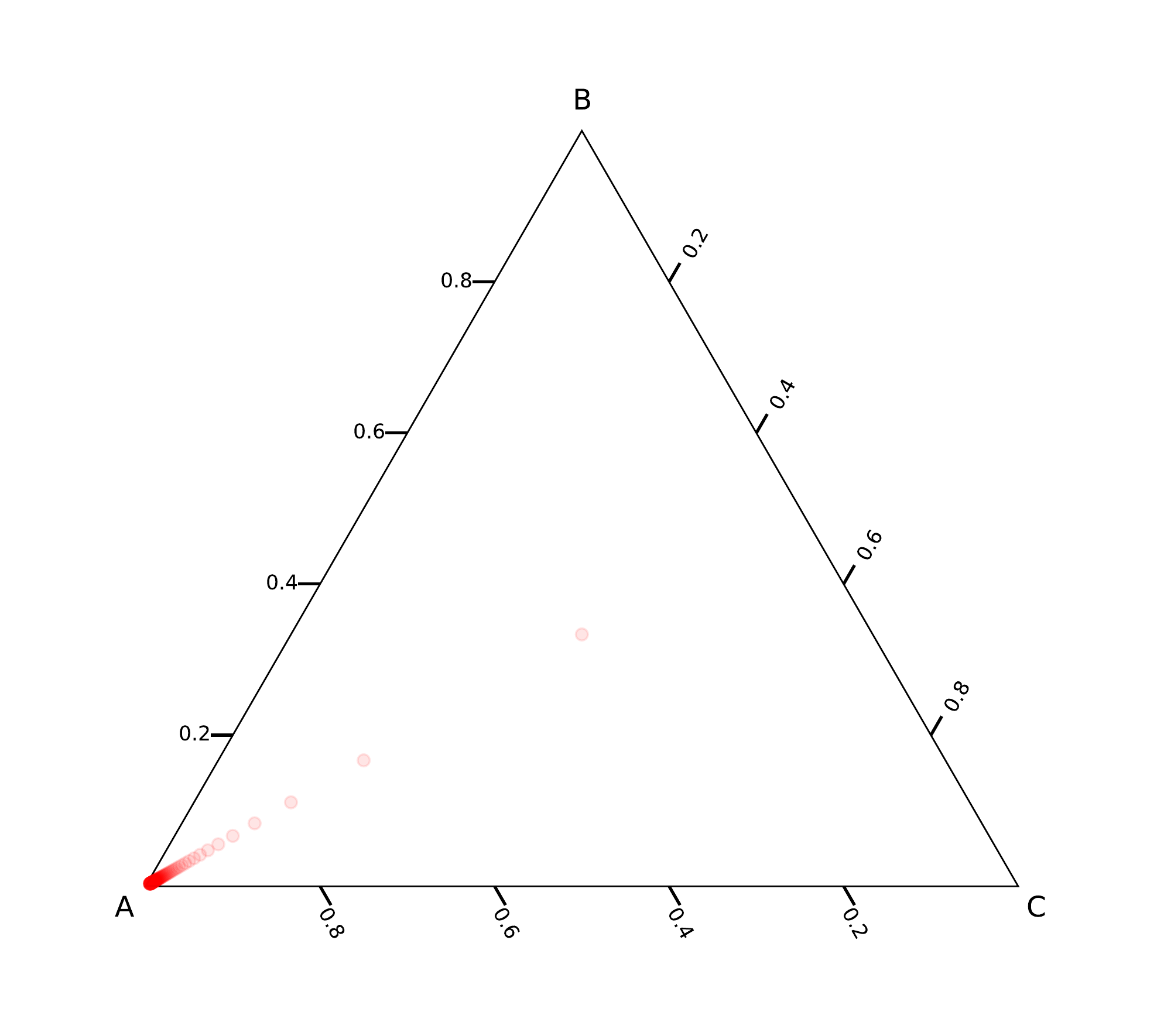}
        \caption{Joint Fictitious Play on the Almost-Dominated Strategy Game.}
        \label{fig:jfp_almost_dom}
    \end{subfigure}
    \hfill
    \begin{subfigure}[b]{0.4\textwidth}
        \centering
        \includegraphics[scale=0.3]{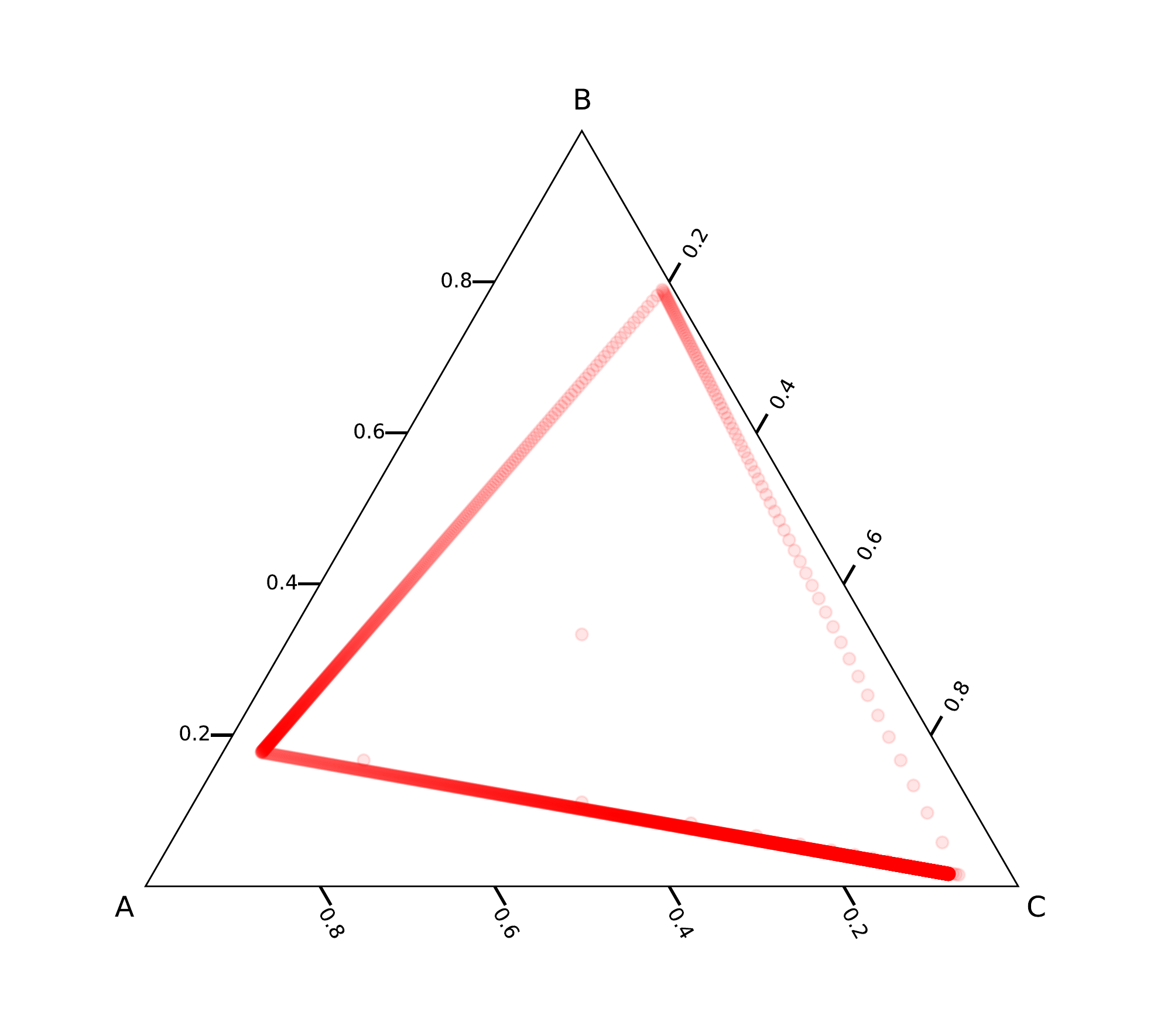}
        \caption{Joint Fictitious Play on the Biased Rock-Paper-Scissors Game.}
        \label{fig:jfp_rps}
    \end{subfigure}
    \caption{Joint Fictitious Play on several Normal-Form Mean-Field Games. Each red circle represents Joint FP's policy at a given step.}
\end{figure}

\begin{figure}
    \centering
    \includegraphics[scale=0.6]{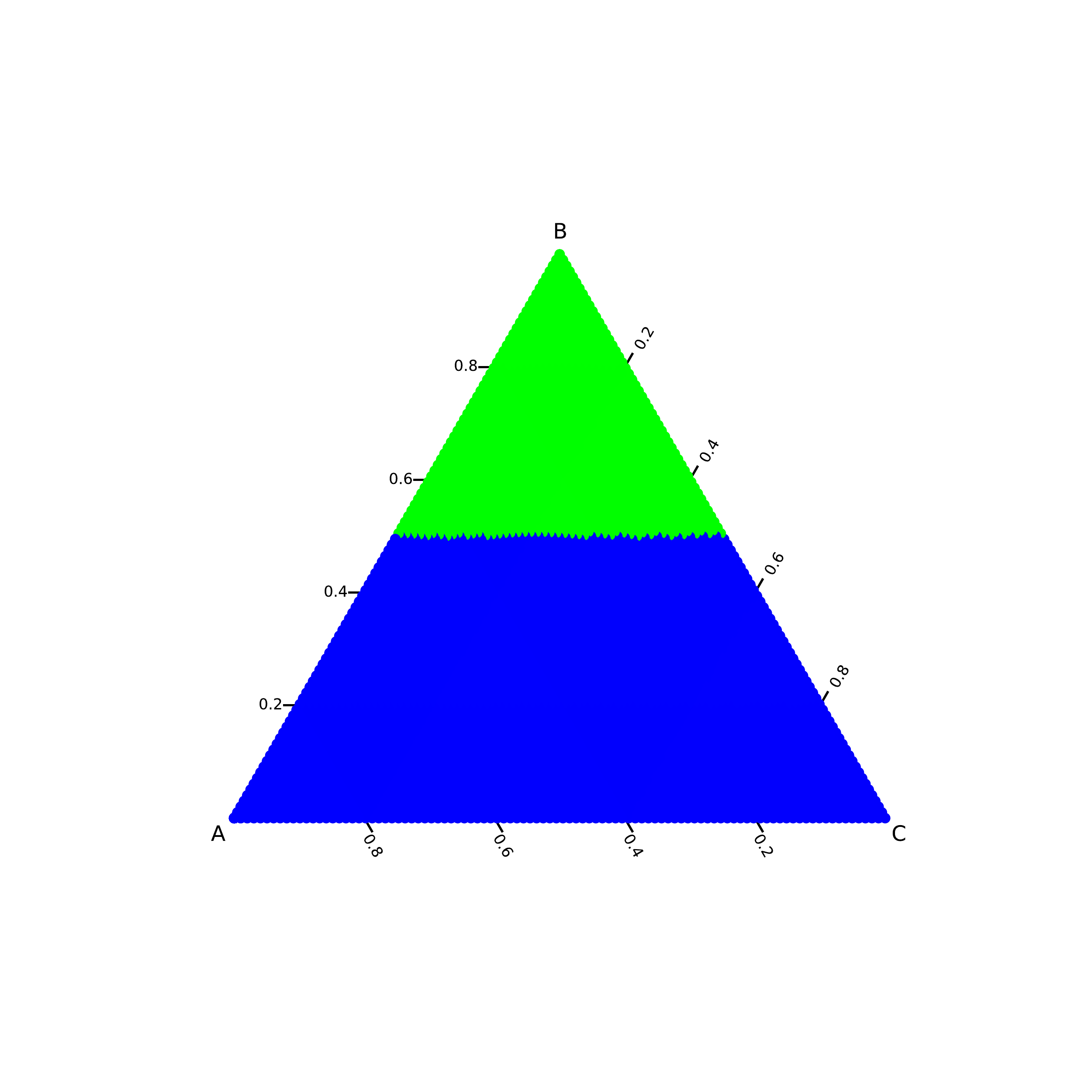}
    \caption{Joint Fictitious Play (JFP) map of starting-points to converged-points on the Almost-dominated Strategy game. Each point represents the starting policy of JFP, and each color, its final policy. Colors are computed following $\pi(A) \text{B} + \pi(B) \text{G} + \pi(C) \text{R}$, where $\pi$ is the final policy, and R, G and B are the primary colors.}
    \label{fig:jfp_dom_many}
\end{figure}

\subsection{Mean-Field PSRO}

\begin{figure}
    \centering
    \begin{subfigure}[b]{0.3\textwidth}
        \centering
        \includegraphics[scale=0.2]{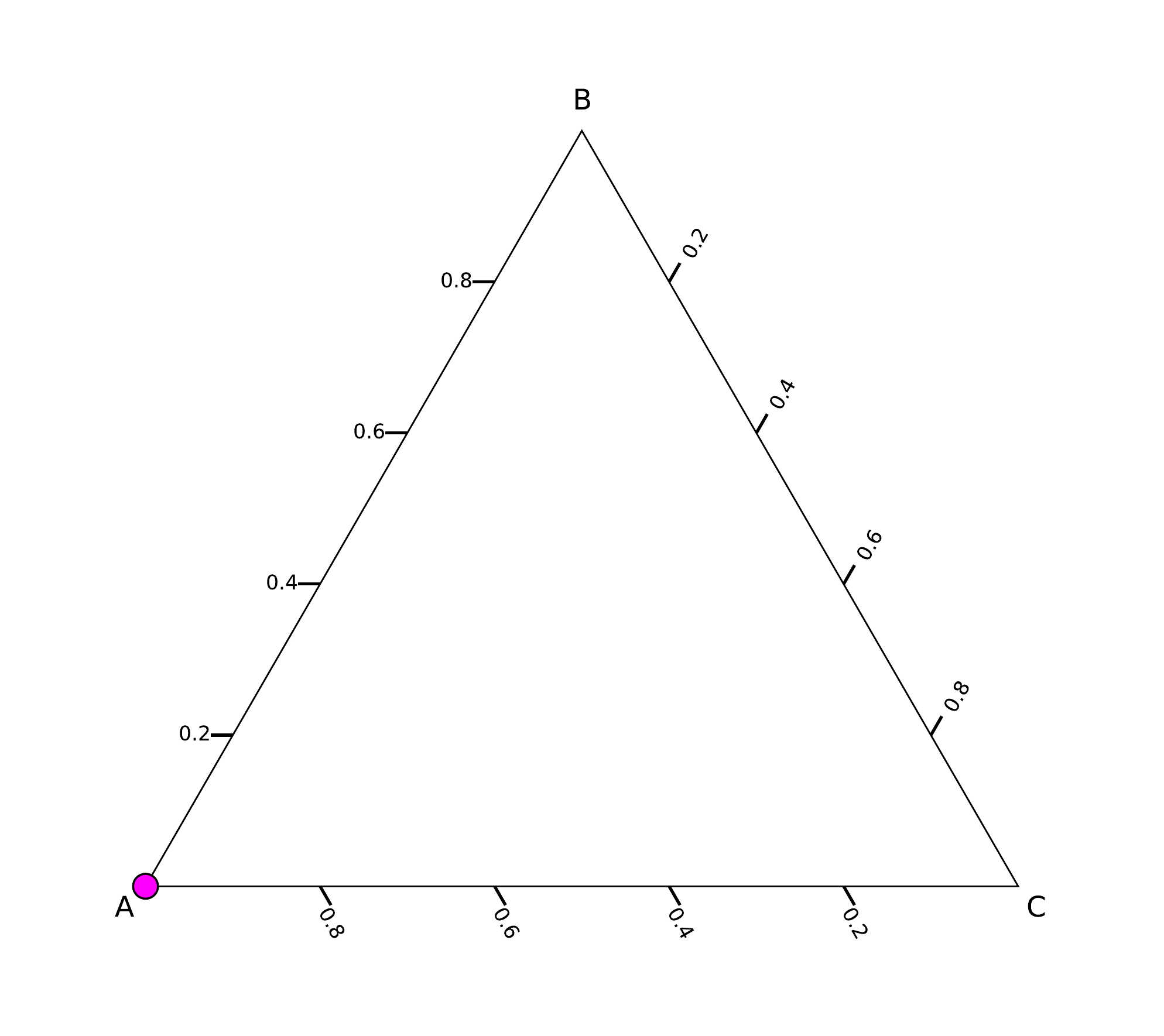}
        \caption{PSRO(CCE) on the Dominated Action game.}
    \end{subfigure}
    \hfill
    \begin{subfigure}[b]{0.3\textwidth}
        \centering
        \includegraphics[scale=0.2]{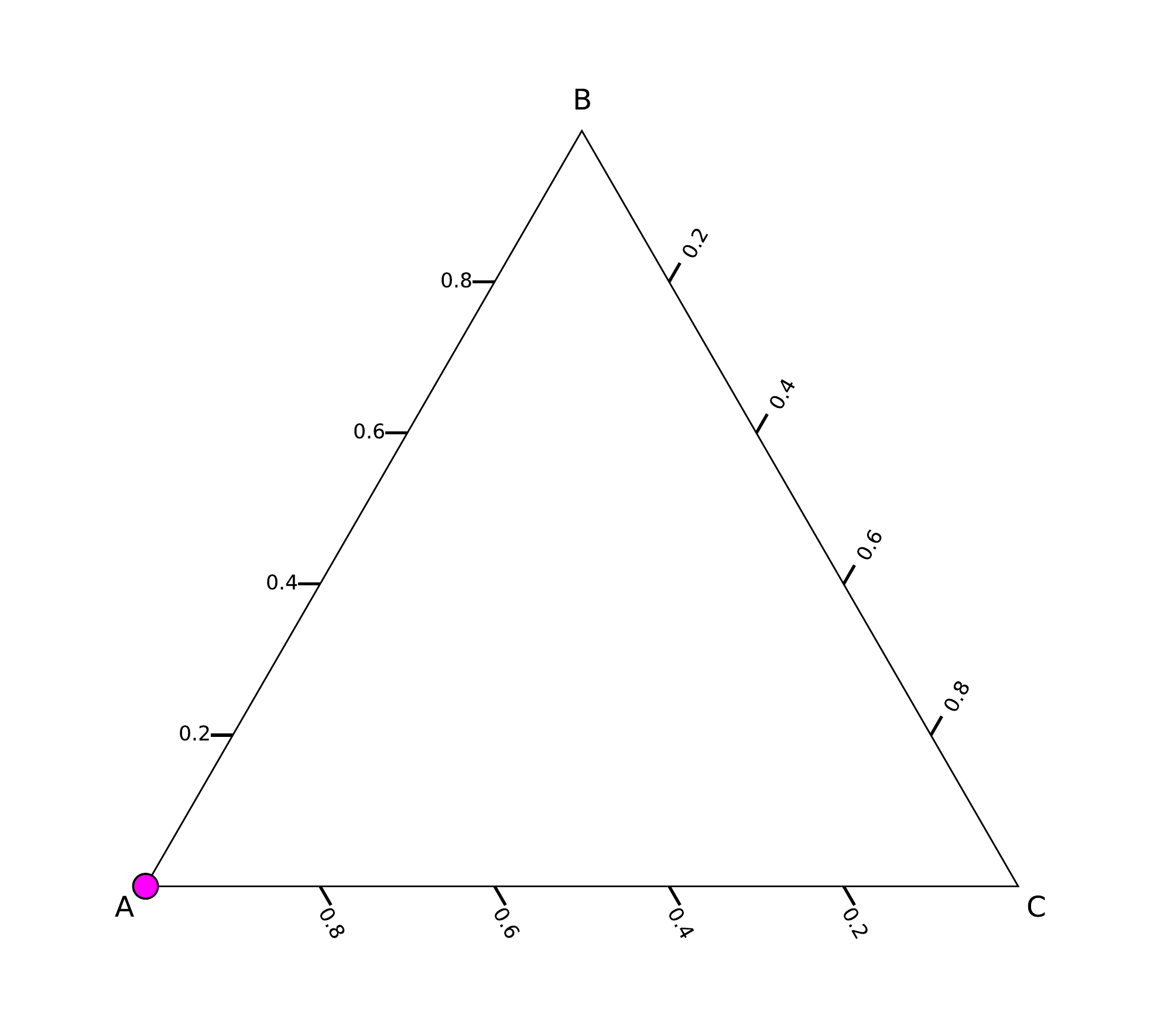}
        \caption{PSRO(CCE) on the Almost-Dominated Action game.}
    \end{subfigure}
    \hfill
    \begin{subfigure}[b]{0.3\textwidth}
        \centering
        \includegraphics[scale=0.2]{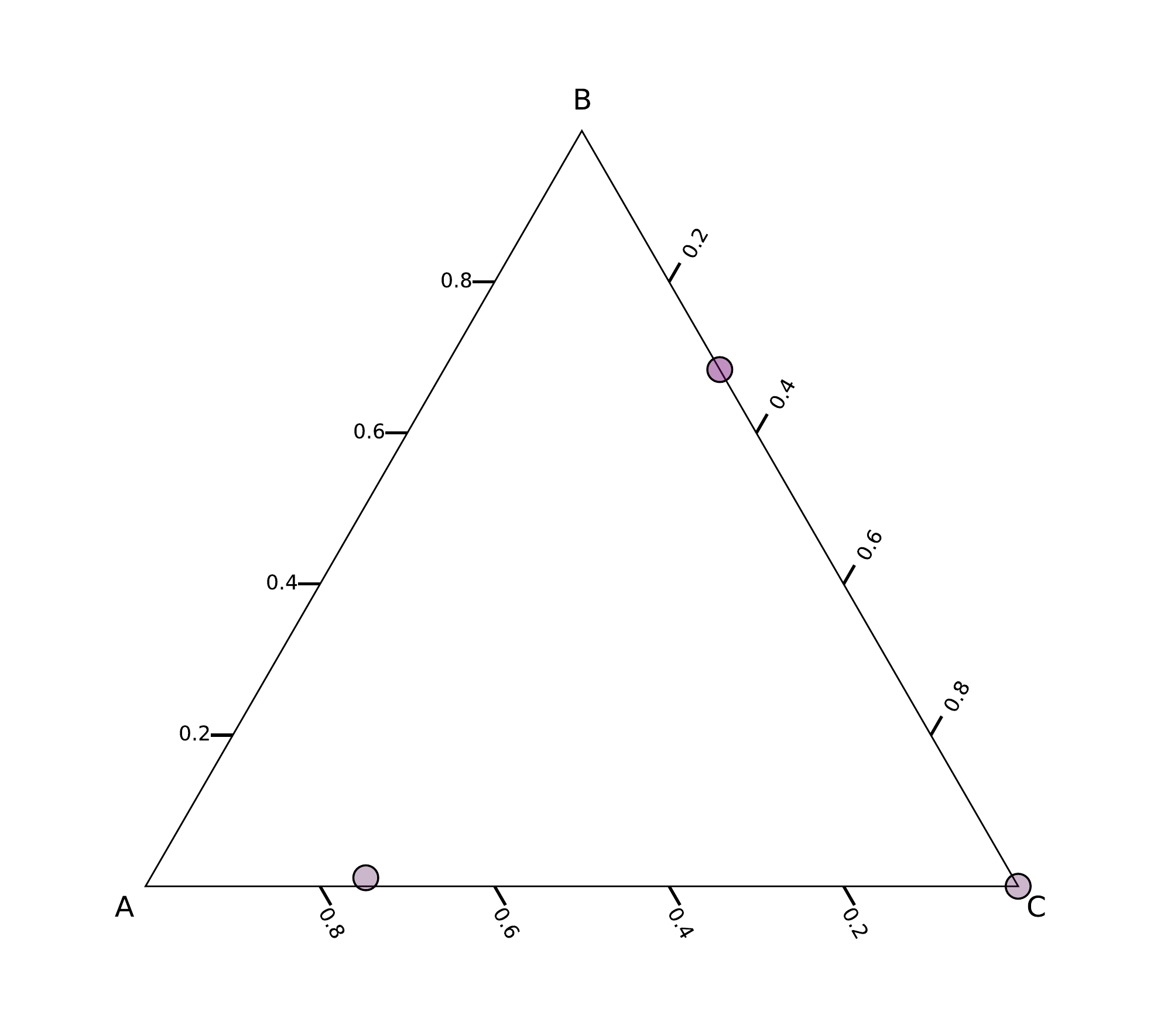}
        \caption{PSRO(CCE) on the Biased Rock-Paper-Scissors game.}
    \end{subfigure}
    \vfill

    \begin{subfigure}[b]{0.3\textwidth}
        \centering
        \includegraphics[scale=0.2]{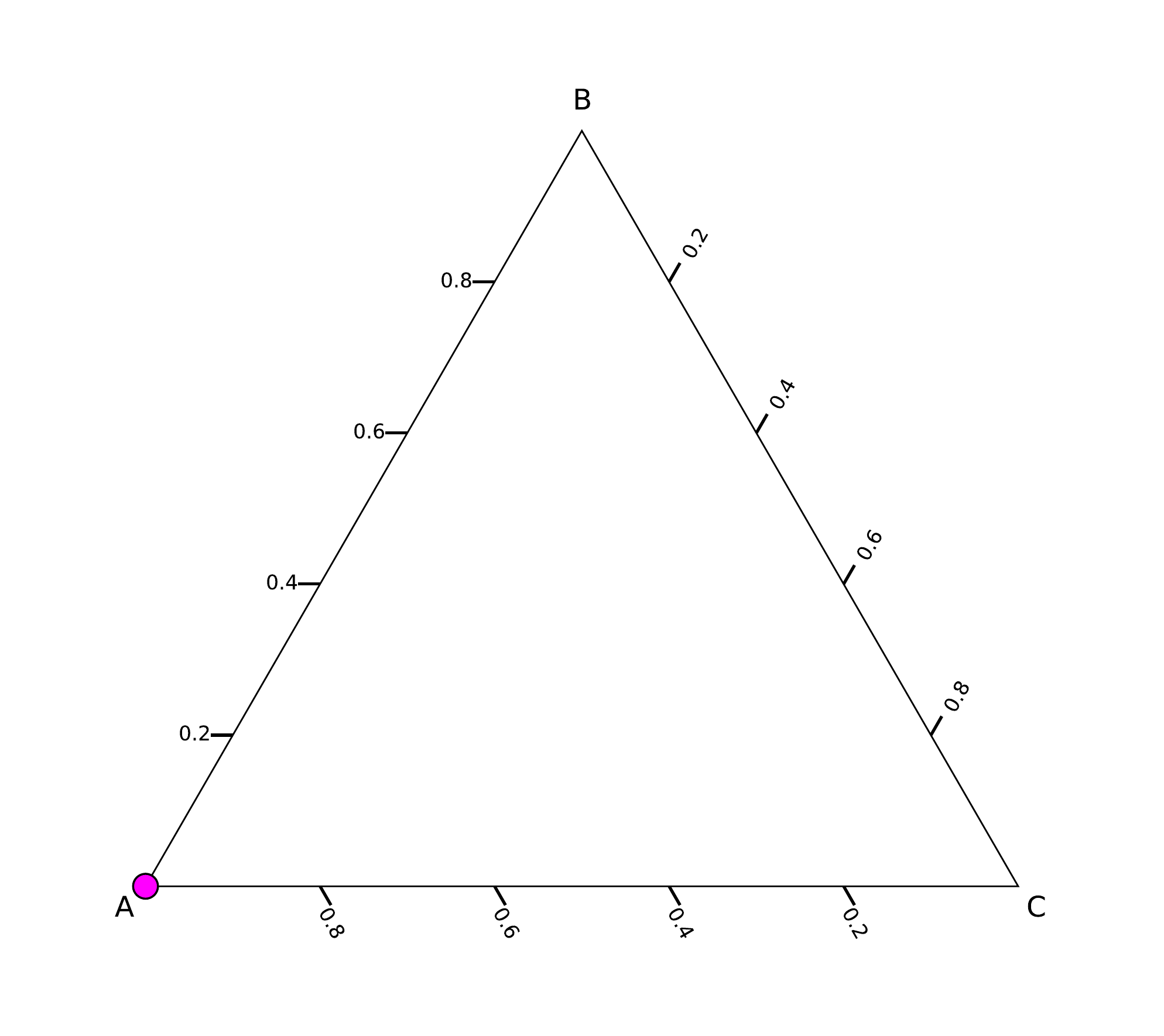}
        \caption{PSRO(CE) on the Dominated Action game.}
    \end{subfigure}
    \hfill
    \begin{subfigure}[b]{0.3\textwidth}
        \centering
        \includegraphics[scale=0.2]{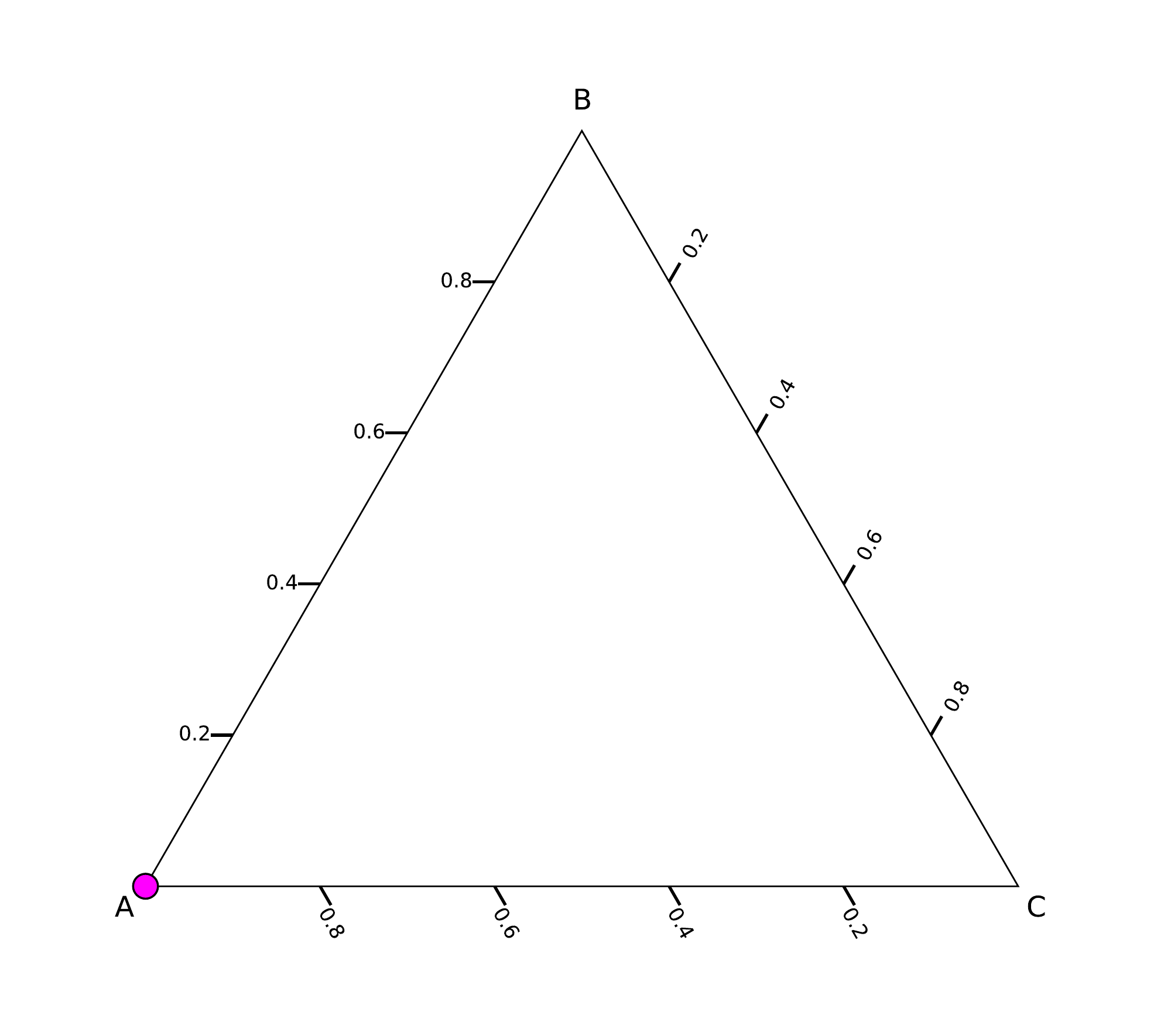}
        \caption{PSRO(CE) on the Almost-Dominated Action game.}
    \end{subfigure}
    \hfill
    \begin{subfigure}[b]{0.3\textwidth}
        \centering
        \includegraphics[scale=0.2]{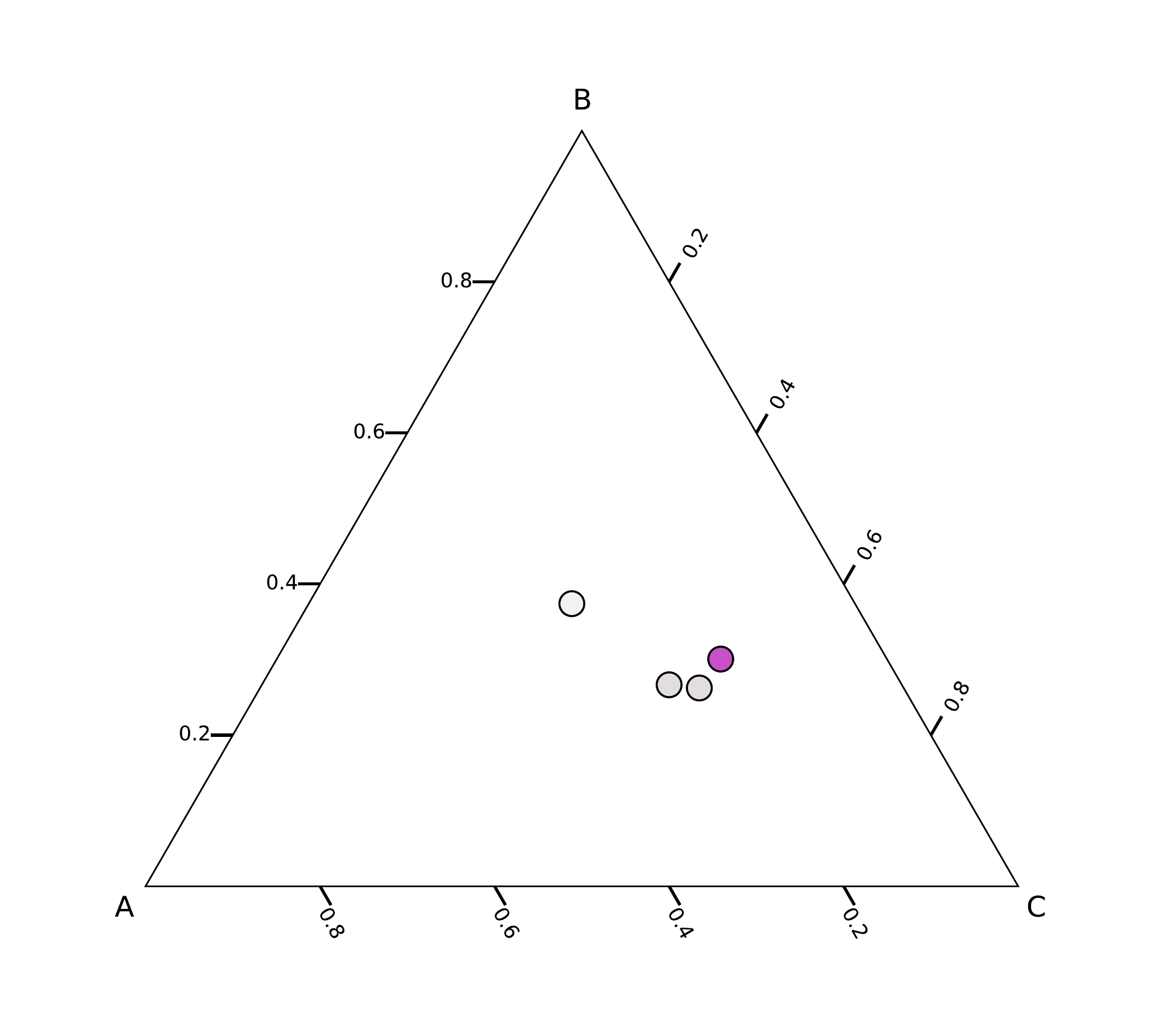}
        \caption{PSRO(CE) on the Biased Rock-Paper-Scissors game.}
    \end{subfigure}
    \vfill

    \begin{subfigure}[b]{0.3\textwidth}
        \centering
        \includegraphics[scale=0.2]{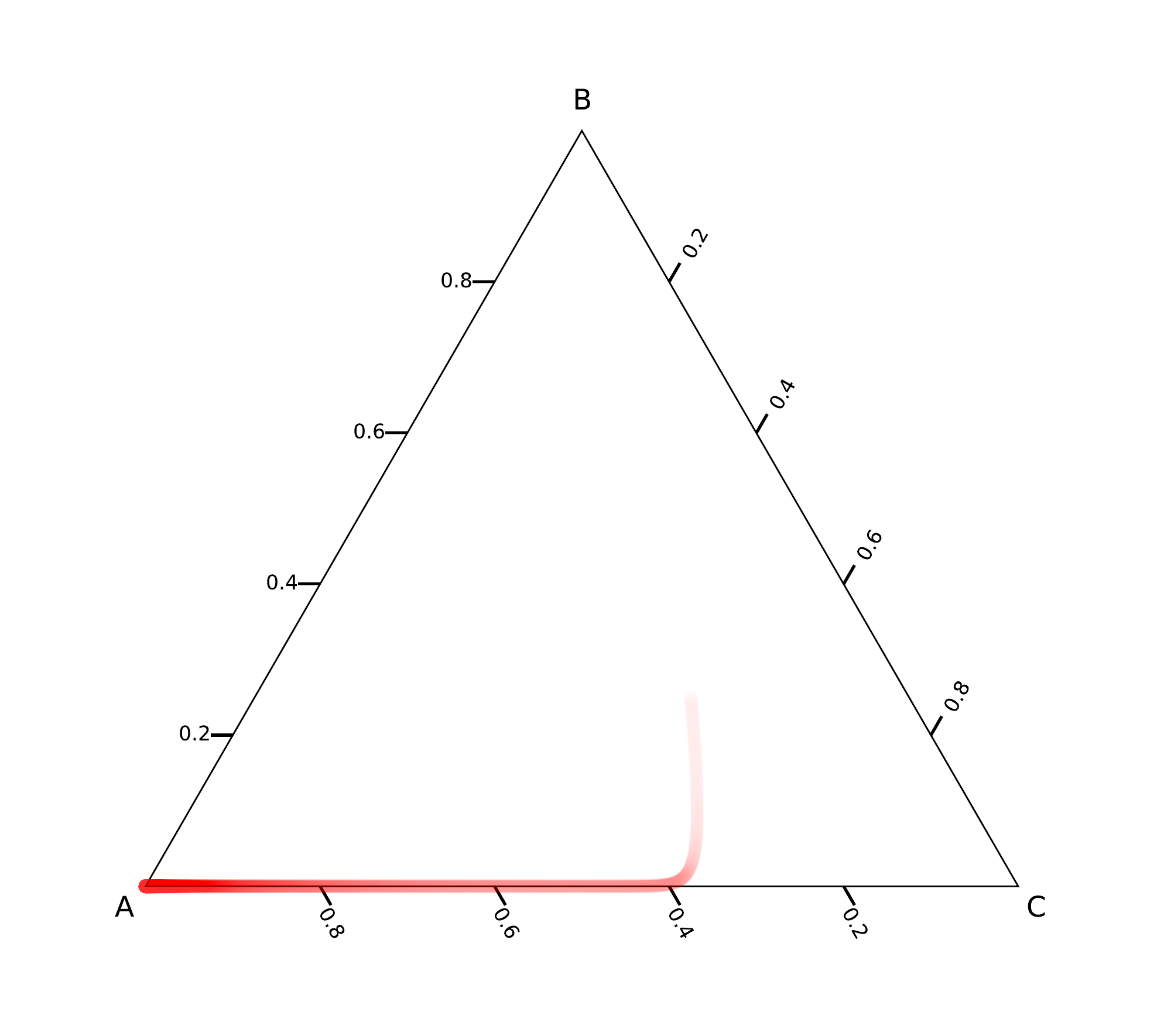}
        \caption{Polynomial Weights (PW) on the Dominated Action game.}
    \end{subfigure}
    \hfill
    \begin{subfigure}[b]{0.3\textwidth}
        \centering
        \includegraphics[scale=0.2]{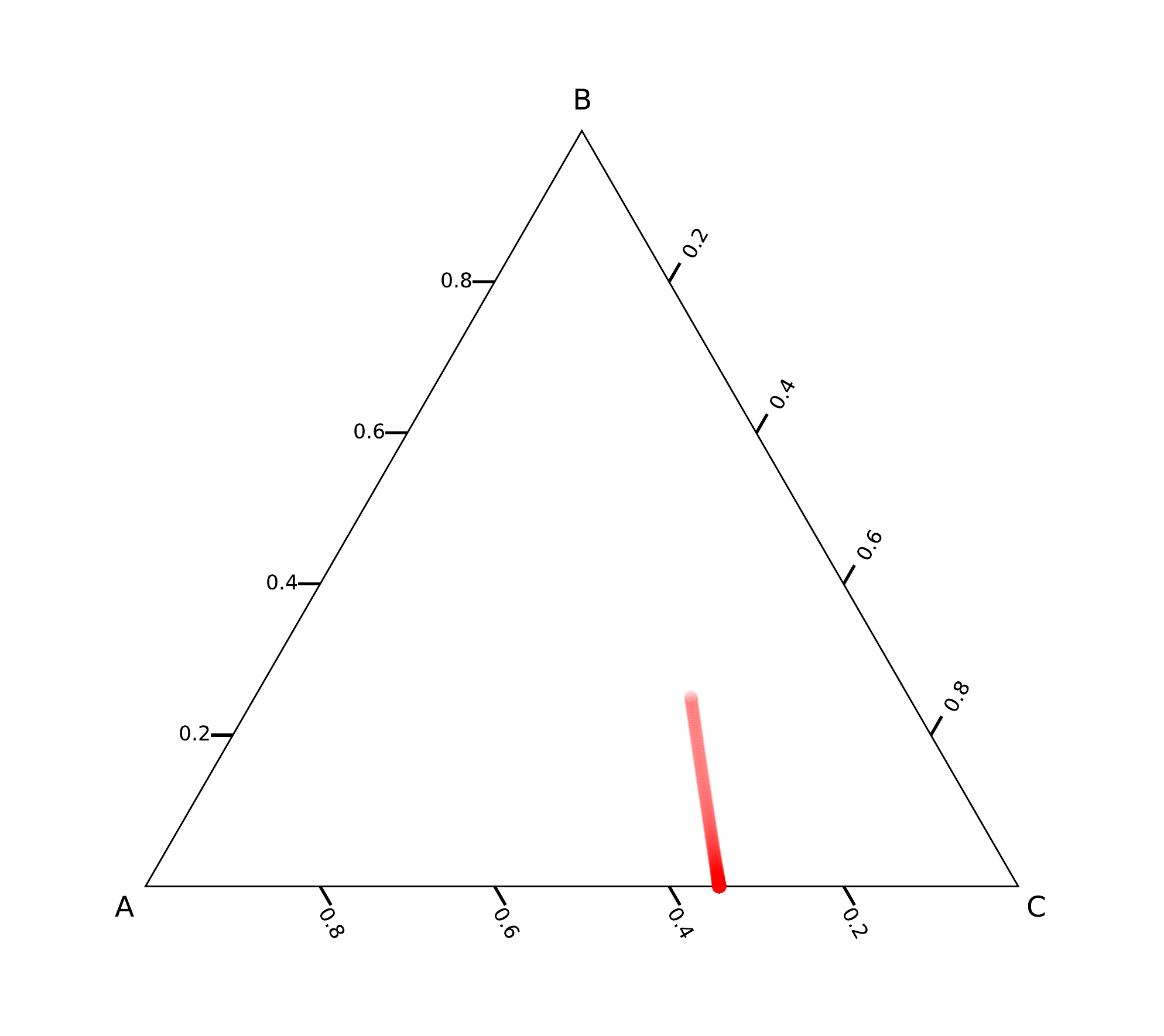}
        \caption{PW on the Almost-Dominated Action game.}
    \end{subfigure}
    \hfill
    \begin{subfigure}[b]{0.3\textwidth}
        \centering
        \includegraphics[scale=0.2]{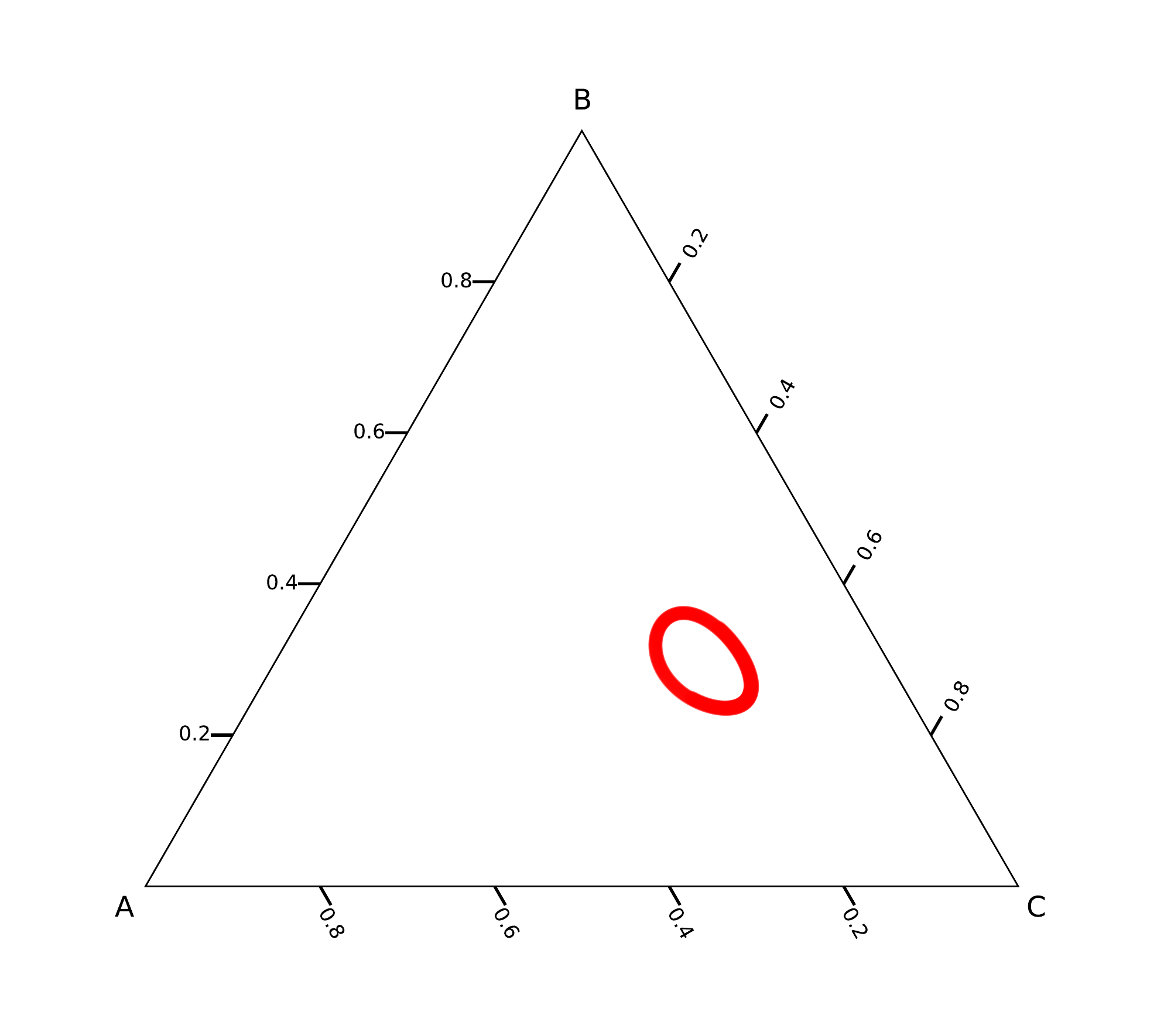}
        \caption{PW on the Biased Rock-Paper-Scissors game.}
    \end{subfigure}
    \vfill

    \begin{subfigure}[b]{0.3\textwidth}
        \centering
        \includegraphics[scale=0.2]{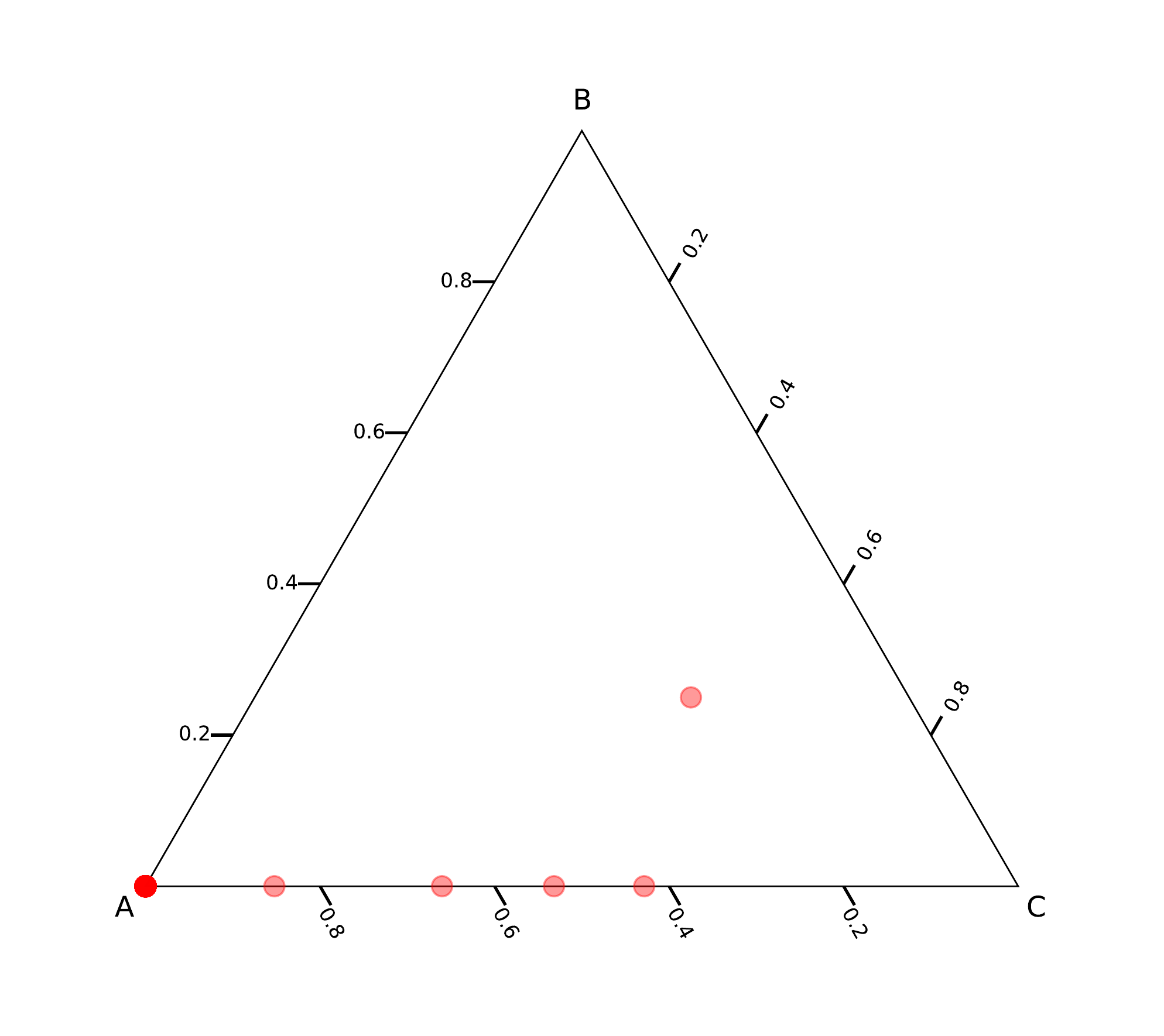}
        \caption{Regret Matching (RM) on the Dominated Action game.}
    \end{subfigure}
    \hfill
    \begin{subfigure}[b]{0.3\textwidth}
        \centering
        \includegraphics[scale=0.2]{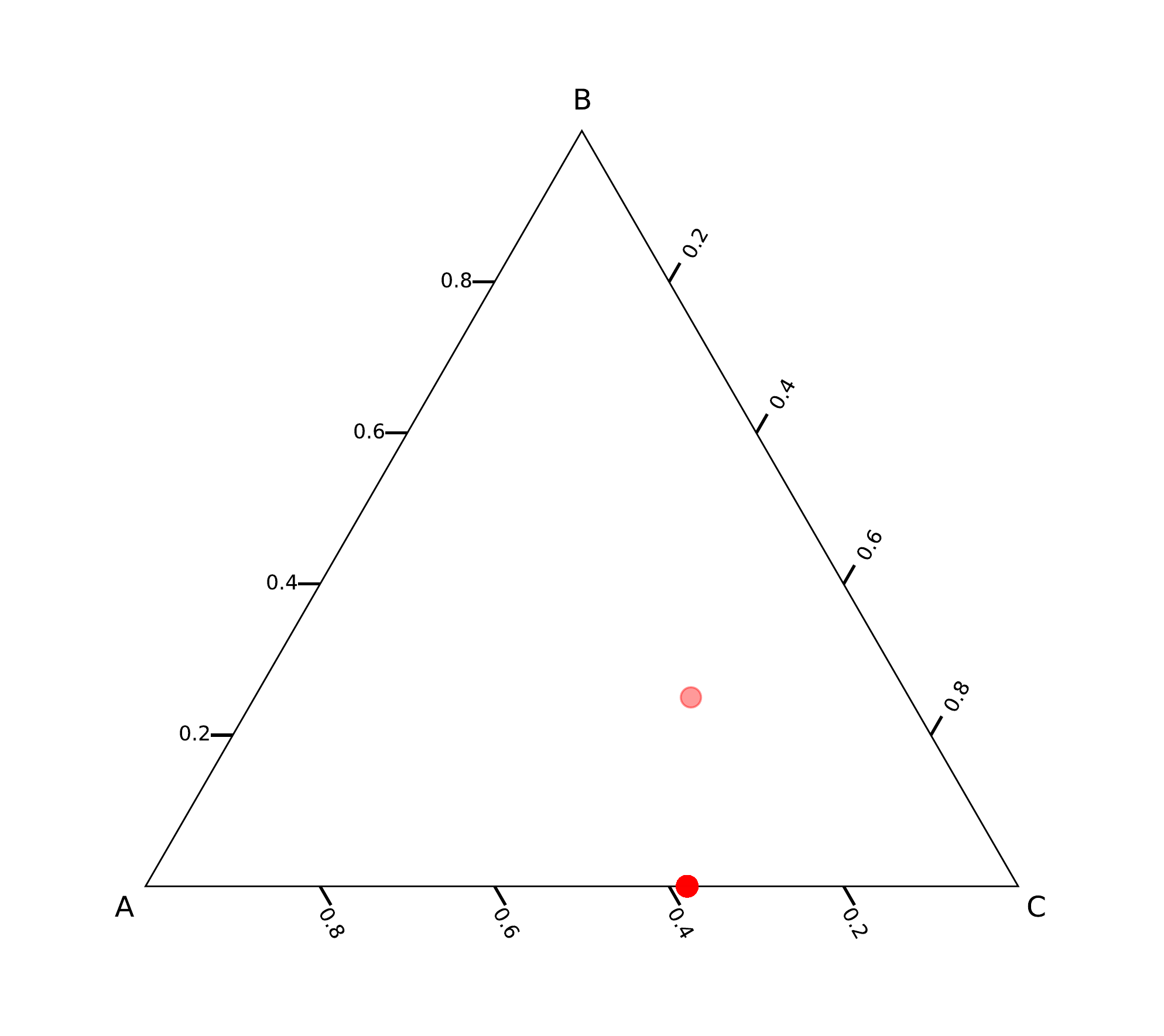}
        \caption{RM on the Almost-Dominated Action game.}
    \end{subfigure}
    \hfill
    \begin{subfigure}[b]{0.3\textwidth}
        \centering
        \includegraphics[scale=0.2]{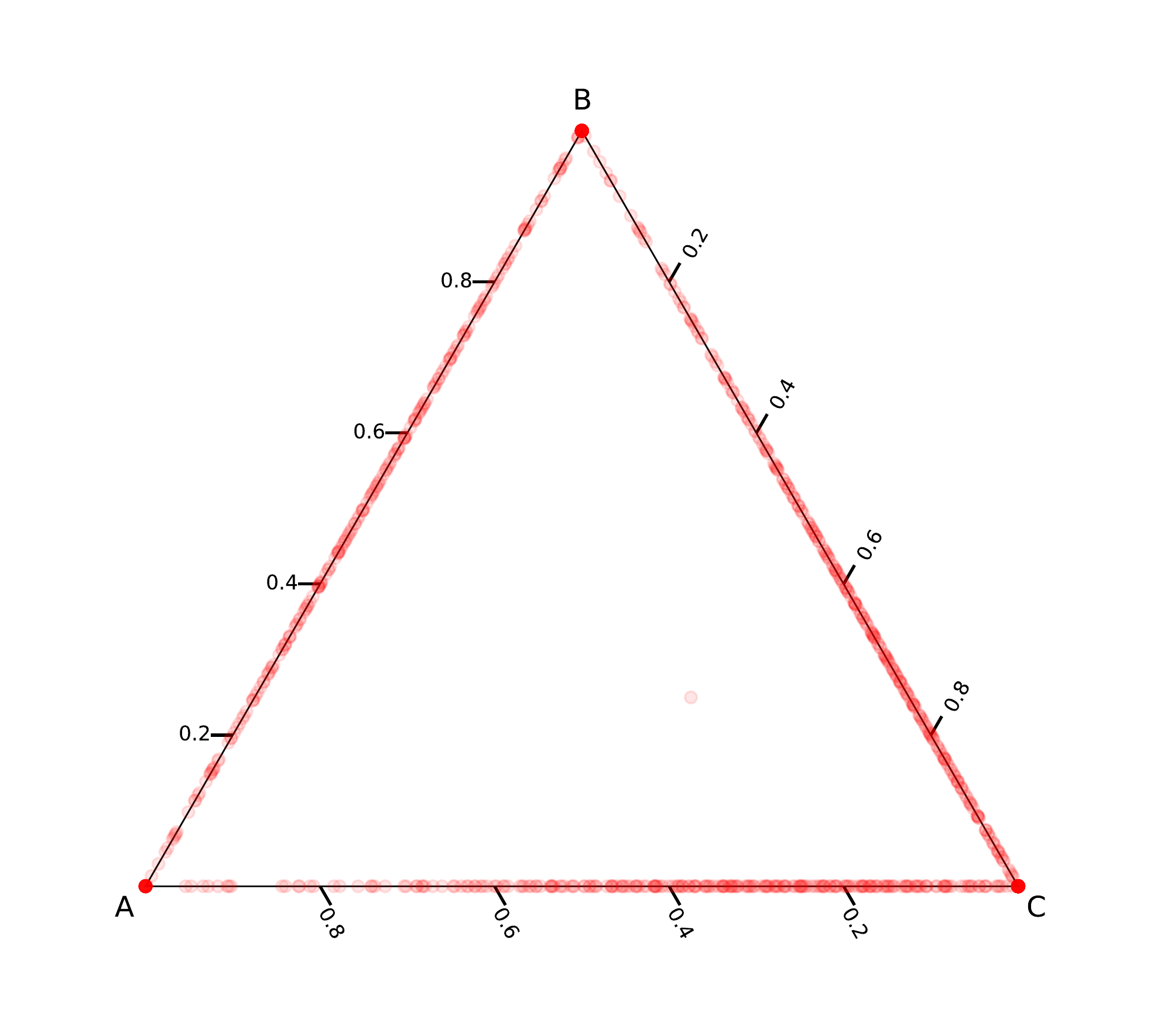}
        \caption{RM on the Biased Rock-Paper-Scissors game.}
    \end{subfigure}
    \caption{PSRO and PSRO regret minimizers on several Normal-Form Mean-Field Games. The PSRO plots show the final equilibrium learnt by PSRO. For PW and RM, each red circle represents OMD's policy at a given step.}
    \label{fig:psro_examples}
\end{figure}

We show on Figure~\ref{fig:psro_examples} several results regarding PSRO. 

The first row shows the equilibria found by Mean-Field PSRO(CCE). We represent each policy played by the equilibrium, and change the policy's color from black to red the more present it is in the mixture. We notice that Mean-Field PSRO removes the dominated action in the dominated action game, and yields an interesting equilibrium for biased Rock-Paper-Scissors. 

The second row shows the same results as the first for PSRO(CE). Here, we see exactly the same behavior as PSRO(CCE) for both the dominated and the almost-dominated action games; however, equilibrium shape changes drastically for the biased RPS game: instead of recommending three almost pure strategies, as did PSRO(CCE) - deviations wouldn't be able to tell which strategy is being recommended, so this is a sensible CCE -, PSRO(CE) is forced to recommend strategies closer to optimality (Though not necessarily optimal !) so as to reach an actual CE. 

On the third and fourth row, we represent the trajectories that respectively the polynomial weights algorithm and the regret matching algorithm, both no-adversarial-regret algorithms, produce when starting with all three actions on different games. We note how much faster regret matching is at finding equilibria, a property that has already been empirically shown in $N$-player games in \cite{tammelin2014solving}. We note that these trajectories were generated without bandit compression, a speedup algorithm introduced in \cite{muller2021learning}. We notice that despite speed and trajectory differences, regret matching and the polynomial weights algorithm yield similar results on the first two (almost-)dominated-action games, whereas their behavior fundamentally differs on biased RPS.

\section{Conclusion}

Mean-Field games stand a lot to benefit from the introduction of equilibrium concept other than Nash, and this work has contributed to introducing and enunciating the properties of two main new types: correlated and coarse correlated equilibria, through a natural generalization of the $N$-player case of anonymous-symmetric games.

We have proven that they can be reused with $\mathcal{O}\left( \frac{1}{\sqrt{N}} \right)$ optimality loss in $N$-player games, connected them with existing notions of correlated equilibrium in the literature. We linked the equilibria characterization to notions of regret and showed that three popular algorithms, Policy Space Response Oracle (PSRO), Online Mirror Descent (OMD) and a slight alteration of Fictitious Play, Joint Fictitious Play (Joint FP) converge to Correlated or Coarse Correlated equilibria. Noteworthily, such convergence happens in a large class of Mean Field Games and provides insights on the asymptotic behavior of these learning algorithms, when several Nash equilibria exist. 
We have also shown the relationship between OMD's, Joint FP's and PSRO's reached equilibria and dominated strategies - namely in which circumstances they exclude those, and provided empirical demonstrations of their asymptotic behavior.

There are a variety of interesting directions for future research; applications of correlated and coarse correlated equilibria are a natural evolution of this work. Indeed, the social welfare of many situations in real life can be increased simply by the addition of a correlating signal (traffic signaling being one such example). Studying what characterizes this type of games, and how to reach high-welfare equilibria consistently constitute promising research directions. Increasing the resolution of the recommender, by introducing notions of extensive-form correlated equilibria where action recommendations are provided at each state, is also a promising research direction. %

\newpage
\section*{Acknowledgement}

We wish to thank Sertan Girgin and Raphael Marinier for technical support through OpenSpiel~\cite{lanctot2020openspiel}, and Ian Gemp for thoughtful discussions.

\bibliographystyle{plainnat}
\bibliography{reference.bib}

\appendix
\input{appendix}

\end{document}

%% file: appendix.tex
\newpage
\begin{center}
    \huge{\textbf{Appendices}}
\end{center}

\section{Equilibrium equivalence in N-player games - Theorem \ref{theorem:n_player_equivalence}}\label{Appendix_Sec_A}

We begin by recalling Theorem~\ref{theorem:n_player_equivalence}.

\begin{unnum_theorem}[Equilibrium Equivalence]\label{}
    In symmetric-anonymous N-player games, there is one to one correspondence between Symmetric-Anonymous $\epsilon$-(C)CE and $\epsilon$-(C)CE with symmetric correlating device, i.e. such that $\rho\in\Delta_{sym}(\Pi^N)$.
\end{unnum_theorem}

\begin{proof}[Proof of Theorem \ref{theorem:n_player_equivalence}]

Let $\Pi = \Pi_1 = \dots = \Pi_N$. For $\pi \in \Pi^N$, and $\bar\rho \in \Delta(\Pi^N)$ a classical and \emph{symmetric} (coarse) correlated equilibrium, let $\bar\rho_{\pi_i}$ be the conditional distribution on $\Pi^{N-1}$ given player $i$ is recommended policy $\pi_i$, and let $\nu^{N}_{\pi} = \frac{1}{N}\sum_{j=1}^{N} \delta_{\pi_{j}}$ be the empirical population distribution of $\pi \in \Pi^N$. From Equation~\ref{eq:j_jmathcal}, we have that 
\[
    \mathcal{J}(\pi_i,\nu^{N-1}_{\pi_{-i}}) := J_i(\pi_i, \pi_{-i}). %
\]

and    

\[
    \mathcal{J}(\pi_i, \nu^{N-1}_{\pi_{-i}})
    =
    \mathcal{J}\left(\pi_i, \frac{N}{N-1} (\nu^N_{\pi} - \frac{1}{N}\delta_{\pi_i})\right),
\]
so this quantity only depends on $\nu^N_{\pi}$ and $\pi_i$. Moreover, it will be useful to consider empirical distributions containing a given policy. More precisely, note that, for $\nu^N \in \Delta_N(\Pi)$, 
\begin{equation}
    \label{eq:equiv-pii-DeltaNm1}
    \exists \tilde\pi \in \Pi^N  \hbox{ s.t. } \tilde\pi_i = \pi_i \hbox{ and }  \nu^N_{\tilde\pi} = \nu^N
    \Leftrightarrow 
    \underbrace{\frac{N}{N-1} (\nu^N - \frac{1}{N}\delta_{\pi_i})}_{=: \nu^N_{-\pi_i}} \in \Delta_{N-1}(\Pi),
\end{equation}
meaning that $\pi_i$ is a point in the support of the empirical distribution $\nu^N$ if and only if $\frac{N}{N-1} (\nu^N - \frac{1}{N}\delta_{\pi_i}) \in \Delta_{N-1}(\Pi)$ is an empirical distribution with $N-1$ points. Let us denote by $\Delta_N(\Pi, \pi_i)$ the set of empirical distributions $\nu^N$ satisfying the above condition, \textit{i.e.},
\[
    \Delta_N(\Pi, \pi_i) = \Big\{ \nu^N \in \Delta_N(\Pi) : \frac{N}{N-1} (\nu^N - \frac{1}{N}\rho_{\pi_i}) \in \Delta_{N-1}(\Pi) \Big\}.
\]
For simplicity, we denote:
\[
    \nu^N_{-\pi_i} = \frac{N}{N-1} (\nu^N - \frac{1}{N}\rho_{\pi_i}),
    \qquad 
    \nu^{N-1}_{+\pi_i} = \frac{N-1}{N}(\nu^{N-1} + \frac{1}{N-1} \rho_{\pi_i}).
\]

For $\bar\rho \in \Delta(\Pi^N)$, let $\rho \in \Delta(\Delta_N(\Pi))$ be the distribution over empirical distributions induced by $\bar\rho$, i.e., for every $\nu^{N} \in \Delta_N(\Pi)$
\[
    \rho(\nu^{N}) 
    = \bar\rho\left(\{\tilde\pi \in \Pi^N : \nu^N_{\tilde\pi} = \nu^N\}\right).
\]

Since we assume that $\bar\rho$ is symmetric, we can say that for all $\pi$, 
\[
\sum_\pi \bar\rho(\pi) \mathcal{J}(u(\pi_i), \nu^{N-1}_{\pi_{-i}}) = \sum_\pi \sum_{\pi_i} \bar\rho(\pi) \frac{N_{\pi_i \in \pi}}{N} \mathcal{J}(u(\pi_i), \nu^{N-1}_{\pi_{-i}})
\]
where $N_{\pi_i \in \pi}$ is the number of players playing $\pi_i$ when the joint policy is $\pi$. What this means is that if $\bar\rho$ recommends $\pi$, then it will also recommend all possible permutations thereof, hence the expected payoff for player $i$ when a given policy $\pi$ is recommended is the same as the expected payoff averaged over all players when $\pi$ is recommended.

\begin{align*}
    \mathbb{E}_{\pi \sim \bar\rho} [J_i(u(\pi_i), \pi_{-i})] &= \sum_\pi \bar\rho(\pi) J_i(u(\pi_i), \pi_{-i}) \\
    &= \sum_\pi \bar\rho(\pi) \mathcal{J}(u(\pi_i), \nu^{N-1}_{\pi_{-i}}) \\
    &= \sum_\pi \sum_{\pi_i} \bar\rho(\pi) \frac{N_{\pi_i}}{N} \mathcal{J}(u(\pi_i), \nu^{N-1}_{\pi_{-i}}) \\
    &= \sum_{\nu^{N}} \sum_{\pi \mid \nu^N_{\pi} = \nu^{N}} \sum_{\pi_i} \bar\rho(\pi) \nu^N(\pi_i)  \mathcal{J}(u(\pi_i), \nu^{N}_{-\pi_i}) \\
    &= \sum_{\nu^{N}} \underbrace{\sum_{\pi \mid \nu^N_{\pi} = \nu^{N}} \bar\rho(\pi)}_{= \rho(\nu^N)} \sum_{\pi_i} \nu^N(\pi_i)  \mathcal{J}(u(\pi_i), \nu^{N}_{-\pi_i}) \\
    &= \sum_{\nu^N} \sum_{\pi_i} \nu^N(\pi_i) \rho(\nu^N) \mathcal{J}(u(\pi_i), \nu^{N}_{-\pi_i}) \\
    &= \mathbb{E}_{\nu^N \sim \rho, \, \pi_i \sim \nu^N} [\mathcal{J}(u(\pi_i), \nu^N_{-i})],
\end{align*}

which concludes the proof.

\end{proof}

\section{Proof of Existence of (Coarse) Correlated Equilibria}\label{proof:existence_(c)ce}

We begin by recalling Theorem~\ref{thm:cce-existence}:

\begin{unnum_theorem}[(Coarse) Correlated equilibrium existence]
    If the reward function $r$ and the dynamics function $p$ are continuous with respect to $\mu$, then the game admits at least one (coarse) correlated equilibrium.
\end{unnum_theorem}
\begin{proof}
    We begin by recalling Proposition~\ref{prop:nash_to_ce}: if the game admits a Nash equilibrium, then it admits a correlated equilibrium.
    
    We now prove that, under the condition that $r$ is continuous with respect to $\mu$, the game admits at least one Nash equilibrium.
    
    Let $\phi:\Delta(\Pi)\rightarrow 2^{\Delta(\Pi)}$ be the best-response map 
    
    \[
    \forall \nu \in \Delta(\Pi), \quad \phi(\nu) := \argmax\limits_{\nu' \in \Delta(\Pi)} J(\pi(\nu'), \mu(\nu)).
    \] 

    $\Delta(\Pi)$ is non-empty and convex; it is besides closed and bounded in a finite-dimensional space, and therefore compact.
     
    \paragraph{Non-emptiness and convexity of $\phi$:\\}
    
    We define, for all $\nu\in\Delta(\Pi)$, 
    \[
    \argmax\limits_{\nu' \in \Delta(\Pi)} J(\pi(\nu'), \mu(\nu)) \subseteq \Delta(\Pi)
    \] 
    because $\Delta(\Pi)$ is compact. Therefore $\phi(\nu)$ is non-empty: the argmax exists. We now prove that $\phi(\nu)$ is convex.
    
    Let $\nu_1, \; \nu_2 \in \phi(\nu)$, $t \in [0, 1]$. 
    \[
    J(\pi(t \nu_1 + (1-t) \nu_2), \mu(\nu)) = t J(\pi(\nu_1), \mu(\nu)) + (1-t) J(\pi(\nu_2), \mu(\nu))
    \]
    by linearity of $J$ with respect to its first argument.
    This proves us that $t \nu_1 + (1-t) \nu_2 \in \phi(\nu)$, and thus that $\phi(\nu)$ is convex.
    
    \paragraph{Graph($\phi$) closedness:\\}

    $\text{Graph}(\phi) = \{ (\nu, \nu') \in \Delta(\Pi) \times \Delta(\Pi) \; | \; \nu' \in \phi(\nu) \}$. Let $(\nu_k^1, \nu_k^2)_k$ be a sequence of elements of $\text{Graph}(\phi)$ which converges towards $(\nu_*^1, \nu_*^2) \in \Delta(\Pi) \times \Delta(\Pi)$.
    
    $r$ and $p$ are continuous in $\mu$, therefore $J$ is also continuous in $\mu$. Since $J: (\nu_1, \nu_2) \rightarrow J(\pi(\nu_1), \mu(\nu_2))$ is linear in $\nu_1$ because $J\left(\pi(\nu_1), \mu(\nu_2)\right) = \sum_i \nu_1^i J\left(\pi_i, \mu(\nu_2)\right)$, it is also continuous in both variables at the same time. 
    
    Since $J$ is continuous in both variables at $(\nu_*^1, \nu_*^2)$, let $\epsilon > 0$ and $\alpha > 0$ be such that $\forall (\nu_1, \nu_2) \in \Delta(\Pi) \times \Delta(\Pi)$ such that $d\left((\nu_1, \nu_2), (\nu_*^1, \nu_*^2)\right) \leq \alpha$, 
    \[
    \lvert J\left(\pi(\nu_1), \mu(\nu_2)\right) - J(\pi(\nu_*^1), \mu(\nu_*^2)) \rvert \leq \epsilon
    \]
    with $d$ a metric over $\Delta(\Pi) \times \Delta(\Pi)$ under which $J$ is continuous. Let $N_0 > 0$ be such that $\forall n \geq N_0, \; d\left((\nu_k^1, \nu_k^2), (\nu_*^1, \nu_*^2)\right) \leq \alpha$, and let $n \geq N_0$.
    
    By uniform continuity (J is continuous over a compact) and triangle inequality, taking $n$ large enough, for all $\nu \in \Delta(\Pi)$,
    \[
    J\big(\pi(\nu), \mu(\nu_*^2)\big) \leq \epsilon + J\big(\pi(\nu), \mu(\nu_n^2)\big)
    \] 
    where the first line is obtained by uniform continuity of $J$.
    
    \[
    - J\big(\pi(\nu_*^1), \mu(\nu_*^2)\big) \leq \epsilon - J\big(\pi(\nu_n^1), \mu(\nu_n^2)\big)
    \] 
    And by optimality of $\nu_n^1$ against $\mu(\nu_n^2)$, $\forall \nu \in \Delta(\Pi_n)$,
    \[J\big(\pi(\nu), \mu(\nu_n^2)\big) - J\big(\pi(\nu_n^1), \mu(\nu_n^2)\big) \leq 0\]
    We then have, $\forall \nu \in \Delta(\Pi_n)$, 
    \begin{align*}\hspace{-0.2cm}
        J(\pi(\nu), \mu(\nu_*^2)) - J(\pi(\nu_*^1), \mu(\nu_*^2)) &\leq 2\epsilon + J(\pi(\nu), \mu(\nu_n^2)) - J(\pi(\nu_n^1), \mu(\nu_n^2)) \\
        &\leq 2 \epsilon        
    \end{align*}
    This is true for all $\nu$, so also for their sup:
    \[
    0 \leq \sup_\nu J(\pi(\nu), \mu(\nu_*^2)) - J(\pi(\nu_*^1), \mu(\nu_*^2)) \leq 2\epsilon
    \]
    where the first inequality comes from the sup.
    7
    Finally, this is true for all $\epsilon > 0$. Taking $\epsilon$ to 0, we have that $J(\pi(\nu_*^1), \mu(\nu_*^2)) = \sup_\nu J(\pi(\nu), \mu(\nu_*^2))$, and thus $(\nu_*^1, \nu_*^2) \in \text{Graph}(\phi)$. Therefore $\text{Graph}(\phi)$ is closed.

    We have all the hypotheses required to apply Kakutani's fixed point theorem~\cite{kakutani1941generalization}: there thus exists $\nu^* \in \Delta(\Pi) \text{ such that } \nu^* \in \phi(\nu^*)$, ie. $\nu^* = \argmax_{\nu'} J(\pi(\nu'), \mu(\nu^*))$, which means that $\forall \nu' \in \Delta(\Pi), J(\pi(\nu'), \mu(\nu^*)) \leq J(\pi(\nu^*), \mu(\nu^*))$, in other words: $\nu^*$ is a Nash equilibrium of the game, and therefore, by Proposition~\ref{prop:nash_to_ce}, there exists a correlated equilibrium in the game.
\end{proof}

\section{Proofs on N-player Optimality}

\subsection{Useful lemmas}\label{proof:global_to_per_policy_flow_gap}

We define the following Lemma, which we will use in the rest of this section. Its role will be to link per-policy optimality to population optimality.

\begin{lemma}[Local to Global Flow Gap]\label{lemma:global_to_per_policy_flow_gap}
    If, for $t \in \mathcal{T}$, $\forall \pi \in \Pi$ such that $\nu(\pi) > 0$, $\mathbb{E}\left[ \| \mu_{N, \pi}(t) - \mu_\pi(\nu)(t) \|_2^2 \right] = \mathcal{O}\left( \frac{1}{N \nu_m} \right)$, then
    \[
        \mathbb{E}\left[ \| \mu_{N}(t) - \mu(\nu)(t) \|_2^2 \right]  = \mathcal{O}\left(\frac{1}{N \nu_m} \right).
    \]
    where $\nu_m = \min\limits_{\pi, \, \nu(\pi) > 0} \nu(\pi)$.
\end{lemma}
\begin{proof}[Proof of Lemma~\ref{lemma:global_to_per_policy_flow_gap}]
    We write $N_\pi$ the number of players playing policy $\pi$. As the result of $N$ independent samples from a Bernoulli random variable with law $\nu(\pi)$, this is a binomial random variable with parameters $\nu(\pi)$ and $N$.
    
    We develop the squared l2 distance expression:
    \begin{align*}
        \mathbb{E}&\left[ \| \mu_{N}(t) - \mu(\nu)(t) \|_2^2 \right] = \mathbb{E}\left[ \| \sum_\pi \frac{N_\pi}{N} \mu_{N, \pi}(t) - \nu(\pi) \mu_\pi(\nu)(t) \|_2^2 \right] \\
        &\leq \mathbb{E}\left[ \left(\sum_\pi \| \frac{N_\pi}{N} \mu_{N, \pi}(t) - \nu(\pi) \mu_\pi(\nu)(t) \|_2 \right)^2 \right] \\
        &\leq \sum_\pi \sum_{\pi'} \mathbb{E}\left[ \lvert\lvert \frac{N_\pi}{N} \mu_{N, \pi}(t) - \nu(\pi) \mu_\pi(\nu)(t) \rvert\rvert_2 \; \lvert\lvert \frac{N_{\pi'}}{N} \mu_{N, \pi'}(t) - \nu(\pi') \mu_{\pi'}(\nu)(t) \rvert\rvert_2 \right] \\
        &\leq \sum_\pi \sum_{\pi'} \mathbb{E}[ \left(\nu(\pi) \lvert\lvert \mu_{N, \pi}(t) - \mu_\pi(\nu)(t) \rvert\rvert_2 + \lvert \frac{N_\pi}{N} - \nu(\pi) \rvert \, \underbrace{\lvert\lvert \mu_{N, \pi}(t) \rvert\rvert_2}_{\leq 1} \right) \\ 
        &\left(\nu(\pi') \lvert\lvert \mu_{N, \pi'}(t) - \mu_{\pi'}(\nu)(t) \rvert\rvert_2 + \lvert \frac{N_{\pi'}}{N} - \nu(\pi') \rvert \, \underbrace{\lvert\lvert \mu_{N, \pi'}(t) \rvert\rvert_2}_{\leq 1} \right) ] \\
        &\leq \sum_\pi \sum_{\pi'} \nu(\pi)\nu(\pi') \mathbb{E}[\lvert\lvert \mu_{N, \pi}(t) - \mu_\pi(\nu)(t) \rvert\rvert_2 \, \lvert\lvert \mu_{N, \pi'}(t) - \mu_{\pi'}(\nu)(t) \rvert\rvert_2] + \\ 
        &\nu(\pi) \mathbb{E}[\lvert\lvert \mu_{N, \pi}(t) - \mu_\pi(\nu)(t) \rvert\rvert_2 \, \lvert \frac{N_{\pi'}}{N} - \nu(\pi') \rvert] + \\
        &\nu(\pi') \mathbb{E}[\lvert\lvert \mu_{N, \pi'}(t) - \mu_{\pi'}(\nu)(t) \rvert\rvert_2 \, \lvert \frac{N_{\pi}}{N} - \nu(\pi) \rvert] + \\
        &\mathbb{E}[\lvert \frac{N_{\pi}}{N} - \nu(\pi) \rvert \, \lvert \frac{N_{\pi'}}{N} - \nu(\pi') \rvert ) ].
    \end{align*}
    
    We use the Cauchy-Schwarz inequality to separate-out terms in the expectations: 
    \[
    \mathbb{E}[\lvert\lvert \mu_{N, \pi'}(t) - \mu_{\pi'}(\nu)(t) \rvert\rvert_2 \, \lvert \frac{N_{\pi'}}{N} - \nu(\pi') \rvert] \leq \sqrt{\mathbb{E}[\lvert\lvert \mu_{N, \pi'}(t) - \mu_{\pi'}(\nu)(t) \rvert\rvert_2^2] \, \mathbb{E}[\lvert \frac{N_{\pi'}}{N} - \nu(\pi') \rvert^2]}
    \] and similarly so for the other expressions; then bound each term. 
    
    By assumption, $\mathbb{E}[\lvert\lvert \mu_{N, \pi'}(t) - \mu_{\pi'}(\nu)(t) \rvert\rvert_2^2] = \mathcal{O}\left(\frac{1}{N \nu_m}\right)$.
    
    $N_\pi$ is the number of players who have sampled policy $\pi$. This is a binomial random variable with parameters $(\nu(\pi), N)$, and therefore $\mathbb{E}\left[\left(\nu(\pi)-\frac{N_\pi}{N}\right)^2 \right] = \frac{1}{N} \nu(\pi) (1-\nu(\pi))$.
    
    Finally, on the interval $[ \nu_m, 1]$, where 
    $\nu_m = \min\limits_{\pi, \, \nu(\pi) > 0} \nu(\pi)$,
    $\sqrt{\nu(\pi) (1-\nu(\pi))} \leq \nu(\pi) \sqrt{\frac{1-\nu_m}{\nu_m}} $. %
    
    Plugging these back in the former expressions, we obtain
    \begin{align*}
        \mathbb{E}\left[ \| \mu_{N}(t) - \mu(\nu)(t) \|_2^2 \right] &\leq \sum_\pi \sum_{\pi'} \nu(\pi)\nu(\pi')  \mathcal{O}\left(\frac{1}{N \nu_m}\right) + 2 \nu(\pi) \nu(\pi') \mathcal{O}\left(\frac{1}{N \nu_m}\right) \\
        + \nu(\pi) \nu(\pi') \mathcal{O}\left(\frac{1}{N \nu_m}\right) \\
        &= \mathcal{O}\left(\frac{1}{N \nu_m}\right)
    \end{align*}
    which concludes the proof.
\end{proof}

We will also implicitly use a lemma linking Lipschitzness in $p$ and $r$ to Lipschitzness in $J$.

\begin{lemma}[Lipschitzness of J]\label{lemma:lipschitz_r_p_to_j}
    Assume $r$ and $p$ are $\gamma_r$- and $\gamma_p$-Lipschitz in $\mu$, respectively.
    Then $J$ is $\left( \frac{\lvert\mathcal{X}\rvert}{\lvert \mathcal{X} \rvert - 1} \gamma_p R_M \sqrt{T - 2 \frac{1 - \lvert \mathcal{X} \rvert^{T}}{1 - \lvert \mathcal{X} \rvert} + \frac{1 - \lvert \mathcal{X} \rvert^{2T}}{1 - \lvert \mathcal{X} \rvert^2}} + \sqrt{T} \gamma_r \right)$-Lipschitz in $\mu$ where $R_M$ is the highest absolute reward obtainable in the game.
\end{lemma}
\begin{proof}
    Take $\mu_1$, $\mu_2 \in \mathcal{M}$. We start by proving that, given a policy $\pi$, its expected distribution $\mu^\pi_{\mu_1}$ under $\mu_1$ and $\mu^\pi_{\mu_2}$ under $\mu_2$ are such that $ \forall t,\, x,\; \mu^\pi_{\mu_1, t}(x) - \mu^\pi_{\mu_2, t}(x) \leq \gamma_p \frac{1 - \lvert \mathcal{X} \rvert^t}{1 - \lvert \mathcal{X} \rvert} \| \mu_1 - \mu_2 \|_2 $.
    
    We prove this by induction over game time.
    
    At $t = 0$, $\mu^\pi_{\mu_1, 0} = \mu^\pi_{\mu_2, 0} = \mu_0$, hence the relationship is verified.
    
    Assuming that $\mu^\pi_{\mu_1, t}(x) - \mu^\pi_{\mu_2, t}(x) \leq \gamma_p \frac{1 - \lvert \mathcal{X} \rvert^t}{1 - \lvert \mathcal{X} \rvert} \| \mu_{1, t} - \mu_{2, t} \|_2$, then take $x \in \mathcal{X}$.
    \[
        \mu^\pi_{\mu_1, t+1}(x) = \sum_{x_t} p^\pi(x \mid x_t, \mu_{1, t}) \mu^\pi_{\mu_1, t}(x_t),
    \]
    where $p^\pi( \cdot \mid x, \mu) = \sum_a \pi(x, a) p(\cdot \mid x, a, \mu)$.
    
    \begin{align*}
        \mu^\pi_{\mu_1, t+1}(x) &\leq \sum_{x_t} (p^\pi(x \mid x_t, \mu_{2, t}) + \gamma_p \| \mu_{1, t} - \mu_{2, t} \|_2) \mu^\pi_{\mu_1, t}(x_t)  
        \\
        \mu^\pi_{\mu_1, t+1}(x) &\leq \sum_{x_t} p^\pi(x \mid x_t, \mu_{2, t}) \mu^\pi_{\mu_{1}, t}(x_t) + \sum_{x_t} \gamma_p \| \mu_{1, t} - \mu_{2, t} \|_2 \mu^\pi_{\mu_1, t}(x_t) 
        \\
        \mu^\pi_{\mu_1, t+1}(x) &\leq \sum_{x_t} p^\pi(x \mid x_t, \mu_{2, t}) \mu^\pi_{\mu_1, t}(x_t) + \gamma_p \| \mu_{1, t} - \mu_{2, t} \|_2 \underbrace{\sum_{x_t} \mu^\pi_{\mu_1, t}(x_t)}_{= 1} 
        \\
        \mu^\pi_{\mu_1, t+1}(x) &\leq \sum_{x_t} \underbrace{p^\pi(x \mid x_t, \mu_{2, t})}_{\leq 1} (\mu^\pi_{\mu_{2}, t}(x_t) + \gamma_p \frac{1 - \lvert \mathcal{X} \rvert^t}{1 - \lvert \mathcal{X} \rvert} \|  \mu_{1, t} - \mu_{2, t} \|_2) + \gamma_p \| \mu_{1, t} - \mu_{2, t} \|_2 
        \\
        \mu^\pi_{\mu_1, t+1}(x) &\leq \mu^\pi_{\mu_{2}, t+1}(x) + \lvert \mathcal{X} \rvert \gamma_p \frac{1 - \lvert \mathcal{X} \rvert^t}{1 - \lvert \mathcal{X} \rvert} \|  \mu_{1, t} - \mu_{2, t} \|_2 + \gamma_p \|  \mu_{1, t} - \mu_{2, t} \|_2 
        \\
        \mu^\pi_{\mu_1, t+1}(x) - \mu^\pi_{\mu_{2}, t+1}(x) &\leq \gamma_p \frac{1 - \lvert \mathcal{X} \rvert^{t+1}}{1 - \lvert \mathcal{X} \rvert} \|  \mu_{1, t} - \mu_{2, t} \|_2. 
    \end{align*}

    The property is hereditary and initialized, hence it is true.
    
    We now turn to the proof of $J$ being Lipschitz. Take $\pi \in \Pi$, $\mu_1$, $\mu_2 \in \mathcal{M}$. Then we have
    
    \begin{align*}
        \lvert J(\pi, \mu_1) - J(\pi, \mu_2) \rvert &= \lvert \sum_x \sum_t \mu^\pi_{\mu_1, t}(x) r^\pi(x, \mu_{1, t}) - \mu^\pi_{\mu_{2, t}, t}(x) r^\pi(x, \mu_{2, t}) \rvert \\
        &= \lvert \sum_x \sum_t \left( \mu^\pi_{\mu_1, t}(x) - \mu^\pi_{\mu_2, t}(x) \right) r^\pi(x, \mu_{2, t}) + \mu^\pi_{\mu_1, t}(x) \left(r^\pi(x, \mu_1) - r^\pi(x, \mu_{2, t}) \right) \rvert \\
        &\leq \sum_x \sum_t \lvert \mu^\pi_{\mu_1, t}(x) - \mu^\pi_{\mu_2, t}(x) \rvert \lvert r^\pi(x, \mu_{2, t}) \rvert + \sum_x \sum_t \mu^\pi_{\mu_1, t}(x) \lvert r^\pi(x, \mu_{1, t}) - r^\pi(x, \mu_{2, t}) \rvert \\
        &\leq \sum_x \sum_t \gamma_p \frac{1 - \lvert \mathcal{X} \rvert^{t}}{1 - \lvert \mathcal{X} \rvert} \| \mu_{1, t} - \mu_{2, t} \|_2 \underbrace{\lvert r^\pi(x, \mu_{2, t}) \rvert}_{\leq R_M} + \sum_t \underbrace{\sum_x \mu^\pi_{\mu_1, t}(x)}_{=1} \lvert r^\pi(x, \mu_{1, t}) - r^\pi(x, \mu_{2, t}) \rvert \\
        &\leq \lvert \mathcal{X} \rvert \gamma_p R_M \sum_{t \geq 1}  \frac{1 - \lvert \mathcal{X} \rvert^{t}}{1 - \lvert \mathcal{X} \rvert} \| \mu_{1, t-1} - \mu_{2, t-1} \|_2 + \gamma_r \sum_t \| \mu_{1, t} - \mu_{2, t} \|_2 \rvert \\
        &\leq \lvert \mathcal{X} \rvert \gamma_p R_M \sum_t  \frac{1 - \lvert \mathcal{X} \rvert^{t}}{1 - \lvert \mathcal{X} \rvert} \| \mu_{1, t} - \mu_{2, t} \|_2 + \gamma_r \sqrt{\sum_t \| \mu_{1, t} - \mu_{2, t} \|_2^2}\sqrt{\sum_t 1}  \\
        &\leq \lvert\mathcal{X}\rvert \gamma_p R_M \sqrt{\sum_t \left(\frac{1 - \lvert \mathcal{X} \rvert^{t}}{1 - \lvert \mathcal{X} \rvert}\right)^2} \sqrt{\sum_t \| \mu_{1, t} - \mu_{2, t} \|_2^2} + \sqrt{T} \gamma_r \| \mu_1 - \mu_2 \|_2   \\
        &\leq \lvert\mathcal{X}\rvert \gamma_p R_M \sqrt{\sum_t \frac{1 - 2 \lvert \mathcal{X} \rvert^{t} + \lvert \mathcal{X} \rvert^{2t}}{\left(1 - \lvert \mathcal{X} \rvert\right)^2}} \| \mu_{1} - \mu_{2} \|_2 + \sqrt{T} \gamma_r \| \mu_1 - \mu_2 \|_2   \\
        &\leq \lvert\mathcal{X}\rvert \gamma_p R_M \sqrt{\frac{T - 2 \frac{1 - \lvert \mathcal{X} \rvert^{T}}{1 - \lvert \mathcal{X} \rvert} + \frac{1 - \lvert \mathcal{X} \rvert^{2T}}{1 - \lvert \mathcal{X} \rvert^2}}{\left(1 - \lvert \mathcal{X} \rvert\right)^2}} \| \mu_{1} - \mu_{2} \|_2 + \sqrt{T} \gamma_r \| \mu_1 - \mu_2 \|_2   \\
        &\leq \frac{\lvert\mathcal{X}\rvert}{\lvert \mathcal{X} \rvert - 1} \gamma_p R_M \sqrt{T - 2 \frac{1 - \lvert \mathcal{X} \rvert^{T}}{1 - \lvert \mathcal{X} \rvert} + \frac{1 - \lvert \mathcal{X} \rvert^{2T}}{1 - \lvert \mathcal{X} \rvert^2}} \| \mu_{1} - \mu_{2} \|_2 + \sqrt{T} \gamma_r \| \mu_1 - \mu_2 \|_2 \\
    \end{align*}
    which concludes the proof. Of course, the case $\lvert \mathcal{X} \rvert = 1$ is trivially solved: if there is only one state, then all distributions are equal.
\end{proof}

Note that if $\gamma_p = 0$, \emph{i.e.} the transition function doesn't depend on $\mu$, the above Lipschitz constant becomes the same as in the transition-independent case in Theorem~\ref{thm:optimality_no_dynamics_dependence}'s proof.

\subsection{Proof of Theorem~\ref{thm:full_optimality_theorem_discrete_rho}}\label{proof:full_optimality_theorem_discrete_rho}

We recall Theorem~\ref{thm:full_optimality_theorem_discrete_rho}:

\begin{unnum_theorem}
    Let $\rho$ be an $\epsilon \geq 0$-Mean-Field (coarse) correlated equilibrium. Then, if
    \begin{itemize}
        \item The reward and transition functions are lipschitz in $\mu$ for the $L_2$ norm, and
        \item $\rho$ is a finite sum of diracs,
    \end{itemize}
    then $\rho$ is an $\epsilon + O \left(\frac{1}{\sqrt{N}} \right)$ (coarse) correlated equilibrium of the corresponding $N$-player game.
\end{unnum_theorem}
\begin{proof}

    Let $u \in \mathcal{U}_{\{CE, CCE\}}$. An $\epsilon$-(coarse) correlated equilibrium $\rho$ in the Mean-Field's corresponding $N$-player game satisfies: 
    \[
    \mathbb{E}_{\nu \sim \rho, \pi \sim \nu}\left[\mathbb{E}_{\mu^N \sim \mu(\nu)} \left[J(u(\pi), \mu^N_{-\pi, u(\pi)}) - J(\pi, \mu^N)\right]\right] \leq \epsilon. 
    \]
    
    The outline of the proof is the following: We proceed first by bounding the difference between $\mu^N$ and $\mu(\nu)$, which we do by induction over timesteps and by separating players by the policy they sampled. Once this is done, we bound the difference between $\mu^N$ and $\mu^N_{u(\pi), -\pi}$, and finally use a lipschitz argument to relate $\mathbb{E} \left[ \lvert J(u(\pi), \mu(\nu)) - J(u(\pi), \mu^N) \rvert \right]$ to $\mathbb{E} \left[ \| \mu(\nu) - \mu^N \|_2 \right]$, which we have just bounded : indeed, Lemma~\ref{lemma:lipschitz_r_p_to_j} shows that if $p$ and $r$ are $\mu$-Lipschitz, then so is $J$.
    
    Indeed, assuming $\mathbb{E}_{\mu^N \sim \mu(\nu)} \left[ \| \mu^N - \mu(\nu)) \|_2 \right] = \mathcal{O} \left( \frac{1}{\sqrt{\alpha N}} \right)$ with any $\alpha > 0$, we have
    
    \begin{align*}
        \mathbb{E}_{\mu^N \sim \mu(\nu)} \left[J(\pi, \mu^N)\right] &= J(\pi, \mu(\nu)) + \mathbb{E}_{\mu^N \sim \mu(\nu)} \left[J(\pi, \mu^N) - J(\pi, \mu(\nu)) \right] \\ 
        &\leq J(\pi, \mu(\nu)) + \mathbb{E}_{\mu^N \sim \mu(\nu)} \left[ \lvert J(u(\pi), \mu^N_{-\pi, u(\pi)}) - J(\pi, \mu^N) \rvert \right] \\
        &\leq J(\pi, \mu(\nu)) + \gamma_r \mathbb{E}_{\mu^N \sim \mu(\nu)} \left[ \| \mu(\nu) - \mu^N \|_2 \right] \\
        &\leq J(\pi, \mu(\nu)) + \mathcal{O} \left( \frac{1}{\sqrt{\alpha N}} \right)  
    \end{align*}
        
    Once we have reached this point, we do the same operation and get the expected result for $J(u(\pi), \mu^N_{-\pi, u(\pi)})$.
    Unfortunately, the term $\alpha$ in the $\mathcal{O}$ depends on $\nu$, hence we will need to be careful when taking the expectation with respect to $\nu$. This yields to two different cases: In the case when $\rho$ is discrete - which is the case which typically interests us -, we keep the $\mathcal{O} \left( \frac{1}{\sqrt{N}} \right)$ bound; but in the case when $\rho$ is continuous, we are left with a less strong bound of $\mathcal{O} \left( \frac{1}{N^{\frac{1}{4}}} \right)$. 
    
    Let us first prove that, for all $\nu$, there exists some $\alpha > 0$ which we will show is equal to $\min_{\pi \mid \nu(\pi) > 0} \nu(\pi)$, such that $\mathbb{E}_{\mu^N \sim \mu(\nu)} \left[ \| \mu^N - \mu(\nu)) \|_2 \right] = \mathcal{O} \left( \frac{1}{\sqrt{\alpha N}} \right)$.
    We start by noting that $\mu_N(t) = \sum_\pi \frac{N_\pi}{N} \mu_{N, \pi}(t)$ and $\mu(\nu)(t) = \sum_\pi \nu(\pi) \mu_\pi(\nu)(t)$. %
    
    The above development shows that the proof requires us to bound $\mathbb{E} \left[ \| \mu_N - \mu(\nu) \|_2 \right]$. We proceed to do precisely this by induction on time, bounding each term $\mathbb{E} \left[ \| \mu_N(t) - \mu(\nu)(t) \|_2^2 \right]$ for any $\nu \in \Delta(\Pi)$.
    
    Indeed, 
    \begin{align*}
        \mathbb{E} \left[ \| \mu_N - \mu(\nu) \|_2 \right] &= \mathbb{E}\left[\sqrt{\sum_t \sum_x \left(\mu_N(t)(x) - \mu(\nu)(t)(x)\right)^2} \right] \\
        &\leq \sqrt{\sum_t \mathbb{E}\left[ \sum_x \left(\mu_N(t)(x) - \mu(\nu)(t)(x)\right)^2 \right]} \\
        &\leq \sqrt{\sum_t \mathbb{E}\left[ \| \mu_N(t) - \mu(\nu)(t) \|_2^2 \right]}
    \end{align*}
    
    Our induction hypothesis is, $\forall \nu \in \Delta(\Pi)$, $\forall t \in \mathcal{T}$, $\forall \pi \in \Pi$ such that $\nu(\pi)>0$,
    \[
        \mathbb{E} \left[ \| \mu_{N, \pi}(t) - \mu_\pi(\nu)(t) \|_2^2 \right] = \mathcal{O}\left( \frac{1}{N \nu_m} \right)
    \]
    and %
    \[
        \mathbb{E} \left[ \| \mu_{N}(t) - \mu(\nu)(t) \|_2^2 \right] = \mathcal{O}\left( \frac{1}{N \nu_m} \right).
    \]

    \paragraph{Induction initialization:}
    
    We initialize the induction with $t = 0$, and consider any $\pi$ such that $\nu(\pi) > 0$.
    
    \[
        \mathbb{E}\left[ \| \mu_{N, \pi}(0) - \mu_\pi(\nu)(0) \|_2^2 \right] 
        = \sum_{n=0}^N \mathbb{P}(N_\pi = n) \mathbb{E}\left[ \| \mu_{N, \pi}(0) - \mu_\pi(\nu)(0) \|_2^2 \, \big| N_\pi = n \right].
    \]
    
    When $N_\pi = 0$, then $\mu_{N, \pi} = 0$ everywhere, as there are no agents playing $\pi$. In this case, we have that 
    \[
    \mathbb{E}\left[ \| \mu_{N, \pi}(0) - \mu_\pi(\nu)(0) \|_2^2 \mid N_\pi = 0 \right] = \mathbb{E}\left[ \| \mu_\pi(\nu)(0) \|_2^2 \right] \leq 1.
    \]

    We have that $\forall x \in \mathcal{X}, \, \mu_{N, \pi}(0)(x)$ is the empirical mean of $N_\pi$ i.i.d. variables $\delta_{X_0 = x} \sim \mathcal{B}\left(\mu_0(x)\right)$, since at time $0$, all $N$ players are independently distributed according to $\mu(\nu)(0) = \mu_0$.
    
    Therefore $\mathbb{E}\left[ \left(\mu_{N, \pi}(0)(x) - \mu_\pi(\nu)(0)(x)\right)^2  \, \big| N_\pi \right] = \frac{1}{N_\pi} \mu(\nu)(0)(x) (1-\mu(\nu)(0)(x))$ and thus, since $\mu(\nu)(0)(x) (1-\mu(\nu)(0)(x)) \leq \frac{1}{2}$,

    \[
    \mathbb{E}\left[ \| \mu_N(0) - \mu(\nu)(0) \|_2^2 \, \big| N_\pi  \right] \leq \frac{1}{2 N_\pi} .
    \]
    
    Taking the expectation over $N_\pi$ yields, since $N_\pi$ is a binomial random variable with parameters $(\nu(\pi), N)$, 
    \begin{align}\label{eq:expectation_one_over_pi}
        \mathbb{E}&\left[ \| \mu_{N, \pi}(0) - \mu_\pi(\nu)(0) \|_2^2 \right] \nonumber\\
        &\leq \binom{N}{0} (1 - \nu(\pi))^N \mathbb{E}\left[ \| \mu_{N, \pi}(0) - \mu_\pi(\nu)(0) \|_2^2 \mid N_\pi = 0 \right] +  \sum_{n=1}^N \binom{N}{n} \nu(\pi)^{n} (1 - \nu(\pi))^{N-n} \frac{1}{2 n} \nonumber\\
        &\leq (1 - \nu(\pi))^N + \frac{1}{\nu(\pi) (N+1)} \underbrace{\sum_{n=1}^N \binom{N+1}{n+1} \nu(\pi)^{n+1} (1 - \nu(\pi))^{(N+1)-(n+1)}}_{\leq 1} \underbrace{\frac{n+1}{2 n}}_{\leq 1} \nonumber\\
        &\leq \left(1 - \nu(\pi)\right)^N + \frac{1}{\nu(\pi) (N+1)} \nonumber\\
        &= \mathcal{O}\left( \frac{1}{\nu_m N} \right).
    \end{align}
    
    Applying Lemma~\ref{lemma:global_to_per_policy_flow_gap} concludes the initialization step.

    \paragraph{Induction step:} %

    Let $t \geq 0$, $x \in \mathcal{X}$, and assume that $\mathbb{E}\left[ \| \mu_{N, \pi}(t) - \mu_\pi(\nu)(t) \|_2^2 \right] = \mathcal{O}\left( \frac{1}{\nu_m N} \right)$ for all $\pi \in \Pi$ such that $\nu(\pi) > 0$. We also write $p_x = \sum_{x_t} p^\pi(x \mid x_t, \mu_{N}(t)) \mu_{N, \pi}(t)(x_t)$ the expected state density at state $x$. %
    
    \begin{align}\label{eq:recurrence_source}
        \mathbb{E}&\left[ (\mu_{N, \pi}(t+1)(x) - \mu_\pi(\nu)(t+1)(x))^2 \right] \nonumber \\
        &= \mathbb{E}\left[ \left( (\mu_{N, \pi}(t+1)(x) - p_x) + (p_x - \mu_\pi(\nu)(t+1)(x)) \right)^2 \right] \nonumber \\
        &= \mathbb{E}\left[(\mu_{N, \pi}(t+1)(x) - p_x)^2\right] + 2 \mathbb{E}\left[(\mu_{N, \pi}(t+1)(x) - p_x) (p_x - \mu_\pi(\nu)(t+1)(x)) \right] \nonumber\\
        & \qquad\qquad + \mathbb{E}\left[(p_x - \mu_\pi(\nu)(t+1)(x))^2 \right]
    \end{align}

    We will bound each term in Equation~\ref{eq:recurrence_source} separately. We start with its first term.
    
    The evolution equation for the subpopulation playing $\pi$ is $$ \mathbb{E}\left[\mu_{N, \pi}(t+1)(x) \right] = \mathbb{E} \left[\sum_{x_t} p^\pi(x \mid x_t, \mu_{N}(t)) \mu_{N, \pi}(t)(x_t) \right]$$

    We note that for all $x \in \mathcal{X}$, we can write $\mu_{N, \pi}(t+1)(x) = \sum_{x_t} \mu_{N, \pi}(t+1)(x \lvert x_t) \mu_{N, \pi}(t)(x_t)$, where $\mu_{N, \pi}(t+1)(x \lvert x_t)$ is the proportion of particles at state $x_t$ at time $t$ which went to state $x$ at time $t+1$.
    
    We observe that $(N_\pi \mu_{N, \pi}(t)(x_t)) \mu_{N, \pi}(t+1)(x \lvert x_t)$ is the number of players playing $\pi$ present at state $x_t$ at time $t$ who moved to state $x$ at time $t+1$, which is a binomial random variable of parameters $p^\pi(x \mid x_t, \mu_{N}(t))$, the probability of moving to $x$ from $x_t$ when playing $\pi$, and $N_\pi \mu_{N, \pi}(t)(x_t)$, the number of players playing $\pi$ at $x_t$ at time $t$.
    
    Hence, if we write $\Delta(x, x_t) = \mu_{N, \pi}(t+1)(x \lvert x_t) - p^\pi(x \mid x_t, \mu_{N}(t))$ and recall that $p_x = \sum_{x_t} p^\pi(x \mid x_t, \mu_{N}(t)) \mu_{N, \pi}(t)(x_t)$
    
    \begin{align*}
        &\mathbb{E}\left[(\mu_{N, \pi}(t+1)(s) - p_x)^2 \mid N_\pi, p_x \right] 
        \\
        = &\mathbb{E}\left[(\sum_{x_t} ( \mu_{N, \pi}(t+1)(x \lvert x_t) - p^\pi(x \mid x_t, \mu_{N}(t)) ) \mu_{N, \pi}(t)(x_t))^2 \mid N_\pi, p_x \right]
        \\
        = &\mathbb{E}\left[\sum_{x_t} \sum_{x_t'} \Delta(x, x_t) \Delta(x, x_t') \mu_{N, \pi}(t)(x_t) \mu_{N, \pi}(t)(x_t') \mid N_\pi, p_x \right]
        \\
        \leq &\sum_{x_t} \sum_{x_t'} \sqrt{\mathbb{E}\left[ \Delta(x, x_t)^2 \mu^2_{N, \pi}(t)(x_t) \mid N_\pi, p_x \right]} \sqrt{\mathbb{E}\left[ \Delta(x, x_t')^2 \mu^2_{N, \pi}(t)(x_t') \mid N_\pi, p_x \right]} %
    \end{align*}

    By virtue of $\mu^2_{N, \pi}(t)(x_t)$ being $\mu_{N, \pi}(t)(x_t)$-measurable, we have $$\mathbb{E}\left[ \Delta(x, x_t)^2 \mu^2_{N, \pi}(t)(x_t) \mid N_\pi, p_x \right] = \mathbb{E}\left[ \mathbb{E}\left[ \Delta(x, x_t)^2 \mid \mu_{N, \pi}(t)(x_t) \right] \mu^2_{N, \pi}(t)(x_t) \mid N_\pi, p_x \right]$$
    
    Given that $(N_\pi \mu_{N, \pi}(t)(x_t)) \mu_{N, \pi}(t+1)(x \lvert x_t)$ is a binomial random variable with parameters $p^\pi(x \mid x_t, \mu_{N}(t))$ and $N_\pi \mu_{N, \pi}(t)(x_t)$,
    
    \begin{align*}
        \mathbb{E}\left[ \Delta(x, x_t)^2 \mid \mu_{N, \pi}(t)(x_t), N_\pi, p_x \right] &= \frac{1}{N^2_\pi \mu^2_{N, \pi}(t)(x_t)} \mathbb{E}\left[ N_\pi^2 \mu^2_{N, \pi}(t)(x_t) \Delta(x, x_t)^2 \mid \mu_{N, \pi}(t)(x_t), N_\pi, p_x \right] \\
        &= \frac{1}{N^2_\pi \mu^2_{N, \pi}(t)(x_t)} p^\pi(x \mid x_t, \mu_{N}(t)) (1 - p^\pi(x \mid x_t, \mu_{N}(t))) N_\pi \mu_{N, \pi}(t)(x_t) \\
        &\leq \frac{1}{N_\pi \mu_{N, \pi}(t)(x_t)}
    \end{align*}
    
    Thus 
    \begin{align*}
        \mathbb{E}\left[ \Delta(x, x_t)^2 \mu^2_{N, \pi}(t)(x_t) \mid N_\pi, p_x \right] &\leq \mathbb{E}\left[ \frac{1}{N_\pi} \mu_{N, \pi}(t)(x_t) \mid N_\pi, p_x \right] \\
        &\leq \frac{1}{N_\pi}.
    \end{align*}
    
    Plugging this back into the former equation, this yields
    
    \begin{align*}
        \mathbb{E}\left[(\mu_{N, \pi}(t+1)(x) - p_x)^2 \mid N_\pi, p_x \right] & \leq \sum_{x_t} \sum_{x_t'} \frac{1}{N_\pi} \\
        & \leq \frac{1}{N_\pi} \lvert \mathcal{X} \rvert^2.
    \end{align*}
    Taking the expectation over $N_\pi$ and following the same steps as Equation~\ref{eq:expectation_one_over_pi}, we have $\mathbb{E}\left[(\mu_{N, \pi}(t+1)(x) - p_x)^2 \right] \leq  \left( \left(1 - \nu(\pi)\right)^N + \frac{2}{\nu(\pi) (N+1)} \right) \lvert \mathcal{X} \rvert^2 = \mathcal{O}\left( \frac{1}{\nu_m N} \right)$.

    \textbf{Bounding the second and third terms in Equation~\ref{eq:recurrence_source}:}

    The middle-term is simplified using the Cauchy-Schwarz inequality: 
    
    \[
    \mathbb{E}\left[(\mu_{N, \pi}(t+1)(x) - p_x) (p_x - \mu_\pi(\nu)(t+1)(x)) \right] \leq \sqrt{\mathbb{E}\left[(\mu_{N, \pi}(t+1)(x) - p_x)^2 \right] \mathbb{E}\left[(p_x - \mu_\pi(\nu)(t+1)(x))^2 \right]}
    \]
    
    We have an upper bound for the first term, let us now bound the second term.
    \begin{align*}
        &\mathbb{E}[ (p_x - \mu_\pi(\nu)(t+1)(x))^2 ] \\
        &= \mathbb{E}\left[\left( \sum_{x_t} \sum_a \pi(x_t, a) \left( p(x \mid x_t, a, \mu_{N}(t)) \mu_{N, \pi}(t)(x_t) - p(x \mid x_t, a, \mu(\nu)(t)) \mu_{\pi}(\nu)(t)(x_t) \right) \right)^2 \right] \\
        & \leq \mathbb{E}[( \underbrace{\sum_{x_t} p^{\pi}(x \mid x_t, \mu(\nu)(t))\left(\mu_{N, \pi}(t)(x_t) - \mu_{\pi}(\nu)(t)(x_t) \right)}_{\leq \sqrt{\sum_{x_t}  p^{\pi}(x \mid x_t, \mu(\nu)(t))^2} \sqrt{\sum_{x_t} \left( \mu_{N, \pi}(t)(x_t) - \mu_{\pi}(\nu)(t)(x_t) \right)^2}} + \\
        & \;\;\;\;\;\;\;\;\;\;\;\;\;\underbrace{\sum_{x_t} \mu_{N, \pi}(t)(x_t)}_{= 1} \underbrace{\sum_a \pi(x_t, a)}_{= 1} \gamma_p \| \mu_{N}(t) - \mu(\nu)(t) \|_2 )^2 ] \\
        & \leq \mathbb{E}[( \underbrace{\sqrt{\sum_{x_t}  p^{\pi}(x \mid x_t, \mu(\nu)(t))^2}}_{\leq \sqrt{\mid\mathcal{X}\mid}} \| \mu_{N, \pi}(t) - \mu_{\pi}(\nu)(t) \|_2 + \gamma_p \| \mu_{N}(t) - \mu(\nu)(t) \|_2 )^2 ] \\
        & \leq \mathbb{E}\left[\left(\sqrt{\lvert\mathcal{X}\rvert} \| \mu_{N, \pi}(t) - \mu_{\pi}(\nu)(t) \|_2 + \gamma_p \| \mu_{N}(t) - \mu(\nu)(t) \|_2 \right)^2 \right] \\
        & \leq \lvert\mathcal{X}\rvert \underbrace{\mathbb{E}\left[\lvert\lvert \mu_{N, \pi}(t) - \mu_{\pi}(\nu)(t) \rvert\rvert_2^2 \right]}_{= \mathcal{O}\left( \frac{1}{\nu_m N} \right)} + 2 \gamma_p \sqrt{\lvert\mathcal{X}\rvert} \mathbb{E}\left[\lvert\lvert \mu_{N, \pi}(t) - \mu_{\pi}(\nu)(t) \rvert\rvert_2 \, \lvert\lvert \mu_{N}(t) - \mu(\nu)(t) \rvert\rvert_2 \right] + \gamma_p \underbrace{\mathbb{E}\left[\lvert\lvert \mu_{N}(t) - \mu(\nu)(t) \rvert\rvert_2^2 \right]}_{= \mathcal{O}\left( \frac{1}{\nu_m N} \right)}
    \end{align*}
    where the third line inequality comes from $p$ being $\gamma_p$-lipschitz in $\mu$, and the last equalities in underbraces come from the induction assumption and Lemma~\ref{lemma:global_to_per_policy_flow_gap}.
    We apply the Cauchy-Schwarz inequality to the middle term: 
    
    \begin{align*}
        \mathbb{E}[\lvert\lvert &\mu_{N, \pi}(t) - \mu_{\pi}(\nu)(t) \rvert\rvert_2 \, \lvert\lvert \mu_{N}(t) - \mu(\nu)(t) \rvert\rvert_2 ] \\
        &\leq \sqrt{\mathbb{E}\left[\lvert\lvert \mu_{N, \pi}(t) - \mu_{\pi}(\nu)(t) \rvert\rvert_2^2 \right] \mathbb{E}\left[\lvert\lvert \mu_{N}(t) - \mu(\nu)(t) \rvert\rvert_2^2 \right]}  \\
        &= \mathcal{O}\left( \frac{1}{\nu_m N} \right)
    \end{align*}
    
    Therefore, for all $t \geq 0$, $\pi \in \Pi, \, \nu(\pi) > 0$, $\mathbb{E}\left[ (\mu_{N, \pi}(t+1)(x) - \mu_\pi(\nu)(t+1)(x))^2 \right] = \mathcal{O}\left( \frac{1}{\nu_m N} \right)$, and thus $\mathbb{E} \left[ \| \mu_N - \mu(\nu) \|_2 \right] = \mathcal{O}\left( \frac{1}{\sqrt{\nu_m N}} \right)$.

    \paragraph{Neglecting the deviation term:}
    We now consider the case when one player deviates from policy $\pi$ to policy $u(\pi)$. The effect of this defection is an impurity of the policy distribution with, as a result, an increase of $N_{u(\pi)}$ and a decrease of $N_\pi$ by 1 each. We briefly describe how this change can be neglected.
    
    \textbf{If the deviated-to policy is in the support of $\nu$}: We see that the result of Lemma~\ref{lemma:global_to_per_policy_flow_gap} remains unchanged, as the additional $\frac{1}{N_{u(\pi)}}$ / $\frac{-1}{N_{\pi}}$ can be separated using the triangle inequality, and is $\mathcal{O}\left(\frac{1}{\nu_m N} \right)$. 
    
    Both the initialization and the inheritance parts of the recurrence involve the quantity $N_\pi$, but the only influence of this impurity is in the expectation's conditioning (or in the summation indices). We see that this change replaces the $\frac{1}{N_\pi}$ term by $\frac{1}{N_\pi \pm 1}$, and therefore ultimately only changes the bounds by a constant amount. If this leads to a policy which is not played anymore (\textit{i.e.} $N_\pi = 1$ before the deviation), then we can use the previously-developed argument regarding the $N_\pi = 0$ case, noting that the probability of $N_\pi = 1$ is $N \nu(\pi) (1 - \nu(\pi))^{N-1}  =  O \left ( \frac{1}{\nu_m N} \right )$.
    
    Thus we also have that 
    \[
    \mathbb{E} \left[ \| \mu^N_{-\pi, u(\pi)} - \mu(\nu) \|_2 \right] = \mathcal{O}\left( \frac{1}{\sqrt{N}} \right) 
    \]
    
    \textbf{If the deviated-to policy is \emph{not} in the support of $\nu$}: Then it creates a single new term in the local-to-global development (We note $N_\pi'$ the "updated" number of players playing $\pi$: either it is equal to $N_\pi$ for the non-deviating policies, or it is equal to $N_\pi-1$ for the policy the deviating player played; and $u(\pi)$ the deviated-to policy):
    
    \begin{align*}
        \mathbb{E}&\left[ \| \mu_{N}(t) - \mu(\nu)(t) \|_2^2 \right] = \mathbb{E}\left[ \| \frac{1}{N} \mu_{N, u(\pi)}(t) + \sum_\pi \frac{N_\pi'}{N} \mu_{N, \pi}(t) - \nu(\pi) \mu_\pi(\nu)(t) \|_2^2 \right] \\
        &\leq \mathbb{E}\left[ \left(\frac{1}{N} \underbrace{\| \mu_{N, u(\pi)}(t) \|_2}_{\leq 1} + \sum_\pi \| \frac{N_\pi'}{N} \mu_{N, \pi}(t) - \nu(\pi) \mu_\pi(\nu)(t) \|_2 \right)^2 \right] \\
    \end{align*}
    
    We see that this new impurity adds a $\frac{1}{N}$ term within the sum, which doesn't alter the end result regarding the closeness of $\mu^N$ to $\mu(\nu)$. 
    
    \paragraph{Integrating over $\Delta(\Pi)$:}
    The bound derived above depends on $\nu$, yet to compute expected deviation payoffs, we must integrate over $\Delta(\Pi)$ following $\rho$'s distribution. In the current case, $\rho$ is a sum of finitely many diracs.
    
    Then 
    \[
    \mathbb{E}_{\nu \sim \rho}\left[ \frac{1}{\sqrt{\nu_m N}} \right] = \sum_{\nu \mid \rho(\nu) > 0} \frac{\rho(\nu)}{\sqrt{\nu_m N}}
    \]
    \emph{i.e.}
    \[
    \mathbb{E}_{\nu \sim \rho}\left[ \frac{1}{\sqrt{\nu_m N}} \right] = \frac{1}{\sqrt{N}} \sum_{\nu \mid \rho(\nu) > 0} \frac{\rho(\nu)}{\sqrt{\nu_m}}
    \]
    and we keep the $\frac{1}{\sqrt{N}}$ bound, with an added term representing the non-optimality of each $\nu$ in the discrete support of $\rho$ weighted by $\rho$.
\end{proof}

\subsection{Proof of Theorem~\ref{thm:full_optimality_theorem_continuous_rho}}\label{proof:full_optimality_theorem_continuous_rho}

We recall Theorem~\ref{thm:full_optimality_theorem_continuous_rho}:

\begin{unnum_theorem}
    Let $\rho$ be an $\epsilon \geq 0$-Mean-Field (coarse) correlated equilibrium. Then, if
    \begin{itemize}
        \item The reward and transition functions are lipschitz in $\mu$ for the $L_2$ norm, and
        \item $\rho$ is not a finite sum of diracs,
    \end{itemize}
    then $\rho$ is an $\epsilon + O \left(\frac{1}{N^{\frac{1}{4}}} \right)$ (coarse) correlated equilibrium of the corresponding $N$-player game.
\end{unnum_theorem}
\begin{proof}
    The proof starts like proof~\ref{proof:full_optimality_theorem_discrete_rho}, and starts diverging at the last step: integration over $\rho$.

    In the case when $\rho$ is not a finite sum of diracs, given the bound $\frac{1}{\nu_m N}$, it could be that $\rho$ assigns mass on a sequence of $\nu$ for which $\nu(\pi) \rightarrow 0$. We however note that, given a threshold $\alpha \in ]0,\, 1[$, we can separate, using the triangular inequality, policies whose selection probability according to $\nu$ is lower than $\alpha$ from those for which it is higher than $\alpha$:
    
    \begin{align*}
        \mathbb{E} \left[ \| \mu^N - \mu(\nu) \|_2 \right] &= \mathbb{E} \left[ \| \sum_\pi \frac{N_\pi}{N} \mu^N_\pi - \nu(\pi) \mu_\pi(\nu) \|_2 \right] \\
        &\leq \mathbb{E} \left[ \| \sum_{\pi \mid \nu(\pi) > \alpha} \frac{N_\pi}{N} \mu^N_\pi - \nu(\pi) \mu_\pi(\nu) \|_2 \right] + \sum_{\pi \mid \nu(\pi) \leq \alpha} \mathbb{E}\left[  \| \frac{N_{\pi}}{N} \mu^N_{\pi} - \nu(\pi) \mu_{\pi}(\nu) \|_2 \right]
    \end{align*}

    We examine the second term for a given policy $\pi$, which contains only policies whose selection probability is lower than $\alpha$. 
    \begin{align*}
        \mathbb{E} \left[ \| \frac{N_{\pi}}{N} \mu^N_{\pi} - \nu(\pi) \mu_{\pi}(\nu) \|_2 \right] &\leq \mathbb{E} \left[ \left\lvert \frac{N_{\pi}}{N} - \nu(\pi) \right\rvert \underbrace{\| \mu^N_{\pi} \|_2}_{\leq \sqrt{T}} + \nu(\pi) \underbrace{\|  \mu^N_{\pi} - \mu_{\pi}(\nu) \|_2}_{\leq \sqrt{2T}} \right] \\
        &\leq \mathbb{E} \left[ \sqrt{T} \left\lvert \frac{N_{\pi}}{N} - \nu(\pi) \right\rvert + \sqrt{2T} \nu(\pi) \right] \\
        &\leq \frac{\sqrt{T}}{N} \sqrt{\mathbb{E} \left[ \left( N_{\pi} - N \nu(\pi) \right)^2 \right]}  + \sqrt{2T} \nu(\pi) \\
        &\leq \frac{\sqrt{T}}{N} \sqrt{N \nu(\pi) (1-\nu(\pi))} + \sqrt{2T} \nu(\pi) \\
        &\leq \sqrt{\frac{\alpha T}{N}} + \sqrt{2T} \alpha
    \end{align*}
    
    We have of course that for each $t \in \mathcal{T}$, using similar steps, and assuming $\alpha \leq \frac{K}{\sqrt{N}}$ with $K \in \mathbb{R_+^*}$ a given constant,
    \[
    \mathbb{E} \left[ \| \frac{N_{\pi}}{N} \mu^N_{\pi}(t) - \nu(\pi) \mu_{\pi}(\nu)(t) \|_2 \right] \leq \mathcal{O}\left( \frac{1}{\sqrt{N}} \right)
    \]
        
    We now make the point that the whole demonstration can be done while only considering policies whose play probabilities is $ > \frac{1}{\sqrt{N}} $, up to a $\mathcal{O} \left(\frac{\lvert \Pi \rvert}{\sqrt{N}} \right)$ term. 
    
    Indeed, at each step of the proof, for each time $t \in \mathcal{T}$, for policies whose play probability is $ > \frac{K}{\sqrt{N}}$, the only term which involves other policies is $\mathbb{E} \left[ \| \mu^N(t) - \mu(\nu)(t) \|_2 \right]$. 
    
    This term can be separated in two using the triangular inequality, between policies whose play probability is lower than $\frac{K}{\sqrt{N}}$, and policies whose play probability isn't. The first group adds at most a $\mathcal{O} \left(\frac{\lvert \Pi \rvert}{\sqrt{N}} \right)$ term. The second term, a $\mathcal{O} \left(\frac{1}{\sqrt{N}} \right)$ with partial dependency on some $\nu$ for which all interesting components are $ > \frac{1}{\sqrt{N}}$.
        
    More specifically, writing $\nu_{m, \alpha} = \min_{\pi \mid \nu(\pi) > \alpha} \nu(\pi)$ and noting that $\frac{1}{\sqrt{\nu_{m, \alpha}}} < \frac{1}{\sqrt{\alpha}}$, 
    \begin{align*}
        \mathbb{E}_{\nu \sim \rho}\left[\mathbb{E}_{\mu^N \sim \mu(\nu)} \left[ \| \mu^N - \mu(\nu) \|_2 \right] \right] &= \mathbb{E}_{\nu \sim \rho}[\underbrace{\mathbb{E} \left[ \| \sum_{\pi \mid \nu(\pi) > \alpha} \frac{N_\pi}{N} \mu^N_\pi - \nu(\pi) \mu_\pi(\nu) \|_2 \right]}_{\leq \frac{1}{\sqrt{N \nu_{m, \alpha}}} } + \underbrace{\sum_{\pi \mid \nu(\pi) \leq \alpha} \mathbb{E} \left[  \| \frac{N_{\pi}}{N} \mu^N_{\pi} - \nu(\pi) \mu_{\pi}(\nu) \|_2 \right]}_{\leq \lvert \Pi \rvert (\sqrt{\frac{\alpha T}{N}} + \sqrt{2T} \alpha)} ] \\
        &\leq \frac{1}{\sqrt{N \nu_{m, \alpha}}} + \lvert \Pi \rvert (\sqrt{\frac{\alpha T}{N}} + \sqrt{2T} \alpha) \\
        &\leq \frac{1}{\sqrt{N \alpha}} + \lvert \Pi \rvert ( \sqrt{\frac{\alpha T}{N}} + \sqrt{2T} \alpha)
    \end{align*}
    
    We look for the optimal value of $\alpha$ while remembering that we must have $\alpha \leq \frac{K}{\sqrt{N}}$ with $K$ independent of $N$ for the above developments to remain true.
    
    It is not straightforward to find the minimum value of $\frac{1}{\sqrt{N \alpha}} + \lvert \Pi \rvert \left( \sqrt{\frac{\alpha T}{N}} + \sqrt{2T} \alpha \right)$ when varying $\alpha$. However, we are only interested in $\mathcal{O}$ relationships. Now, looking at the first term $\frac{1}{\sqrt{N \alpha}}$, we notice that we actually want $\alpha$ to be as large as possible for this term to be as small as possible. Taking $\alpha$ to be the largest "allowed" value, $\frac{K}{\sqrt{N}}$, transforms this term into a $\mathcal{O} \left( \frac{1}{N^{\frac{1}{4}}} \right)$, while the other two terms are $\mathcal{O} \left( \frac{1}{N^{\frac{3}{4}}} \right)$ and $\mathcal{O} \left( \frac{1}{\sqrt{N}} \right)$ respectively. We therefore see that the expectation is $\mathcal{O}\left( \frac{1}{N^{\frac{1}{4}}} \right)$.
        
    Going back to the initial developments of the proof, we have that, if $\rho$ is discrete 
    \[
    \mathbb{E}_{\nu \sim \rho, \pi \sim \nu}\left[\mathbb{E}_{\mu^N \sim \mu(\nu)} \left[J(u(\pi), \mu^N_{-\pi, u(\pi)}) - J(\pi, \mu^N)\right]\right] \leq \mathbb{E}_{\nu \sim \rho, \pi \sim \nu}\left[J(u(\pi), \mu(\nu)) - J(\pi, \mu(\nu))\right] + \mathcal{O}\left( \frac{1}{\sqrt{N}} \right) 
    \]
    \emph{i.e.}
    \[
    \mathbb{E}_{\nu \sim \rho, \pi \sim \nu}\left[\mathbb{E}_{\mu^N \sim \mu(\nu)} \left[J(u(\pi), \mu^N_{-\pi, u(\pi)}) - J(\pi, \mu^N)\right]\right] \leq \epsilon + \mathcal{O}\left( \frac{1}{\sqrt{N}} \right) 
    \]
    
    And, if $\rho$ is continuous,
    \[
    \mathbb{E}_{\nu \sim \rho, \pi \sim \nu}\left[\mathbb{E}_{\mu^N \sim \mu(\nu)} \left[J(u(\pi), \mu^N_{-\pi, u(\pi)}) - J(\pi, \mu^N)\right]\right] \leq \mathbb{E}_{\nu \sim \rho, \pi \sim \nu}\left[J(u(\pi), \mu(\nu)) - J(\pi, \mu(\nu))\right] + \mathcal{O}\left( \frac{1}{N^{\frac{1}{4}}} \right) 
    \]
    \emph{i.e.}
    \[
    \mathbb{E}_{\nu \sim \rho, \pi \sim \nu}\left[\mathbb{E}_{\mu^N \sim \mu(\nu)} \left[J(u(\pi), \mu^N_{-\pi, u(\pi)}) - J(\pi, \mu^N)\right]\right] \leq \epsilon + \mathcal{O}\left( \frac{1}{N^{\frac{1}{4}}} \right) 
    \]

    Which concludes the proof.
\end{proof}

\section{Online Mirror Descent No-regret Proof}\label{appendix:omd_no_regret_proof}

We begin by recalling Theorem~\ref{prop:omd_no_regret}:

\begin{unnum_theorem}
 Online Mirror Descent is a regret minimizing strategy in Mean Field games (no monotonicity required):
\[
\frac{1}{\tau} \extregret((\pi(s))_{0 \leq s \leq \tau}; (\mu^{\pi(s)})_{0 \leq s \leq \tau}) = O(\frac{1}{\tau})
\]
\end{unnum_theorem}

\begin{proof}

We introduce $t \in \mathcal{T}$ the game time. In the following arguments, we draw the reader's attention towards the distinction between game time $t$ and learning time $\tau$.

We define, for all $\pi \in \bar\Pi$, and for $y,\, Q$ and $\pi(\tau)$ the quantities defined above,
\[
\mathcal{L}(\pi, y(., ., \tau)) = \sum \limits_{x \in \mathcal{X}} \sum \limits_{t \in \mathcal{T}} \mu^{\pi}_t(x)[h^*(y_t(x,.,\tau)) - h^*(y_{\pi, t}(x,.)) - \langle \pi_t, y_t(x,.,\tau)-y_{\pi, t}(x,.) \rangle]
\] where $y_\pi$ is such that $\pi(.\mid x) = \Gamma(y_\pi(x,.))$.

We can deduce that $\frac{d}{d\tau} \mathcal{L}(\pi, y(., ., \tau)) = V_0^{\pi(\tau), \mu^{\pi(\tau)}} - V_0^{\pi, \mu^{\pi(\tau)}}$.

Indeed:
\begin{align}
    &\frac{d}{d\tau} \mathcal{L}(\pi, y(., ., \tau)) \nonumber\\
    &= \frac{d}{d\tau} \sum \limits_{x \in \mathcal{X}} \sum\limits_{t \in \mathcal{T}} \mu^{\pi}(x)_t[h^*(y_t(x,.,\tau)) - h^*(y_{\pi, t}(x,.)) - <\pi_t, y_t(x,.,\tau)-y_{\pi, t}(x,.)>]\nonumber\\
    &= \sum \limits_{x \in \mathcal{X}} \sum\limits_{t \in \mathcal{T}} \mu^{\pi}_t(x) \frac{d}{d\tau} [h^*(y_t(x,.,\tau)) - h^*(y_{\pi, t}(x,.)) - <\pi_t,y_t(x,.,\tau)-y_{\pi, t}(x,.)>]\nonumber\\
    &= \sum \limits_{x \in \mathcal{X}} \sum\limits_{t \in \mathcal{T}} \mu^{\pi}_t(x) [\frac{d}{d\tau} h^*(y_t(x,.,\tau)) - <\pi_t,\frac{d}{d\tau} y_t(x,.,\tau)>]\nonumber\\
    &= \sum \limits_{x \in \mathcal{X}} \sum\limits_{t \in \mathcal{T}} \mu^{\pi}_t(x) [<\frac{d}{d\tau} y_t(x, ., \tau), \nabla h^*(y_t(x,.,\tau))> - <\pi_t, \frac{d}{d\tau} y_t(x,.,\tau)>]\nonumber\\
    &= \sum \limits_{x \in \mathcal{X}} \sum\limits_{t \in \mathcal{T}} \mu^{\pi}_t(x) [\underbrace{<Q_t^{\pi(\tau), \mu_\tau}(x, .), \pi(x, ., \tau)>}_{V^{\pi(\tau), \mu_\tau}_t(x)} - <\pi_t,Q_t^{\pi(\tau), \mu_\tau}(x, .)>]\nonumber\\
    &<\pi_t,Q_t^{\pi(\tau), \mu_\tau}(x, .)>\nonumber\\
    &=\sum_a \pi_t(x, a) [r(x, a, \mu_\tau) + \sum_{x' \in \mathcal{X}} p(x' \mid  x, a, \mu_t) V_{t+1}^{\pi(\tau), \mu_\tau}(x')] \nonumber \\
    &= \underbrace{\sum_a \pi_t(x, a) [r(x, a, \mu_\tau) + \sum_{x' \in \mathcal{X}} p(x' \mid  x, a, \mu_\tau) V_{t+1}^{\pi, \mu_\tau}(x')]}_{= V_{t}^{\pi, \mu_\tau}(x)} + \sum_a \pi_t(x, a) \sum_{x' \in \mathcal{X}} p(x' \mid  x, a, \mu_t) [ V_{t+1}^{\pi(\tau), \mu_\tau}(x') - V_{t+1}^{\pi, \mu_\tau}(x')] \nonumber \\
    &= V_{t}^{\pi, \mu_\tau}(x) + \sum_a \pi_t(x, a) \sum_{x' \in \mathcal{X}} p(x' \mid  x, a, \mu_t) [ V_{t+1}^{\pi(\tau), \mu_\tau}(x') - V_{t+1}^{\pi, \mu_\tau}(x')] \nonumber \\
    &\frac{d}{d\tau} \mathcal{L}(\pi, y(., ., \tau)) \nonumber\\
    &= \sum \limits_{x \in \mathcal{X}} \sum\limits_{t \in \mathcal{T}} \mu^{\pi}_t(x) [V_t^{\pi(\tau), \mu_\tau}(x) - V_t^{\pi, \mu_\tau}(x)] - \sum\limits_{t \in \mathcal{T}} \underbrace{\sum \limits_{x \in \mathcal{X}} \mu^{\pi}_t(x) \pi_t(x, a) \sum_{x' \in \mathcal{X}} p(x' \mid  x, a, \mu_t)}_{= \sum\limits_{x' \in \mathcal{X}} \mu^{\pi}_{t+1}(x')} [V_{t+1}^{\pi(\tau), \mu_\tau}(x') - V_{t+1}^{\pi, \mu_\tau}(x')] \nonumber\\
    &= \sum \limits_{x \in \mathcal{X}} \sum\limits_{t \in \mathcal{T}} \mu^{\pi}_t(x) [V_t^{\pi(\tau), \mu_\tau}(x) - V_t^{\pi, \mu_\tau}(x)] - \sum \limits_{x \in \mathcal{X}} \sum\limits_{t \in \mathcal{T}}  \mu^{\pi}_{t+1}(x) [V_{t+1}^{\pi(\tau), \mu_\tau}(x) - V_{t+1}^{\pi, \mu_\tau}(x)] \nonumber\\
    &= \sum \limits_{x \in \mathcal{X}} \mu^{\pi}_0(x) [V_0^{\pi(\tau), \mu_\tau}(x) - V_0^{\pi, \mu_\tau}(x)] \nonumber \\
    &= V_0^{\pi(\tau), \mu_\tau} - V_0^{\pi, \mu_\tau}
\end{align}

The proof is concluded by saying:
\begin{align}
    \extregret((\pi(\tau))_{0 \leq \tau \leq \tau_0}; (\mu^{\pi(\tau)})_{0 \leq \tau \leq \tau_0}) &= \max \limits_{\pi} \int\limits_{0}^{\tau_0} V_0^{\pi, \mu_\tau} - V_0^{\pi(\tau), \mu_\tau} d\tau \nonumber \\
    &= \max \limits_{\pi} \int \limits_0^{\tau_0} -\frac{d}{d\tau} \mathcal{L}(y(., ., \tau)) d\tau \nonumber\\
    &=\max_\pi [\mathcal{L}(\pi, y(., ., 0)) - \mathcal{L}(\pi, y(., ., \tau_0))]
\end{align}

and 

\begin{align}
    &\mathcal{L}(\pi, y(., ., 0)) - \mathcal{L}(\pi, y(., ., \tau_0))\nonumber\\
    &= \sum \limits_{x \in \mathcal{X}} \sum_{t \in \mathcal{T}} \mu_t^{\pi}(x)[h^*(y_t(x,.,0)) - \langle \pi_t, y_t(x,.,0) \rangle \underbrace{- h^*(y_t(x,.,\tau_0)) + \langle \pi_t, y_t(x,.,\tau_0) \rangle}_{\leq h(\pi_t)}]\nonumber \\
    &\leq \sum \limits_{x \in \mathcal{X}} \sum_{t \in \mathcal{T}} \mu_t^{\pi}(x)\left[\underbrace{h^*(y_t(x,.,0)) - \langle \pi_t(0)(x, .), y_t(x,.,0) \rangle}_{=-h(\pi_t(0)(x, .)} + h(\pi_t) \underbrace{- \langle \pi_t - \pi_t(0)(x, .), y_t(x,.,0) \rangle}_{\leq \| y(.,.,0)\| _{+\infty}}\right] \nonumber\\
    &\leq \sum \limits_{x \in \mathcal{X}} \sum_{t \in \mathcal{T}} \mu_t^{\pi}(x)\left[h(\pi_t) -h(\pi_t(0)(x, .)) + \| y(.,.,0)\|_{+\infty}\right] \nonumber\\
    &\leq \underbrace{\left(\sum \limits_{x \in \mathcal{X}} \sum_{t \in \mathcal{T}} \mu_t^{\pi}(x)\right)}_{T}[h_{\max} - h_{\inf} + \| y(.,.,0)\|_{+\infty}] \nonumber
\end{align}

In the end by combining those two results we have
\[
\frac{1}{\tau_0} \extregret((\pi(\tau))_{0 \leq \tau \leq \tau_0}; (\mu^{\pi(\tau)})_{0 \leq \tau \leq \tau_0}) \leq \frac{T}{\tau_0} \left ( h_{\max} - h_{\min} + \| y(.,.,0)\|_{+\infty} \right )
\]
\end{proof}